\documentclass[preprint,numafflabel, sort&compress]{elsarticle}
\usepackage{etoolbox}
\usepackage{lipsum}
\makeatletter
\def\ps@pprintTitle{%
	\let\@oddhead\@empty
	\let\@evenhead\@empty
	\def\@oddfoot{}%
	\let\@evenfoot\@oddfoot}
\makeatother
\usepackage[top=2cm, bottom=2cm, outer=2.5cm, inner=2.5cm]{geometry}
\usepackage[english]{babel}
\usepackage{lineno,hyperref}
\usepackage{multicol}

\usepackage{mathtools}
\usepackage{bm}
\usepackage{amsmath}
\usepackage{amssymb}
\usepackage{bbm} 
\usepackage{widetext} 
\usepackage{amsthm}
\usepackage{multirow}
\newtheorem{remark}{Remark}
\usepackage{stmaryrd}
\usepackage{physics}
\usepackage{mathrsfs}

\usepackage{longtable}
\usepackage{graphicx}
\usepackage{caption}
\usepackage{array}
\usepackage{float}
\usepackage{ragged2e}
\usepackage{graphbox} 
\usepackage{booktabs} 
\usepackage{colortbl} 
\usepackage{cancel}
\usepackage{soul}
\usepackage[usernames,dvipsnames,svgnames,table]{xcolor}
\usepackage{hyperref}

\journal{...}
\modulolinenumbers[5]

\newtheorem{definition}{Definition} 

\newtheorem{theorem}{Theorem}

\newtheorem{example}{Example}

\newtheorem{proposition}{Proposition}
\newcommand{\assign}{:=}

\newcommand{\supp}{\textnormal{supp}\,} 
\newcommand{\defeq}{\doteq} 
\newcommand{\re}{\mathbb{R}} 
\newcommand{\co}{{\mathbb{C}}}  
\newcommand{\ads}{\mathbb{A}\textnormal{d}\mathbb{S}_{d+1}} 
\newcommand{\pads}{\textnormal{P}\mathbb{A}\textnormal{d}\mathbb{S}_{d+1}} 
\newcommand{\npads}{\textnormal{P}\mathbb{A}\textnormal{d}\mathbb{S}} 
\newcommand{\man}{\mathcal{M}} 
\newcommand{\ie}{\textit{i.e.} }  
\newcommand{\iee}{\textit{i.e.}, }  
\newcommand{\cf}{\textit{cf.} }   
\newcommand{\uppmink}{{\mathbb{H}}^{d+1}}    
\newcommand{\nuppmink}{{\mathbb{H}}}    
\newcommand{\mink}{{{\mathbb{M}}^{d+1}}}    
\newcommand{\nmink}{{{\mathbb{M}}}}    
\newcommand{\microca}{{\mathcal{F}_{\mu c}(\man)}}  
\newcommand{\reg}{{\mathcal{F}_{reg}(\man)}}  
\newcommand{\loc}{{\mathcal{F}_{loc}(\man)}}  
\newcommand{\tensor}{\otimes}       
\newcommand{\test}{\mathcal{D}(\man)}   
\newcommand{\pertmicroca}{{\mathcal{F}_{\mu c}}\llbracket \hbar \rrbracket}  
\newcommand{\pertreg}{{\mathcal{F}_{reg}}\llbracket \hbar \rrbracket}  
\newcommand{\pertloc}{{\mathcal{F}_{loc}}\llbracket \hbar \rrbracket}  
\newcommand{\alg}{\mathcal{A}} 
\font\testa=ec-lmr10 at 12pt 
\newcommand{\nord}[1]{\,\vcentcolon\mathrel{#1}\vcentcolon} 
\providecommand{\vcentcolon}{\mathrel{\mathop{\textnormal{\testa{:}}}}}
\newcommand{\moll}{\mathsf{r}}       
\newcommand{\hyp}{{}_2F_1}    

\begin{document}


\begin{frontmatter}
		
		

		\title{On the interplay between boundary conditions\\ and the Lorentzian Wetterich equation}
		\author{Claudio Dappiaggi\footnote{claudio.dappiaggi@unipv.it}$^{\dagger\star}$}
        \author{Filippo Nava\footnote{filippo.nava02@universitadipavia.it}$^{\dagger}$}
		\author{Luca Sinibaldi\footnote{luca.sinibaldi01@universitadipavia.it}$^{\dagger\star}$}
		\address{$^{\dagger}$Dipartimento di Fisica, Universit\`a degli Studi di Pavia, Via Bassi, 6, 27100 Pavia, Italy}
		\address{$^{\star}$Istituto Nazionale di Fisica Nucleare -- Sezione di Pavia, Via Bassi 6, 27100 Pavia, Italy\\ and \\ Istituto Nazionale di Alta Matematica, Sezione di Pavia, via Ferrata 5, 27100 Pavia, Italia }

		\begin{abstract}
		In the framework of the functional renormalization group and of the perturbative, algebraic approach to quantum field theory (pAQFT), in \cite{mother} it has been derived a Lorentian version of a flow equation à la Wetterich, which can be used to study non linear, quantum scalar field theories on a globally hyperbolic spacetime. In this work we show that the realm of validity of this result can be extended to study interacting scalar field theories on globally hyperbolic manifolds with a timelike boundary. By considering the specific examples of half Minkowski spacetime and of the Poincaré patch of Anti-de Sitter, we show that the form of the Lorentzian Wetterich equation is strongly dependent on the boundary conditions assigned to the underlying field theory. In addition, using a numerical approach, we are able to provide strong evidences that there is a qualitative and not only a quantitative difference in the associated flow and we highlight this feature by considering Dirichlet and Neumann boundary conditions on half Minkowski spacetime.
		\end{abstract}
		
		
		
\end{frontmatter}
\tableofcontents

\section{Introduction}

In the realm of the mathematical formulation of quantum field theories, the algebraic approach represents a successful and versatile framework which is particularly suited to formulate and to study models living on generic curved backgrounds, see {\it e.g.} \cite{Brunetti:2015vmh}. The successes which can be ascribed to algebraic quantum field theory are many and, among them, most notable has been the formulation of a local and covariant theory of Wick polynomials \cite{hw01,hw02} and of renormalization \cite{Brunetti:1999jn}. This has paved the way in turn to devising a mathematical framework to discuss interacting quantum field theories at a perturbative level, commonly referred to as perturbative algebraic quantum field theory (pAQFT), see \cite{bdf09,Rej16}. Among the many features of pAQFT, noteworthy is the proof in \cite{bdf09} of the existence of a connection between the Stueckelberg-Petermann renormalization group which describes the freedom in the perturbative construction and the approach à la Wilson \cite{wilson} relating theories at different energy scales.

In the spirit of Wilson approach to renormalization noteworthy is the average effective action formalism to the functional renormalization group, which describes the flow of the fluctuations from a microscopic to a macroscopic scale in terms of the Wetterich equation \cite{Ringwald:1989dz,Wetterich:1989xg,wett93}. This has been formulated for Euclidean quantum field theories and, due to its effectiveness in computing explicitly physical quantities, it has been employed especially in the asymptotic safety program, see {\it e.g.}, \cite{disp_erge}. Recently, in \cite{mother}, using the framework of pAQFT, it has been constructed a Lorentzian counterpart of the Wetterich equation, considering interacting scalar field theories on globally hyperbolic spacetimes, although an extension to gauge theories has been discussed in \cite{DAngelo:2023tis} and in \cite{DAngelo:2023ssw} local existence of the solutions has been proven using as key ingredient the Nash-Moser theorem. 

These works rest on the key assumption that the underlying background is globally hyperbolic since this entails that, regardless of the form of the interaction  considered and focusing the attention for simplicity on scalar theories, the normally hyperbolic operator ruling the linear contribution to the dynamics enjoys notable properties. Among these we highlight the existence of unique advanced and retarded fundamental solutions as well the existence of quasi-free Hadamard states, see for example \cite{Khavkine:2014mta}. The Hadamard condition is a constraint on the singular structure of the underlying two-point correlation function and, abiding by it, entails the possibility of devising an explicit, covariant scheme to construct Wick polynomials. These notable properties have led for many years to the implicit belief that only field theories on globally hyperbolic spacetimes were fit to be described in the algebraic language. Earlier works by B. Kay \cite{Kay:1992es} as well as the realization that many relevant physical models, ranging from the Casimir effect to the AdS/CFT correspondence, required a background which is not globally hyperbolic, prompted in the past decade the beginning of several researches aimed at changing this situation. From a geometric viewpoint, this has led to extend the class of backgrounds that one can consider in order to formulate quantum field theories in the algebraic language, introducing the class of globally hyperbolic spacetimes with a timelike boundary \cite{hauflosanch}. Without entering here into the technical side of the definition, we highlight that it encompasses many physically relevant examples, among which noteworthy are the half-Minkowski spacetime $\uppmink$ and, up to a conformal transformation, both the universal cover and the Poincar\'e patch of Anti-de Sitter spacetime, see Example \ref{Exam: backgrounds} in the main body of this paper.  

On top of these backgrounds, it is possible to formulate any free field theory, though the presence of a timelike boundary entails that initial value problems are ill-defined unless one supplements them with a suitable boundary condition. Among the plethora of possibilities, noteworthy for our purposes, are those leading to a well-defined quantization scheme which allows to implement the canonical commutation/anti-commutation relations. Restricting for simplicity the attention once more to scalar field theories, this novel freedom has far reaching consequences, since each admissible boundary condition corresponds a priori to a physically different system. In addition, standard results such as the existence of unique advanced/retarded fundamental solutions and of Hadamard states are no longer guaranteed in these scenarios. In particular, since the singular structure of the underlying two-point function is constrained by causal propagation, the presence of a timelike boundary entails that the standard notion of Hadamard states needs to be amended to account for the possibility that singularities can be reflected at the boundary, much in the same spirit of what occurs in geometric optics with light rays. 

Although our current understanding of these scenarios is not on par with the counterparts on globally hyperbolic spacetimes with an empty boundary, considerable progresses have been made in the past few years, both in classifying the class of admissible boundary conditions \cite{dappia_wave} and in understanding how to extend the notion of Hadamard states and of propagation of singularities \cite{asym_Gannot:2018jkg,dappia_prop}. Therefore an algebraic formulation of free scalar quantum field theories on large classes of globally hyperbolic spacetimes with a timelike boundary is nowadays available, see {\it e.g.} \cite{dappia_ads2}. This suggests therefore the possibility also of considering non linear interactions and, in this paper, we move a first step forward in this direction. In particular we consider the Lorentzian Wetterich equation as constructed in \cite{mother} and we investigate whether it is still applicable when the underlying background is a globally hyperbolic spacetime with a timelike boundary. A close inspection of \cite{mother} unveils that a generalization to all possible scenarios is a daunting task since one needs to take into account the existence both of additional freedoms in the construction and of technical hurdles. On the one hand, among the plethora of admissible boundary conditions, one can consider those of dynamical type. These entail that the restriction to the boundary of the underlying manifold $\man$ of a solution of the equations of motion is constrained to abide by a dynamics which is intrinsically defined on $\partial\man$. While, at the level of free field theories, this leeway can be accounted for, see {\it e.g.} \cite{Dappiaggi:2022dwo,Zahn:2015due}, at the level of interacting field theories, it opens the possibility of considering underlying Lagrangians with non linear contributions which are defined intrinsically at the boundary. In addition to this technical complication, another hurdle is related to the notion of Hadamard states. On a globally hyperbolic spacetime, thanks to the seminal papers of Radzikowski \cite{radz,Radzikowski:1996ei}, it is understood that there is an equivalence between the global form of a Hadamard state, which specifies a constraint on the singular structure of the underlying two-point correlation function and the local counterpart. This is a prescription on the form of the integral kernel of the two-point correlation function, when restricted to convex geodesic neighbourhoods, see also \cite{Khavkine:2014mta}. Without entering into more technical details, it is important to stress that such correspondence plays a prominent r\^{o}le in devising an algorithmic, covariant scheme, to construct Wick polynomials and it is a key ingredient in \cite{mother}. Alas, there is no counterpart of the results in \cite{radz,Radzikowski:1996ei} on a generic globally hyperbolic spacetime with a timelike boundary due to the difficulties in accounting for the reflection of singularities at the level of integral kernel of the underlying two-point correlation function.  

In this paper our main goal is twofold. On the one hand we wish to show that the Lorentzian Wetterich equation can be considered also in presence of timelike boundaries, while, on the other hand, we want to highlight that the resulting flow has a quantitative and qualitative behaviour which is strongly dependent on the choice of the underlying boundary condition. Therefore, in order to fulfill these goals and to bypass the hurdles mentioned above, we restrict our attention to specific examples of manifolds with a timelike boundary, namely half Minkowski spacetime and the Poincar\'e patch of Anti-de Sitter spacetime. In both these cases, not only the classification of admissible boundary conditions is known, but, if one selects those of Dirichlet and Neumann type, it is possible to give an explicit analytic expression of the two-point correlation function associated to the ground state of a real, massive scalar field in terms of the geodesic distance built out of the underlying metric. In addition a local expression for the Hadamard two-point correlation function is known and this feature enables us to translate to this setting the approximation procedures introduced in \cite{mother}.

More in detail the paper is organized as follows. In Section \ref{sec:geometry} we outline the class of backgrounds that we consider, especially in Example \ref{Exam: backgrounds}, while, in Section \ref{sec:funct_approach}, we give a succinct overview of the functional approach to algebraic quantum field theory, focusing on interacting models in Section \ref{sec:inter}. For simplicity of the presentation we focus until this point on globally hyperbolic spacetimes without boundary, while in Section \ref{Sec: QFT with timelike boundaries} we consider a real, massive free scalar field theory on half-Minkowski spacetime and on the Poincar\'e patch of Anti-de Sitter. In both cases we give the form of the two-point correlation function for the ground state both with Dirichlet and with Neumann boundary conditions. 
In Section \ref{sec:wett} we introduce the Wetterich equation following \cite{mother} and we highlight where the presence of a timelike boundary entails differences with respect to the standard scenario. Section \ref{sec:applications} is the core of the paper. In Section \ref{sec:flowhalf} we consider half Minkowski spacetime and we derive the relevant flow equations working for definiteness in four spacetime dimensions. In the process one has to perform the subtraction of the Hadamard parametrix, encoding the singular contribution of the underlying two-point correlation function. We highlight that, due to the presence of the timelike boundary, two options are available and we refer to them as minimal and full Hadamard subtraction. In the first one, we follow the same rationale employed when the underlying spacetime has no boundary and we account only for the singularities present at the coinciding point limit. In the second one, instead, we also account for the reflection of the singularities at the timelike boundary. In Equation \eqref{eq:scaling_couplconst_dimensionless} and \eqref{eq:betafunc_casim}, we prove that in both instances there is a contribution to the flow equation due to the boundary condition. Since an analytic solution is not available, we study numerically both systems using the software Mathematica. While in the full Hadamard subtraction the ensuing flow is qualitatively similar to the one on full Minkowski spacetime, in the minimal case, the numerical solution suggests that approaching the boundary does not affect the ultraviolet behaviour of the theory, as it can be still observed a Gaussian ultraviolet, fixed point, for the non-interacting theory. Yet, the infrared behaviour is significantly different, since the theory, under the renormalization group flow, seems to spontaneously acquire mass although drawing any definitive conclusion is premature. As a matter of fact, in such regime and close to the boundary, the observed feature might be a signal that the contribution due to the reflected singularities is relevant and, therefore, that one should consider the full Hadamard subtraction. In Section \ref{Sec: Neumann flow}, we repeat the analysis using Neumann boundary conditions. Here it becomes clear that this has a drastic effect in the flow equations \eqref{eq:scaling_couplconst_dimensionless_choosemu_neumann} and \eqref{eq:scaling_couplconst_dimensionless_choosemu_casimir_neumann} since the systems display a behaviour proper of asymptotically free theories, namely the coupling constants vanish at large energies. At last in Section \ref{sec:pads_scal} we study the Wetterich equation for a real interacting scalar field theory on the Poincar\'e patch of the four-dimensional Anti-de Sitter spacetime. In this scenario we consider only Dirichlet boundary conditions since those of Neumann types are admissible only for a limited range of masses. Hence, since the Wetterich equation entails a flow of this parameter, one can realize that such constraint would be easily violated in the process. The resulting flow equation is here considered only with reference to the minimal Hadamard subtraction and we find evidence of the existence of an ultraviolet fixed point. The infrared behaviour is instead very different from the one obtained on half-Minkowski spacetime. This is a by-product of the geometric effect of the negative curvature which leads to a confinement of the massless modes and, as a consequence, to an infrared behaviour which is less singular.

\section{Setting}
In this section we introduce the class of backgrounds that we consider in this work, together with their key geometric properties and we review the key definitions and results at the heart of Perturbative Algebraic Quantum Field Theory (pAQFT).

\subsection{Geometry}\label{sec:geometry}

In this paper we denote by $\man\equiv(\mathcal{M},g)$ a $(d+1)$-dimensional manifold, $d\geq 2$, with a non empty boundary $\partial\man$, see \cite{lee}, while $g$ is a smooth Lorentzian metric. The interior of $\man$ is $\mathring{\man}\doteq\man\setminus\partial\man$. In addition, following \cite{hauflosanch}, we are interested on {\em manifolds with a time-like boundary}, namely letting $\iota:\partial\man\to\man$ be the natural inclusion map, $(\partial\man,\iota^*g)$ ought to be a $d$-dimensional Lorentzian manifold. Among the plethora of these backgrounds we are especially concerned with those which are globally hyperbolic, namely they must not contain closed, causal curves and $J^+(p)\,\cap\,J^-(q)$ has to be a compact set for all $p,q\in\man$, where $J^\pm$ denote the causal future and past.

Similarly to the scenario with an empty boundary, in \cite{hauflosanch}, it has been proven that, for any globally hyperbolic manifold $(\man, g)$ with a time-like boundary there exists an isometry $\psi:\mathbb{R}\times\Sigma\to\man$ such that, denoting by $t:\mathbb{R}\times\Sigma\to\mathbb{R}$ the coordinate along the $\mathbb{R}$-direction, the line element associated to the metric $\psi^*g$ reads
\begin{equation}\label{eq: line-elements}
ds^2=-\beta dt^2+h_t,
\end{equation}
where $\beta\in C^\infty(\mathbb{R}\times\Sigma)$ is a positive function while $h_t$ is a family of smooth Riemannian metrics on $\Sigma_t\doteq\{t\}\times\Sigma$ varying smoothly with $t$. Furthermore, each $\Sigma_t$, $t\in\mathbb{R}$, is a smooth Cauchy hypersurface with a non empty boundary, that is every inextensible,
piecewise smooth, timelike curve intersects $\Sigma_t$ only once. Henceforth we shall leave the function $\psi$ and its action implicit, since no confusion can arise.
  
\begin{remark}\label{def:standardstatic}
	In order to avoid unnecessary technical hurdles, which could make some of our results more inaccessible and obscure, we focus the attention on a specific subclass of globally hyperbolic spacetimes with a timelike boundary, which we refer to as {\bf (standard) static}. With reference to Equation \eqref{eq: line-elements} these are characterized by the constraints
	$$\beta\neq\beta(t),\quad\textrm{and}\quad h_t=h_0\equiv h,$$
	that is both $\beta$ and the metric on $\Sigma_t$ are time-independent.
\end{remark}

\begin{example}\label{Exam: backgrounds}
	We introduce the spacetimes which will play a leading r\^{o}le in our analysis:
\begin{enumerate}
    \item The $(d+1)$-dimensional \textbf{\emph{Minkowski half-space}} $(\mathbb{H}^{d+1},\eta)$. This is the standard upper half space $\mathbb{H}^{d+1}$, such that, considering the standard Cartesian coordinates $(t,z,x_1,\ldots,x_{d-1})$, the associated line element reads
    $$ds^2_{\mathbb{H}^{d+1}}=-dt^2+dz^2+\sum\limits_{i=1}^{d-1}x^2_i,\;z\geq 0,\;\textrm{and}\; t,x_i\in\mathbb{R}.$$ 
    Observe that the boundary $\partial\mathbb{H}^d$ is the locus $z=0$, which is isometric to $d$-dimensional Minkowski spacetime $(\mathbb{R}^d,\eta)$. This is manifestly a standard static, globally hyperbolic spacetime with a non-empty timelike boundary.
   \item The {\bf Poincar\'e patch} $\pads$ of the $d+1$-dimensional \textbf{\emph{Anti-de Sitter spacetime}} $\ads$, which is in turn the maximally symmetric solution of Einstein's field equations with a cosmological constant $\Lambda<0$. As a manifold this is diffeomorphic to $\mathring{\mathbb{H}}^{d+1}$ and, endowing it with the standard Cartesian coordinates $t,x_i\in\mathbb{R}$, $i=1,\dots,d-1$ and $z>0$, the associated line element reads
\begin{equation}
    \label{eq:padsmetric}
    ds^2_{\npads}=\frac{l^2}{z^2}(-dt^2+dz^2+\sum\limits_{i=1}^{d-1}x^2_i),
\end{equation}
where $l^2=-\frac{d(d-1)}{2\Lambda}$. Observe that $\pads$ is not directly a globally hyperbolic spacetime with a timelike boundary on account of the singularity in the metric as $z\to 0^+$, but, by means of a conformal transformation, one can switch back to the previous example. Hence $\pads$ admits a conformal, timelike boundary. We shall exploit this feature in our investigation.
\end{enumerate}

\end{example}

\noindent On top of the backgrounds introduced in Example \ref{Exam: backgrounds} we consider a free real scalar field $\phi:\man \to \mathbb{R}$ whose dynamics is ruled by
\begin{equation}\label{EQ:KG}
    P\phi=(\Box_g - m^2-\xi R)\phi=0,
\end{equation}
where $\Box_g$ is the d'Alembert wave operator built out of the underlying metric $g$, while $m^2\geq 0$ is the mass parameter whereas $\xi \in \mathbb{R}$ denotes an arbitrary coupling to the scalar curvature $R$. On the one hand we observe that, if we consider $\pads$, by means of a conformal transformation, we can transform Equation \eqref{EQ:KG} into an equivalent wave-like equation in $\mathbb{H}^{d+1}$, though with the linear term $m^2+\xi R$ replaced by a suitable singular potential which diverges quadratically as $z\to 0^+$, see \cite{dappia_ads}. On the other hand, the presence of a timelike boundary entails that the assignment of initial data on a Cauchy surface do not suffice to construct a unique solution of Equation \eqref{EQ:KG} and they need to be supplemented with suitable boundary conditions assigned at $\partial\man$. The analysis of the class of admissible boundary conditions which entail in turn the existence of well-defined advanced and retarded fundamental solutions is a rather involved topic which has received lately a lot of attention. For an interested reader we refer to \cite{dappia_wave} for the analysis of this specific problem for the Klein-Gordon operator $P$ on a globally hyperbolic, static background with a timelike boundary, whereas the specific scenario of $\pads$ has been considered in \cite{dappia_ads, dappia_ads2, asym_Gannot:2018jkg}.

In this work, for technical reasons which will be elucidated in the following, we shall consider only very specific boundary conditions, namely those of Dirichlet type on both $\uppmink$ and $\pads$ and those of Neumann type on $\uppmink$. Therefore we shall focus only on these specific scenarios for which it is unnecessary to outline the underlying theory in its full generality, as it will become clear especially in Section \ref{Sec: QFT with timelike boundaries}.

\subsection{Functional approach to AQFT}
\label{sec:funct_approach}
In this section we overview the functional approach developed in the context of Algebraic Quantum Field Theory (AQFT). For simplicity of the presentation, we assume at the beginning that $(\man,g)$ is a globally hyperbolic spacetime with an empty boundary and we follow mainly \cite{bdf09, Rej16}. At the very end we comment on the necessary modifications when $\partial\man\neq\emptyset$ is a timelike manifold and when we consider the specific classes of boundary conditions, mentioned at the end of the previous section.

\begin{definition}[Off-shell field configurations]
    Let $(\man,g)$ be a globally hyperbolic spacetime. We call $\mathcal{E}(\man)\defeq\,C^\infty(\man)$ the space of \emph{off-shell configurations}. 
\end{definition}

\begin{remark}
	For clarity of the presentation, in the following, we will stick to the convention that a generic off-shell configuration is denoted by the symbol $\chi$ when there is no need to make a reference to a specific field theory. On the contrary, we shall adopt the symbol $\phi$ whenever we refer to a configuration which is related to an underlying real scalar field, which might be required to abide by the Klein-Gordon equation \eqref{EQ:KG}.
\end{remark}

\noindent Observe that the space of off-shell field configurations is a topological space, endowed with the Fréchet topology, \cf \cite[Section 3.1]{Rej16}. Following the functional approach, \cf \cite{bdf09}, we denote by $\mathcal{F}(\man)$ the space of complex-valued functionals $F:\mathcal{E}(\man)\rightarrow \mathbb{C}$. Among the plethora of elements lying in $\mathcal{F}(\man)$, we individuate two distinguished subclasses.

\begin{definition}[Regular and local functionals]
    \label{def:reg_funct}
    An element $F\in\mathcal{F}(\man)$ is called a {\bf regular functional} if $\exists f\in\,C^\infty_0(\man)$ such that
    $$F(\chi)\equiv X_f(\chi)\defeq\int_\man d\mu_g(x)\,f(x)\chi(x),\quad\forall\chi\in\mathcal{E}(\man)$$
    where $d\mu_g$ is the metric induced volume form. The space of regular functionals is denoted by $\mathcal{F}_{reg}(\man)$.
    Similarly, we say that $F$ is a \emph{local functional} if $\exists f\in\,C^\infty_0(\man)$ such that
    $$F(\chi)=\int_\man d\mu_g(x)\,f(x)j_x(\chi),\;\;f\in\,C^\infty_0(\man),$$
    where $j_x(\chi)$ is the \emph{jet prolongation} of $\chi$, \cf \cite{Rej16}. The space of local functionals is denoted by $\mathcal{F}_{loc}(\man)$.
    The \emph{support} of $F\in\mathcal{F}(\man)$ is
    $$\supp(F)\defeq\{x\in\man\;|\;\forall\,U\;\textnormal{neighborhood of}\;x\;\exists\,\chi,\psi\in\,\mathcal{E}(\man),\,\supp\chi\in U\;|\;F(\psi+\chi)\neq F(\psi)\}.$$
\end{definition}

\noindent Subsequently we introduce the notion of \emph{\textbf{functional derivative}} of an element lying in $\mathcal{F}(\man)$.

\begin{definition}[Functional derivative]
    \label{def:functdev}
    Let $F\in\mathcal{F}(\man)$ and let $\chi\in \mathcal{E}(\man)$, $\psi_1,...,\psi_n\in C^\infty_0(\man)$. We say that, if existent,
    $$\langle F^{(n)}(\chi),\psi_1\otimes...\otimes \psi_n\rangle\defeq \frac{d^n}{ds_1...ds_n}F(\chi+s_1\psi_1+...+s_n\psi_n)\Bigr|_{s_1 = ... = s_n = 0}$$
    is the n-th functional derivative of $F$ at $\chi$ along the directions $\psi_1,...,\psi_n$.
    A functional is $n$-times differentiable at $\chi$ if its functional derivative exists for all $\psi_1,\ldots,\psi_n\in C^\infty_0(\man)$.
\end{definition}

Products between functionals and their derivatives, or composition with specific distributions such as the propagators associated to a partial differential operator on $\man$, are \textit{a priori} ill-defined operations. In order to bypass this hurdle we make use of tools and techniques proper of microlocal analysis, which suggest  to impose suitable constraints on the singular structure of the derivatives of the functionals that we consider. An interested reader can refer to \cite{horror1} for basic definitions and results on microlocal analysis. 



\begin{definition}[Micro-causal functional]
\label{def:microca}
    Let $F\in\mathcal{F}(\man)$ be a functional on $\mathcal{E}(\man)$. We say that $F$ is a microcausal functional, {\it i.e.}, $F\in\mathcal{F}_{\mu c}(\man)$ if, denoting by $p\equiv(p_1,\dots,p_n)\in\man^n$, $n\geq 1$
    $$F^{(n)} \in \mathcal{E}'
    (\man^n), \quad\textrm{and}\quad \mathrm{WF} (F^{(n)}) \cap \Bigl[ \bigcup_{p\in\man^n} (\bar{V}_p^+ \cup
    \bar{V}_p^-)\Bigr]= \emptyset, \quad n \in \mathbb{N}$$    
    where $F^{(n)}$ is the $n$-th functional derivative of $F$ as per Definition \ref{def:functdev}, while $\mathrm{WF}$ denotes the wavefront set. In addition
    $\bar{V}_p^\pm\doteq\{p_1\}\times\bar{V}_{p_1}^\pm\times\dots\times\{p_n\}\times\bar{V}_{p_n}^\pm $ where $\bar{V}_{p_i}^+,\bar{V}_{p_i}^-$ are, respectively, the sets of future-pointing
    and past-pointing lightike and timelike covectors lying in $T^{\ast}_{p_i}\man$, $i=1,\dots,n$. The collection of all microcausal functionals is a $*$-algebra with respect to the pointwise product $\cdot$ and to the $*$-operation induced by complex conjugation:
    $$(F\cdot G)(\chi)=F(\chi)G(\chi)\quad\textrm{and}\quad F^*(\chi)=\overline{F(\chi)}.\quad\forall F,G\in\microca$$
    This algebra is denoted by $(\microca,\cdot,*)$.
\end{definition}
    
\noindent Observe that, by construction, $\mathcal{F}_{reg}(\man)\subset\mathcal{F}_{loc}(\man)\subset\mathcal{F}_{\mu c}(\man)$. 

Since in this section we consider the Klein-Gordon operator $P$ as per Equation \eqref{EQ:KG} on a globally hyperbolic Lorentzian manifold with an empty boundary, standard results in the theory of PDEs, see, {\it e.g.}, \cite[Theorems 3.3.1, 3.4.3]{baer08}, entail that $P$ admits unique advanced and retarded fundamental solutions $\Delta_{A/R}:C^\infty_0(\man)\to C^\infty(\man)$. Denoting by $\Delta:=\Delta_R-\Delta_A$ the associated causal propagator, we can follow the approach introduced in \cite{bdf09} quantizing the classical theory through a procedure known as \emph{deformation quantization}. This scheme amounts introducing a formal parameter $\hbar$ to \emph{deform} the point-wise product of the underlying algebra of observables, starting from the regular functionals.


\begin{definition}[Quantum *-algebra structure on $\reg$]
    \label{def:deformreg}
    We call $(\pertreg,\star,*)$ the \emph{quantum *-algebra of regular observables} where,  given $F,G\in \reg$ two regular functionals as per Definition \ref{def:reg_funct},
    \begin{equation}\label{eq:deformreg}
    F\star G\defeq F\cdot G+\sum_{n=1}^\infty \frac{\hbar^n}{n!}\langle F^{(n)},(\frac{i}{2}\Delta)^{\tensor n} G^{(n)}\rangle,
    \end{equation}
    where $\Delta$ stands for the causal propagator while $F^{(n)}$ and $G^{(n)}$ denote the $n$-th functional derivatives. In addition $\pertreg$ stands for a power series in $\hbar$ with coefficients in $\reg$, and
    $$(F\star G)^*\defeq G^*\star F^*.$$
\end{definition}

The singular structure of $\Delta$, see \cite{Radzikowski:1996ei}, entails that this deformation cannot be slavishly translated to the space of microcausal functionals $\microca$. To bypass this hurdle we need to introduce the renown {\bf \emph{Hadamard two-point distributions}}, see \cite{radz}, namely $\Delta_+\in\mathcal{D}^\prime(\man\times\man)$ such that 
\begin{equation}
	\label{eq:microloc_spec_cond}
	\mathrm{WF}(\Delta_+)=\{(x,k_x;y,k_y)\in\,T^*(\man\times\man)\setminus\{0\}\,|\,(x,k_x)\sim (y,-k_y),\,k_x\,\vartriangleright 0\},
\end{equation}
where $(x,k_x)\sim (y,-k_y)$ entails that there exists a lightlike geodesic $\gamma$ connecting $x$ and $y$ such that $g^{-1}k_x$ is a lightlike vector, tangent to $\gamma$ at $x$ while $-k_y$ is the parallel transport of $k_x$ along $\gamma$. The symbol $k_x\,\vartriangleright 0$ means that $k_x$ is future directed. In addition to controlling the singular structure of $\Delta_+$, in the following we also need it to identify a quantum state for the free field theory associated to the Klein-Gordon equation \eqref{EQ:KG}. The necessary constraints are specified in the following definition, see \cite{Khavkine:2014mta} for additional details:

\begin{definition}\label{Def: 2-pt quantum state}
	Let $\Delta_+\in\mathcal{D}^\prime(\man\times\man)$ be a Hadamard two-point distribution and let $P$ be the Klein-Gordon operator as in Equation \eqref{EQ:KG}. We say that $\Delta_+$ is an on-shell, positive, Hadamard two-point distribution if for all $f,f^\prime\in C^\infty_0(\man)$
	$$\Delta_+(f,f)\geq 0,\quad \Delta_+(Pf,f^\prime)=\Delta_+(f,Pf^\prime)=0.$$
	In addition we also require that
	$$\Delta_+(f,f^\prime)-\Delta_+(f^\prime,f)=i\Delta(f,f^\prime)\quad\textrm{and}\quad |\Delta(f,f^\prime)|^2\leq 4\Delta_+(f,f)\Delta_+(f^\prime,f^\prime),$$
	where $\Delta=\Delta_R-\Delta_A$ is the causal propagator built out of the advanced and retarded fundamental solutions associated to $P$.
\end{definition}

\begin{definition}[Quantum *-algebra structure on $\microca$]
    \label{def:deform_microca}
    Given a Hadamard two-point distribution abiding by the constraints of Definition \ref{Def: 2-pt quantum state} we call $\mathcal{A}_{\Delta_+}(\man)\equiv(\pertmicroca,\star_{\Delta_+},*)$ the \emph{quantum *-algebra of observables} induced by $\Delta_+$ where, for all $F,G\in\microca$,
    \begin{equation}\label{eq:deformmicroca}
    F\star_{\Delta_+} G\defeq F\cdot G+\sum_{n=1}^\infty \frac{\hbar^n}{n!}\langle F^{(n)},\Delta_+^{\tensor n} G^{(n)}\rangle.
    \end{equation}
    Here $\pertmicroca$ stands for a formal power series in $\hbar$ with coefficients in $\microca$, and
    $$(F\star_{\Delta_+} G)^*\defeq G^*\star_{\Delta_+} F^*.$$
\end{definition}

Observe that one can infer that Equation \eqref{eq:deformmicroca} is well-defined combining Definition \ref{def:microca} with Equation \eqref{eq:microloc_spec_cond} by means of \cite[Thm 8.2.12]{horror1}. To conclude this succinct overview of the key ingredients at the heart of pAQFT we need to discuss how to define in this framework normal ordering which is the building block of perturbation theory for non-linear quantum fields.
\begin{definition}[Normal ordering]
    \label{def:normal_ord}
    Let $X_f$ be a regular, linear observable as per Definition \ref{def:reg_funct}, and let $\Delta_+\in\mathcal{D}^\prime(\man\times\man)$ be as per Definition \ref{Def: 2-pt quantum state}. We call {\bf normal ordered products} or {\bf Wick products} the recursively defined functionals
    \begin{equation}
        \nonumber
        \begin{gathered}
            \begin{aligned}
                &\nord{X_f}_{\Delta_+}\defeq X_f,\\
                &\nord{X_{f_1}\star\ldots \star X_{f_{n+1}}}_{\Delta_+}\defeq \nord{X_{f_1}\star\ldots \star X_{f_{n}}}_{\Delta_+}\star X_{f_{n+1}}+\\
                &-\sum_{i=1}^n \nord{X_{f_1}\star\ldots \star \widetilde{X_{f_i}}\star\ldots \star X_{f_{n}}}_{\Delta_+} \hbar \Delta_+(f_i,f_{n+1}),\quad\quad f_1,\ldots,f_{n+1}\in C^\infty_0(\man),
            \end{aligned}
        \end{gathered}
    \end{equation}
    where $\star$ is the product associated to $\Delta_+$ while $\widetilde{\quad}$ means that the tilded element is left out from the sum. In particular we define the \emph{Wick power of order $n$} as the smearing of $\nord{X_{f_1}\star\ldots \star X_{f_{n}}}_{\Delta_+}$ against the distribution $f(x_1)\delta(x_1,\ldots,x_n)$. 
\end{definition}

\begin{remark}\label{Rem: Covariant Wick Polynomials} The construction as per Definition \ref{def:deform_microca}, being subordinated to the choice of an arbitrary two-point function, does not lead to \emph{local} and \emph{covariant} normal ordered quantum fields. This shortcoming has been already fixed in \cite{hw01} replacing $\Delta_+$ in Definition \ref{def:normal_ord} and also in Equation \eqref{eq:deformmicroca} with the Hadamard parametrix $H$, see \cite{radz,Radzikowski:1996ei}. More precisely $H\in\mathcal{D}^\prime(\man\times\man)$ is such that $\mathrm{WF}(\Delta_+-H)=\emptyset$ and, for any geodesic, convex neighbourhood $\mathcal{O}\subseteq\man$ and for every $x,y\in\mathcal{O}$ the integral kernel of $H$ reads
\begin{equation}
	\label{eq:had_parametrix}
	H(x,y)=\lim_{\epsilon\rightarrow 0^+}u(x,y) \,\sigma_\epsilon^{\frac{1-d}{2}}(x,y)+v_d(x,y)\log\left(\frac{\sigma_\epsilon(x,y)}{\mu^2}\right).
\end{equation}
Here $\mu\in\mathbb{R}$ is a reference, arbitrary scale length while $u,v_d\in C^\infty(\mathcal{O}\times\mathcal{O})$ are completely determined in terms of recursion relations by the geometry of the underlying background and by the Klein-Gordon operator $P$ as per Equation \eqref{EQ:KG}, see \cite[Appendix B]{Moretti:2001qh}. Observe that we have made explicit the dependence of $v_d$ on the spacetime dimension since $v_d=0$ whenever $d$ is even. The last ingredient is 
$$\sigma_\epsilon(x,y)=\sigma(x,y)+i\epsilon(\tau(x)-\tau(y)),$$
where $\sigma$ is the Synge world function, namely the halved squared geodesic distance on $(\man,g)$, whereas $\tau:\man\to\mathbb{R}$ is a time function.
\end{remark}

\subsection{Interacting Theory}
\label{sec:inter}
We switch to describing a quantum, \emph{interacting} field theory. We do not work in full generality and we focus instead on a specific example ruled by the action $I\in\mathcal{F}_{loc}(\man)$ 
\begin{equation}
    \label{eq:interact_action}
    I(\phi)\defeq I_0(\phi)+\lambda V(\phi)=-\int_\man \left(\frac{1}{2}g^{ab}\partial_a\phi\partial_b\phi +\frac{\xi}{2}R\phi^2+\frac{m^2}{2}\phi^2+\lambda\frac{\phi^n}{n!}\right)f\;d\mu_g,
\end{equation}
where $n\geq 2$, $f\in C^\infty_0(\man)$, $V(\phi)$ is the interaction term, here chosen to be of polynomial form, while $\lambda\in\re$ is the coupling constant. Observe that, in the applications, we shall fix $\dim\man=4$ and $n=4$. This action leads to the following equation of motion
\begin{equation}
    \label{eq:int_KG}
    (P+\lambda V^{(1)})\phi=0\;\Rightarrow \;\left(\Box_g-\xi R-m^2\right)\phi=\frac{\lambda}{(n-1)!}\phi^{n-1},
\end{equation}
which we recognize as the Klein-Gordon equation coupled with a non linear contribution due to $\lambda V^{(1)}(\phi)$. In the perturbative approach to AQFT the \emph{interacting *-algebra of observables} is constructed with respect to the interaction term $V$. In other words we wish to implement \emph{\textbf{Dyson's formula}} which defines the \emph{\textbf{formal S-matrix}}
\begin{equation}
    \label{eq:formal_Smatrix_introduction}
    S(V)\overset{formal}{=} \mathbb{I}+\sum_{n=1}^\infty \frac{i^n\lambda^n}{n!}\underset{n\; \textnormal{times}}{[\underbrace{\nord{V}\,\cdot_T\ldots \cdot_T\nord{V}}]},
\end{equation}
where $\cdot_T$ is the \emph{time-ordered product}. To construct the time-ordering it is customary to define it for regular functionals, successively extending it to local observables.

Following \cite[Section 6.2.3]{Rej16}, we define the time-ordered product on $\reg$ as
\begin{equation}
    \label{eq:time_ord_regular}
    F\cdot_T G \defeq F\cdot G+\sum_{n=1}^\infty \frac{\hbar^n}{n!}\langle F^{(n)},\frac{i}{2}\Delta_F^{\tensor n} G^{(n)}\rangle,
\end{equation}
where $\Delta_F\doteq\Delta_++i\Delta_A$, where $\Delta_+$ is a Hadamard two-point distribution abiding by the constraints of Definition \ref{Def: 2-pt quantum state} while $\Delta_A$ is the advanced fundamental solution associated to the Klein-Gordon operator as per Equation \eqref{EQ:KG}. 

\begin{remark}\label{Rem: local and covariant time ordering}
	Observe that, similarly to Definition \ref{def:normal_ord} and to Remark \ref{Rem: Covariant Wick Polynomials}, since $\Delta_F$ is built out of $\Delta_+$, as it stands Equation \eqref{eq:time_ord_regular} is not a local and covariant prescription. In order to by-pass this hurdle, it suffices to observe that $\Delta_F=\Delta_{+,S}+\frac{i}{2}(\Delta_R+\Delta_A)$, where $\Delta_R$ is the retarded fundamental solution associated to the Klein-Gordon operator as per Equation \eqref{EQ:KG}. Hence, it is possible to consistently use in place of $\Delta_F$, the bi-distribution $i\Delta_{\mathcal{D}}$ where 
	$\Delta_{\mathcal{D}}\defeq \frac{1}{2}\left(\Delta_A+\Delta_R\right)$ is the \emph{Dirac propagator}. This implements time-ordering in a local and covariant form.
\end{remark}

\noindent Equation \eqref{eq:time_ord_regular} abides by the \emph{causal factorization axiom}, {\it i.e.},
\begin{equation}
\label{eq:time_ord_regular_causfact}
F\cdot_T G=
    \begin{cases}
        F\star G\quad\textnormal{if}\;\supp G \cap J^+(\supp F)=\emptyset, \\
        G\star F\quad\textnormal{if}\;\supp F \cap J^+(\supp G)=\emptyset.
    \end{cases}
\end{equation}

To extend this construction to local functionals we need to further elaborate on Equation \eqref{eq:time_ord_regular}. Yet, in order to keep the length of this paper at bay, we shall not discuss the whole framework in detail, leaving an interested reader to \cite{Rej16} for a thorough analysis. We shall content ourselves with a succinct overview. 

The starting point is to introduce a {\em time ordering map} $\mathcal{T}$ acting on the space of multi-local functionals, which are tensor products of elements lying in $\mathcal{F}_{loc}(\man)$. More precisely $\mathcal{T}$ is constructed out of a family of multi-linear maps
\[ \mathcal{T}_n : \mathcal{F}_{loc}^{\otimes n}(\man) \rightarrow
\mathcal{F}_{\mu c}(\man), \]
such that $\mathcal{T}_0 = 1$ and $\mathcal{T}_1 = id$. The link between $\mathcal{T}$ and $\mathcal{T}_n$ is codified by the identity
\begin{equation}\label{eq: time-ordered-product}
	F_1\cdot_T\dots\cdot_T F_n\equiv\mathcal{T} \left(\prod_{i=1}^{n} F_i\right)\doteq \mathcal{T}_n \left( \bigotimes_{j = 1}^n F_j
	\right).
\end{equation} 
In addition, one requires the maps $\mathcal{T}_n$ to be such that $\mathcal{T}$ is symmetric and it abides by the natural extension of Equation \eqref{eq:time_ord_regular_causfact}, namely, given $\{F_i \}_{i=1,\ldots, n}, \{G_j \}_{j=1,\ldots, m} \subset\loc$ two arbitrary families of local functionals such that $F_i \gtrsim G_j$ for any $i, j$
\begin{equation}\label{Eq:factoriz-time-ordering}
	\mathcal{T} \left( \bigotimes_i F_i  \bigotimes_j G_j \right) =\mathcal{T}
	\left( \bigotimes_i F_i \right) \star \mathcal{T} \left( \bigotimes_j G_j
	\right), 
\end{equation}
where $\star$ is defined in Equation \eqref{def:deformreg}. Here the symbol $\gtrsim$ entails that $\mathrm{supp} (F_i) \cap J^- (\mathrm{supp} (G_j)) = \emptyset$, where the support of a functional is as per Definition \ref{def:reg_funct}. It is worth noticing that the conditions at the heart of Equation \eqref{Eq:factoriz-time-ordering} are not always met, {\it e.g.}, when we consider functionals with overlapping supports. In this case the definition of $\mathcal{T}$ needs to be extended and this can be done following the Epstein-Glaser inductive procedure \cite{hw02}, to which we abide by. 
As a consequence, for these classes of local functionals the causal factorization property suffices to fully determine the map $\mathcal{T}$. As discussed in \cite{Rej16}, if we work with the algebra $(\mathcal{F}_{\mu c}\llbracket\hbar\rrbracket,\star_{\Delta_+},\ast)$, an explicit realization of the time-ordering map can be obtained as a natural generalization of Equation \eqref{eq:time_ord_regular}, namely
\begin{equation}\label{Eq: time-ordered product}
	\mathcal{T}^{\Delta_F}  (F_1 \otimes \ldots \otimes F_n)\assign
	F_1 \star_{\Delta_F} \ldots \star_{\Delta_F} F_n \assign
	\mathbbmss{M} \circ e^{\sum_{\ell < j}^n D^{\ell j}_{\Delta_F}}  (F_1
	\otimes \ldots \otimes F_n), 
\end{equation}
where $F_i\in\mathcal{F}_{\mu c}\llbracket\hbar\rrbracket$, for all $i=1,\dots,n$ while 
\[ D^{\ell j}_{\Delta_F}\assign \left\langle \Delta_F, \frac{\delta^2}{\delta \chi_{\ell}
	\delta \chi_j} \right\rangle=\int\limits_{\man\times\man}d\mu_g(x)d\mu_g(y)\Delta_F(x,y)\frac{\delta}{\delta\chi(x)}\otimes\frac{\delta}{\delta\chi(y)},\]
which is manifestly symmetric under exchange of $j$ and $\ell$. The map $\mathbbmss{M}$ implements the pointwise product between two microcausal functionals as the pullback on $\microca\otimes\microca$ via the diagonal map $\iota:\mathcal{E}(M)\to\mathcal{E}(M)\times\mathcal{E}(M)$, $\chi\mapsto\iota(\chi)=(\chi,\chi)$, {\it i.e.}, for all $F,G\in\microca$ and for all $\chi\in\mathcal{E}(M)$
$$\mathbbmss{M}(F\otimes G)(\chi)\equiv (F\cdot G)(\chi)\doteq (F\otimes G)(\iota(\chi)).$$
As mentioned in Remark \ref{Rem: local and covariant time ordering}, one can replace $\Delta_F$ with the Dirac propagator $\Delta_{\mathcal{D}}$ in order to obtain a local and covariant prescription.

\subsubsection{S-Matrix}
\label{sec:smatrix}

On account of the identification of a time-ordered product in the previous section, particularly Equation \eqref{eq: time-ordered-product}, we can promote Equation \eqref{eq:formal_Smatrix_introduction} to a definition of the \emph{\textbf{local S-matrix}}
\begin{equation}
    \label{eq:local_Smatrix}
    S(\lambda V)\defeq \sum_{n=0}^\infty \frac{i^n\lambda^n}{\hbar^n n!}\underset{n\; \textnormal{times}}{[\underbrace{\nord{V}\,\cdot_T\ldots \cdot_T\nord{V}}]}=e_{\cdot_T}^{i/\hbar\,\lambda TV},
\end{equation}
which has to be interpreted as an asymptotic series and therefore it should be read as an element of $\alg_{loc}\llbracket \hbar,\lambda\rrbracket$.

Observe that this definition of local S-matrix differs from the notion of \emph{global} S-matrix of standard Quantum Field Theory on a flat spacetime. As a matter of fact the cutoff function $f\in\test$ in Equation \eqref{eq:interact_action} acts by restricting the support of $S(V)$ to a compact region of $\man$. This implements locality and it guarantees that $S(V)$ is not affected by any \emph{infrared divergence}.

Only when the adiabatic limit $f\rightarrow 1$ is taken, if $S(V)$ converges to a well-defined functional, we could make sense of \emph{incoming} and \emph{outgoing} fields. At a more profound level, on curved spacetimes it is not clear how to coherently define an asymptotic region and therefore we continue our analysis restricting the attention to the construction of \emph{local} quantities, that appear more natural in the context of our interest.
  
For an interaction term $V\in\loc$ such as the one in Equation \eqref{eq:interact_action}, we define the \textbf{\emph{relative S-matrix}} $S_V$ as follows: for all $F\in\loc$, 
\begin{equation}
    \label{eq:relativeSmatrix}
    S_V(F)\defeq S(\lambda V)^{-1}\star S(\lambda V+F),
\end{equation}
where the inverse $S(\lambda V)^{-1}$ is with respect to the product $\star$ introduced in Equation \ref{def:deformreg}.  Equation \eqref{eq:relativeSmatrix} works at the level of the abstract algebra of observables $\alg_{loc}(\man)$. Hence, if one is interested in a realization of such equation for a specific *-representation $(\pertloc,\star_{\Delta_+},*)$ a choice of $\Delta_+$ as per Definition \ref{Def: 2-pt quantum state} has to be made.

The relative S-matrix is the generator of the \emph{interacting fields}, in the sense that we can define the \emph{\textbf{Bogoliubov map}}
\begin{equation}
    \label{eq:bogoliub_map}
    R_V(F)\defeq -i\hbar \frac{d}{ds}\Bigr|_{s=0} S_V(s \mathcal{T}^{-1}[F])=S(\lambda V)^{-1}\star [S(\lambda V)\cdot_T F],
\end{equation}
where $\mathcal{T}^{-1}$ is the inverse of the time-ordering operator. In the following we will consider the interaction constant $\lambda$ as being absorbed in $V$.

\begin{definition}[*-algebra of interacting observables]
    \label{def:alg_interacting}
    Let $F,G\in\loc$ be local observables, as per Definition \ref{def:reg_funct}. We call $(\mathcal{F}_{loc}\llbracket \hbar,\lambda\rrbracket,\star_V,*)$ the \emph{interacting *-algebra of observables}, where
    $$F\star_V G\defeq R_V^{-1}(R_V(F)\star R_V(G)).$$
    We call $F_I\defeq R_V(F),\,G_I\defeq R_V(G)$ the interacting observables associated respectively to $F$ and to $G$.
\end{definition}

At last we may discuss the classical limit of Equation \eqref{eq:bogoliub_map}. We expect for $\hbar\rightarrow 0$ to recover the classical field, solution of the classical interacting equations of motion \eqref{eq:int_KG} but, as it stands, it seems that this limit cannot be performed as the local S-matrix in Equation \eqref{eq:local_Smatrix} is a Laurent series in $\hbar$. Nevertheless it has been proven in \cite[Proposition 2]{hbarlimit} that $R_V(F)$ contains no negative powers in $\hbar$ for all $V,F\in\loc$. 
\begin{definition}[Classical (retarded) M{\o}ller map]
    \label{def:mollermap}
    Let $I_1,\,I_2\in \loc$ be two local action functionals and let us denote by $\mathcal{E}_{I_1}(\man),\,\mathcal{E}_{I_2}(\man)\subset\mathcal{E}(\man)$ the collection of smooth solutions of the associated dynamics. Let $I_0\in\loc$ be a third local action functionals whose associated Euler-Lagrange equations are ruled by a second-order normally hyperbolic differential operator $P_0$. We call \emph{classical M{\o}ller map} $\mathsf{r}_{1,2}:\mathcal{E}_1(\man)\rightarrow \mathcal{E}_2(\man)$ such that
    \begin{enumerate}
        \item $\mathsf{r}_{1,2}\circ\mathsf{r}_{2,3}=\mathsf{r}_{1,3}$ for $I_1,\,I_2,\,I_3\in\loc$;
        \item $\displaystyle \frac{d}{d\lambda}\Bigr|_{\lambda=0}\mathsf{r}_{I_0+\lambda V,I_0}(\chi)=-\Delta_{0,R}\left(V^{(1)}(\chi)\right)$ for $V\in\loc$ and $\chi\in\mathcal{E}(\man)$,
    \end{enumerate}
    where $\Delta_{0,R}$ is the retarded fundamental solution associated to $P_0$, see \cite{baer08}.
\end{definition}
\noindent The same definition could be expressed in terms of the advanced propagator $\Delta_{0,A}$, obtaining the classical advanced M{\o}ller map. As suggested in item 2. in Definition \ref{def:mollermap}, the classical M{\o}ller map allows to compare the dynamics induced by the free action $I_0$ and the one perturbatively induced by the interaction term $V$. We define the action of the M{\o}ller map on a functional $F\in\loc$ via its pull-back, \textit{i.e.}, for $\chi\in\mathcal{E}(\man)$
$$(\moll^*_{I_0+\lambda V,I_0}F)[\chi]\defeq F\circ \moll_{I_0+\lambda V,I_0} (\chi).$$

\subsection{Quantum field theory in presence of time-like boundaries}\label{Sec: QFT with timelike boundaries}
In this section, we address the problem of extending the framework highlighted in Section \ref{sec:funct_approach} and \ref{sec:inter} to an interacting real scalar field living on a globally hyperbolic spacetime with a timelike boundary, though restricting the attention to Example \ref{Exam: backgrounds}. A direct inspection of all the definitions and results highlights that the presence of a non empty boundary is not relevant unless one is using analytic structures which are sensible to the necessity of assigning a boundary condition at $\partial\man$. Consequently there are three specific key points of attention: the identification of advanced and retarded fundamental solutions, the construction of a Hadamard two-point correlation function and the definition of a time-ordered product. A full fledged, covariant extension of these structures is at the moment beyond our grasp due to specific technical hurdles which we shall highlight in the following. Yet, if we restrict the attention to the specific backgrounds in Example \ref{Exam: backgrounds}, a detailed analysis is possible thanks to the presence of a large isometry group which allows to derive explicit analytic formulae. 

We discuss both cases separately, addressing first of all the identification of advanced and retarded fundamental solutions and the construction of a Hadamard two-point correlation function abiding by Definition \ref{Def: 2-pt quantum state}.

\paragraph{Scalar quantum field theory on $\pads$ --} On top of $\pads$ we consider a real scalar field $\phi_0:\pads\to\mathbb{R}$ solution of the Klein-Gordon equation 
\begin{equation}\label{eq:KG-AdS}
P\phi_0=(\Box_g-m^2-\xi R)\phi_0=0.
\end{equation}
The associated classical and quantum theory have been thoroughly studied in the literature since the seminal work in \cite{Avis:1977yn} and recently a particularly emphasis has been given to the analysis of the boundary conditions which can be assigned on $\partial\pads$ in order to allow for a consistent implementation of the canonical commutation relation, see for example \cite{dappia_ads,dappia_ads2}. In view of the extensive literature available, we limit ourselves at reporting the key results which are necessary for our analysis, pointing a reader to these references for additional details.

Especially noteworthy are the boundary conditions of Robin type which have been considered in \cite{dappia_ads}, among which we highlight the Dirichlet and Neumann ones since one can associate to them a two-point correlation function $\Delta_+\in\mathcal{D}^\prime(\pads\times\pads)$ which is a weak bi-solution of the Klein-Gordon equation \eqref{eq:KG-AdS} and whose associated integral kernel has an explicit analytic expression in terms of special functions. More precisely, exploiting that Anti-de Sitter spacetimes are a maximally symmetric solution of Einstein's equations with a negative cosmological constant, in \cite{allenjac} it has been shown that
\begin{equation}
	\label{eq:twopointfunct_pads_2sol}
	\begin{gathered}
		\Delta^{\pads}_{+,D}(x,y)=\lim_{\epsilon\rightarrow 0^+} u_\epsilon^{-\frac{d}{2}-\nu}(x,y) \frac{\hyp(\frac{d}{2}+\nu,\frac{1}{2}+\nu,1+2\nu,u_\epsilon^{-1}(x,y))}{\Gamma(1+2\nu)}, \\
		\Delta^{\pads}_{+,N}(x,y)=\lim_{\epsilon\rightarrow 0^+} u_\epsilon^{-\frac{d}{2}+\nu}(x,y) \frac{\hyp(\frac{d}{2}-\nu,\frac{1}{2}-\nu,1-2\nu,u_\epsilon^{-1}(x,y))}{\Gamma(1-2\nu)},
	\end{gathered}
\end{equation}
where the superscripts $D$ and $N$ refer to the Dirichlet and to the Neumann case. In addition $\hyp$ is the \emph{hypergeometric function} with branch cut for the fourth entry along $[1,+\infty)$ on the real axis, \cf \cite[\S 15]{nist_bessel} while $\Gamma$ is the \emph{Euler gamma function} and $\nu=\frac{1}{2}\sqrt{1+4l^2\overline{m}^2}$ with $\overline{m}^2=m^2+(\xi-\frac{d-1}{4d})R$. Observe that the value of $\nu$ ought to be positive so to abide by the Breitenlohner-Freedman bound \cite{Breitenlohner:1982jf} and that the Neumann boundary condition is admissible only if $\nu<1$, see \cite{dappia_ads}. The last ingredient in Equation \eqref{eq:twopointfunct_pads_2sol} is
\begin{equation}
	\label{eq:defu}
	u(x,y)\defeq \cosh^2\left(\frac{\sqrt{2\sigma(x,y)}}{2l}\right),
\end{equation}
where $\sigma(x,y)$ is the Synge's world function. Following \cite{allenjac}, it holds that
$$u(x,y)=1+\frac{\sigma_e(x,y)}{2l^2},$$
where $\sigma_e(x,y)$ is the chordal distance that on $(\pads,g)$, in Cartesian coordinates as per Example \ref{Exam: backgrounds}, can be expressed as
\begin{equation}
    \label{eq:cordaldist_pads}
    \sigma_e(x,y)=\frac{l^2}{2zz'}[-(t-t')^2+(z-z')^2+(x_i-y_i)(x_j-y_j)\delta^{ij}],\quad i,j=1,\ldots,d-1,
\end{equation}
for $x,y\in\pads$, where the primed components are associated to $y$. Consequently 
$$u_\epsilon(x,y)\defeq u(\sigma(x,y))+i\epsilon (t-t'),$$
is the regularization of $u$ such that, for $u(\sigma)\leq 1$, the fourth entry of $\hyp$ lies above/below the branch cut depending on the sign of $t-t^\prime$.  

In sharp contrast to globally hyperbolic spacetimes with empty boundary, the notion of a physically acceptable two-point correlation function needs to be modified since Equation \eqref{eq:microloc_spec_cond} does not account for the existence of a timelike boundary, which entails, in turn, that the propagation of singularities at the hearth of a Hadamard two-point distribution is modified by the fact that lightlike geodesics are reflected at $\partial\man$. This issue has been thoroughly investigated in the past few years, see \cite{dappia_ads2,asym_hadstatesdappia,asym_Gannot:2018jkg} and for the case in hand we report the following result whose proof can be found in \cite[Theorem 4.19]{dappia_ads2}.

\begin{theorem}[Wave-front set of $\Delta_{+,D/N}^{\pads}$]
	\label{teo:wfset_2point_pads}
	Let $\Delta_{+,D/N}^{\pads}\in\mathcal{D}'(\pads\times\pads)$ be as per Equation \eqref{eq:twopointfunct_pads_2sol}. Then its wave-front set reads
	\begin{equation*}
		\mathrm{WF}(\Delta_{+,D/N}^{\pads})=\{(x,k_x;y,k_y)\in T^*(\pads\times\pads)\setminus \{0\}\,|\,(x,k_x)\sim_\pm (y,-k_y),\;k_x\vartriangleright 0\},
	\end{equation*}
	where $k_x\,\vartriangleright 0$ if $k_x$ is future directed, while $(x,k_x)\sim_\pm (y,-k_y)$ if there exist $\gamma,\gamma^{(-)}:[0,1]\rightarrow \ads$ two light-like geodesics such that, with reference to the Cartesian coordinates in $\pads$, one of the following conditions holds true:
	\begin{enumerate}
		\item $\gamma(0)=x=(t,z,\mathbf{x})$, $\gamma(1)=y=(t',z',\mathbf{y})$, $k_x=(\omega_x,k_{x,z},\mathbf{k}_x)$ is coparallel to $\gamma$ at $0$ and $-k_y$ is the parallel transport of $k_x$ along $\gamma$;
		\item $\gamma^{(-)}(0)=\iota_z x=(t,-z,\mathbf{x})$, $\gamma^{(-)}(1)=y=(t',z',\mathbf{y})$, $\iota_z k_x=(\omega_x,-k_{x,z},\mathbf{k}_x)$ is coparallel to $\gamma^{(-)}$ at $0$ and $-k_y$ is the parallel transport of $\iota_z k_x$ along $\gamma$.
	\end{enumerate}
\end{theorem}

\noindent A further advantage of considering Equation \eqref{eq:twopointfunct_pads_2sol} is that, as shown in \cite{dappia_ads}, these are two-point correlation functions of a ground state for the quantum theory associated to Equation \eqref{eq:KG-AdS} and, therefore, it is possible to reconstruct directly from them the advanced and retarded fundamental solutions associated to the Dirichlet and Neumann boundary conditions as well as the associated causal propagator, namely, working at the level of integral kernels:
\begin{gather}\label{Eq: propagators in PAdS}
	\Delta_{D/N}^{\pads}(x,y)=\Delta_{+,D/N}^{\pads}(x,y)-\Delta_{+,D/N}^{\pads}(y,x)\\
	\Delta_{R,D/N}^{\pads}(x,y)=\Theta(t-t^\prime)\Delta_{+,D/N}^{\pads}(x,y)\quad\textrm{and}\quad\Delta_{A,D/N}^{\pads}(x,y)=-\Theta(t^\prime-t)\Delta_{+,D/N}^{\pads}(x,y)
\end{gather} 

\begin{remark}
In view of these data, on $\pads$ and considering Dirichlet or Neumann boundary conditions, it is possible to extend slavishly the notion of regular, local and micro-causal functionals as per Definition \ref{def:reg_funct} and \ref{def:microca}, whereas the algebraic structure codified by the product as per Equation \eqref{eq:deformmicroca} has the same formal expression with $\Delta_+$ replaced either by $\Delta_{+,D}^{\pads}$ or by $\Delta_{+,N}^{\pads}$ as per Equation \eqref{eq:twopointfunct_pads_2sol}. As we shall see in the following section, the constraint to have $\nu$ smaller than $1$ to consider Neumann boundary conditions will prevent us from considering this case when working with the Lorentzian Wetterich equation. Hence, in the following we shall not consider this case further. 
\end{remark}

\begin{remark}
More subtle is the generalization of the notion of Wick polynomials in the case of a globally hyperbolic spacetime with a timelike boundary. On the one hand Definition \ref{def:normal_ord} can be translated to this setting mutatis mutandis, while, on the other hand, it is not obvious that it is possible to give a local and covariant definition of Wick polynomials since the results of Radzikowski \cite{radz,Radzikowski:1996ei} do not apply to this setting. Therefore it is unclear whether a counterpart of the Hadamard parametrix exists for the two-point functions individuated in Theorem \ref{teo:wfset_2point_pads}. While, on a generic globally hyperbolic spacetime with a timelike boundary, this question is still open, in the case of $\pads$, we can exploit the explicit analytic expressions in Equation \eqref{eq:twopointfunct_pads_2sol}. As shown in \cite{dappia_ads}, it holds that, focusing on $\Delta_{+,D}^{\npads_4}$ in view of the applications in Section \ref{sec:applications}, if we consider a geodesic, convex neighbourhood $\mathcal{O}\subset\npads_4$, we can write for every $x,y\in\mathcal{O}$
$$\Delta_{+,D}^{\npads_4}=H^{\npads_4}_D(x,y)+W(x,y),$$
where $W$ is a smooth function while, 
\begin{equation}
	\label{eq:hadpar_nonrefl_pads_D}
H^{\npads_4}_D(x,y)=H^{\npads_4}(x,y)-H^{\npads_4}(\iota_z(x),y).
\end{equation}
Working with Cartesian coordinates on $\pads$ such that $x\equiv(t,z,x_1,\dots x_{d-1})$, $\iota_z(x)=(t,-z,x_1,\dots,x_{d-1})$ and
\begin{equation}
    \label{eq:hadpar_nonrefl_pads}
    \begin{aligned}
         H^{\npads_4}(x,y)=\mathcal{N}\lim_{\epsilon\rightarrow 0^+} &\frac{1}{\Gamma(3/2+\nu)\Gamma(1/2+\nu)}\frac{2l^2}{\sigma_\epsilon(x,y)}+\frac{1}{\,\Gamma(1/2+\nu)\Gamma(-1/2+\nu)}\log \left(\frac{\sigma_\epsilon(x,y)}{2l^2}\right)+\\+&\left[\frac{2\gamma-1+\psi(3/2+\nu)+\psi(1/2+\nu)}{\Gamma(1/2+\nu)\Gamma(-1/2+\nu)}-\frac{1/2+\nu}{\Gamma(3/2+\nu)\Gamma(1/2+\nu)}\right],
    \end{aligned}
\end{equation}
where $\sigma_\epsilon(x,y)=\sigma(x,y)+i\epsilon (t-t')$, $\gamma$ is the \emph{Euler gamma constant}, while $\psi(x)\defeq \Gamma'(x)/\Gamma(x)$ is the \emph{digamma function}. In addition 
$$\mathcal{N}=\frac{\Gamma(3/2+\nu)\Gamma(1/2+\nu)}{2l^2}.$$
Observe that $H(\iota_z(x),y)$ is a by-product of the Dirichlet boundary condition and it is a smooth contribution whenever we consider the convex geodesic neighbourhood such that $\mathcal{O}\cap\partial\man=\emptyset$. This is a translation at the level of the local form of a Hadamard two-point distribution that on $\pads$, due to Theorem \ref{teo:wfset_2point_pads}, singularities are reflected at $\partial\man$.
\end{remark}

\paragraph{Scalar quantum field theory on $(\uppmink,\eta)$ --} On top of $\uppmink$ we consider a real scalar field $\widetilde{\phi}_0:\uppmink\to\mathbb{R}$ solution of the Klein-Gordon equation $(m\geq 0)$
\begin{equation}\label{eq:KG-Mink}
	P\widetilde{\phi}_0=(\Box_g-m^2)\widetilde{\phi}_0=0.
\end{equation}
The associated classical and quantum theory have been thoroughly studied in the literature and a recent investigation of the interplay between boundary conditions and the existence of on-shell Hadamard two-point correlation functions, especially in connection with the Casimir effect, can be found in \cite{dappia_casimir}. For our purposes we are only interested in considering Dirichlet and Neumann boundary conditions. In these two cases the underlying two-point correlation function can be constructed directly from the Poincar\'e invariant ground state on the whole Minkowski spacetime using the method of images, \textit{i.e.}, at the level of integral kernels
    \begin{equation}
    \label{eq:2pointfunct_uppmink_methodimage}
    \Delta^{\uppmink}_{+,D/N}(x,y)\defeq \Delta^{\mink}_+(x,y)\mp \Delta^{\mink}_+(\iota_z(x), y),
    \end{equation}
where, as in Equation \eqref{eq:hadpar_nonrefl_pads}, $\iota_z$ is the map reversing the sign of the $z$-component of $x\equiv(t,z,x_1,\dots,x_{d-1})$, while $\Delta_+^\mathbb{M}\in\mathcal{D}'(\mink\times\mink)$ is the 2-point correlation function of the ground state associated to the operator $P$ on the whole Minkowski spacetime $(\mink,\eta)$. Furthermore Equation \eqref{eq:2pointfunct_uppmink_methodimage} entails that 
\begin{equation}\label{Eq: WF-half-Minks}
	\mathrm{WF}(\Delta^{\uppmink}_{+,D/N})=\{(x,k_x;y,k_y)\in T^*(\uppmink\times\uppmink)\setminus \{0\}\,|\,(x,k_x)\sim_\pm (y,-k_y),\;k_x\vartriangleright 0\},
\end{equation}
where the symbols are defined as in Theorem \ref{teo:wfset_2point_pads}.

\begin{remark}
	In view of these data, similarly to the case of $\pads$ and considering Dirichlet or Neumann boundary conditions, it is possible to extend slavishly the notion of regular, local and micro-causal functionals as per Definition \ref{def:reg_funct} and \ref{def:microca}, whereas the algebraic structure codified by the product as per Equation \eqref{eq:deformmicroca} has the same formal expression with $\Delta_+$ replaced by $\Delta^{\uppmink}_{+,D/N}$ as per Equation \eqref{eq:2pointfunct_uppmink_methodimage}. In comparison to $\pads$ it is easier to translate Definition \ref{def:normal_ord} to this setting since, either by direct inspection of Equation \eqref{eq:2pointfunct_uppmink_methodimage} or by applying the method of images to $H^\nmink$, the Hadamard parametrix associated to $\Delta_+^\mathbb{M}$, the counterpart for $\Delta^{\uppmink}_{+,D/N}$ reads
		\begin{equation}
		\label{eq:hadparametrix_uppmink_methodimage}
		H^{\uppmink}_{D/N}(x,y)\defeq H^\mink(x,y)\mp H^\mink(x,\iota_z y).
	\end{equation}
\end{remark}

\begin{remark}
	In view of the preceding analysis, it seems as if the framework of pAQFT can be adapted to the scenario of a globally hyperbolic spacetime with a timelike boundary without major changes. Yet, this would be a hasty conclusion since we foresee potential hurdles to work in a general scenario. We comment on them succinctly:
	\begin{enumerate}
		\item The construction of the Lorentzian Wetterich equation shall rely chiefly on the Hadamard parametrix associated to a Hadamard two-point distribution as per Definition \ref{Def: 2-pt quantum state} and on our knowledge of its explicit form. While, in absence of a timelike boundary, this is by-product of the seminal papers of Radzikowski \cite{radz,Radzikowski:1996ei}, on a generic globally hyperbolic spacetime with a timelike boundary, this correspondence is still under investigation and the results of this paper are strongly reliant on the existence of a discrete isometry, $z\to -z$ in both $\pads$ and $\uppmink$ which allows to use the method of images.
		\item The class of boundary conditions which can be imposed at $\partial\man$ is rather vast as shown in \cite{dappia_wave}. Among the possible choices, notable are those of dynamical type such as the Wentzell boundary conditions, for which the existence of a Hadamard two-point correlation function has been proven in \cite{Dappiaggi:2022dwo}.	In these scenarios, where the restriction of the underlying field to $\partial\man$ is constrained to obey a different dynamics, there exists a consistent quantization scheme at the level of free field theories. Yet, in order to extend the whole framework of pAQFT to these scenarios, on the one hand, it would be desirable to generalize the class of action functionals as in Equation \eqref{eq:interact_action} so to allow also for additional non-linear terms intrinsically defined on $\partial\man$. On the other hand, this leeway would force us to introduce on $\partial\man$ an intrinsic time-ordered product, compatible with the one defined in $\mathring{\man}$. All these hurdles appear to require a lengthy analysis and we choose to focus on them in future work, restricting in this paper the attention to simpler scenarios which, nonetheless, allow us to make manifest the effect of $\partial\man$ and of the choice of specific boundary conditions in the solution theory of the Lorentzian Wetterich equation.
	\end{enumerate}
\end{remark}

\section{Lorentzian Wetterich Equation in the pAQFT formalism}

In this chapter we introduce the Lorentzian Wetterich equation using the algebraic approach to quantum field theory. In this endeavor we shall first follow the analysis presented in \cite{mother}, generalizing it to the case of an interacting real scalar field theory on $\uppmink$ and on $\pads$. 

\subsection{Generating Functionals}
\label{sec:genfunc}
Following \cite{mother} we present the building blocks of the functional renormalization scheme employing the formalism introduced in the past sections. In this section we consider $(\man,g)$ to be a globally hyperbolic Lorentzian manifold with empty boundary, removing only at a later stage this condition. As a starting point we define the \emph{\textbf{classical current}} $J\in\reg$, such that, for all $\chi\in\mathcal{E}(\man)$ 
\begin{equation}
    \label{eq:class_current}
    J(\chi)\defeq X_j(\chi)=\int_{\man} d\mu_g(x)\, j(x)\chi(x),
\end{equation}
where $d\mu_g$ is the volume element constructed out of the metric $g$ while $j\in C^\infty_0(\man)$. Letting $\omega$ be a quasi-free Hadamard state associated to the Klein-Gordon operator $P=\Box_g-m^2-\xi R$, see \cite{Khavkine:2014mta}, we call \textbf{\emph{generating functional}}
\begin{equation}
    \label{eq:genZ}
    Z(j)\defeq \omega(R_V[S(J)]),
\end{equation}
where $R_V$ is the Bogoliubov map defined in Equation \eqref{eq:bogoliub_map} while $S$ is the local S-matrix, see Equation \ref{eq:local_Smatrix}. In addition we call \emph{\textbf{connected generating functional}} $W(j)$, the solution to the equation
\begin{equation}
    \label{eq:genW}
    e^{W(j)}\equiv\sum_{n=0}^\infty \frac{1}{n!} W^{(n)}(j) = Z(j),
\end{equation}
where all these identities hold true at the level of integral kernels. Hence for $f_1,\ldots,f_n\in C^\infty_0(\man)$
\begin{equation*}
    \begin{gathered}
        \frac{1}{i^n}\langle W^{(n)}(j),f_1\tensor\ldots\tensor f_n\rangle|_{j=0}= -i\,\omega^C\circ R_V(X_{f_1}\cdot_T\ldots\cdot_T X_{f_n}),
    \end{gathered}
\end{equation*}
where $\omega^C$ denotes the connected component of the state $\omega$. This is defined by
$$\omega_C(\chi(x_1)\star\dots\star\chi(x_n))=(-i)^n\frac{\delta^n}{\delta f(x_1)\dots\delta f(x_n)}\left.\log\omega[\exp_\star(iX_f)]\right|_{f=0},$$
where $\exp_\star$ entails that the exponential is defined as a series with respect to the $\star$-product as in Equation \eqref{eq:deformreg} while $X_f$ is a regular functional as per Definition \ref{def:reg_funct}.

We observe that, when we consider either the Minkowski upper half-space $(\uppmink,\eta)$ or the Poincaré patch of Anti-de Sitter $(\pads,g)$, the above definitions, in particular Equations \eqref{eq:genZ} and \eqref{eq:genW} can be left unaltered, the only exception being that one needs to consider states which are of Hadamard form with respect to the definition tailored to this class of backgrounds. These are completely determined by the two-point correlation functions introduced in Section \ref{Sec: QFT with timelike boundaries}, see \cite{dappia_ads2,dappia_casimir}.

\subsection{Functional Renormalization}
\label{sec:reggenfun}
In this section we address the problem of the computation of the generating functionals introduced in Section \ref{sec:genfunc}. Following the approach customary in the study of the \emph{renormalization group flow}, we introduce an energy scale $k$ to \emph{regularize} the generating functionals. This is done according to the prescriptions of Wilsonian renormalization, see \cite{wilson}. 

The idea is to use the energy scale $k$ to obtain $Z_k(j)$ and $W_k(j)$, a regularized version of the generating functionals in Equation \eqref{eq:genZ} and \eqref{eq:genW}. The regularization should be devised in such a way that $Z_k(j)$ and $W_k(j)$ contain only the degrees of freedom with energy $E>k$ by suppressing the low energy modes. The requirements that $Z_k(j)$ must meet are
\begin{itemize}
    \item  $Z_k(j)\xrightarrow{k\rightarrow 0} Z(j)$, that is, in the limit $k\rightarrow 0$ all energy modes contribute;
    \item In the limit $k\rightarrow +\infty$ all energy modes are suppressed, hence there are no quantum contributions and the system is ruled by the classical action. This condition will become clear in the following, when we will introduce the \emph{average effective action} functional.
\end{itemize}

This construction requires the introduction of a \emph{\textbf{regulator}}, which is a local functional, dependent on the energy scale $k$. Following \cite{mother} we choose $Q_k\in\loc$ with the form
\begin{equation}
    \label{eq:regulator}
    Q_k(\chi)\defeq -\frac{1}{2}\int_\man d\mu_g(x) k^2f(x)\chi(x)^2,
\end{equation}
where $k>0$.

As a first step we introduce the regulator in the classical free action $I_0$ identified in Equation \eqref{eq:interact_action} to obtain the \emph{regularized action} functional
\begin{equation}
    \label{eq:reg_classaction}
    I_{0,k}\defeq I_0+Q_k\in\loc,
\end{equation}
which yields the regularized equation of motion
$$P_{0,k}\chi=(\Box_g-m^2-k^2-\xi R)\chi=0,$$
where $k$ is the regularizing energy scale introduced before, acting as a supplementary mass term. Although for $k=0$ the operator reduces to the one in Equation \eqref{EQ:KG}, we feel appropriate to include a subscript $0$ to recall that it is associated to the action functional $I_0$. Subsequently we define the \emph{\textbf{regularized generating functionals}}
\begin{equation}
    \label{eq:reg_genZW}
    \begin{gathered}
    Z_k(j)\defeq \omega(R_V[S(J+Q_k)]), \\
    W_k(j) \defeq -i\log Z_k(j).
    \end{gathered}
\end{equation}
The first functional derivative of the connected generating functional $W_k$ has the interpretation of a \emph{\textbf{classical mean field configuration}}, namely, at the level of integral kernel we define, for $j\in C^\infty_0(\man)$
\begin{equation}
    \label{eq:class_field}
    \phi_{j,k}(x)\defeq \frac{\delta\,W_k(j)}{\delta j(x)}=\frac{1}{Z_k(j)} \omega(R_V[S(J+Q_k)\cdot_T \chi(x)]).
\end{equation}
In the following proposition it is shown that the relation connecting $\phi$ to $j$ may be inverted. The proof is given in \cite[Prop. 3.5]{mother} in the case where $(\man,g)$ has an empty boundary, but the generalization to a globally hyperbolic spacetime with a timelike boundary is straightforward and, therefore, we omit it.
\begin{proposition}
    \label{prop:invert_j}
    Let $(\man,g)$ be a Lorentzian manifold, possibly with a non-empty, time-like boundary, and let $\omega$ be a quasi-free Hadamard state on $(\man,g)$ associated to the Klein-Gordon operator $P=\Box_g-m^2-\xi R$. Let $\phi\in \mathcal{E}(\man)$ be such that $\phi=i\Delta_F \widetilde{j}_0$ with $\widetilde{j}_0\in C^\infty_0(\man)$, while $\Delta_F=\Delta_+ + i\Delta_A\in\mathcal{D}'(\man\times\man)$ is the Feynman propagator, where $\Delta_+$ is the two-point correlation function of $\omega$ and $\Delta_A$ the advanced fundamental solution associated to $P$. Then there exists and it is unique $j_{\phi,k}\in C^\infty_0(\man)\llbracket V \rrbracket$, solution to Equation \eqref{eq:class_field}, constructed as a perturbative series with respect to the interaction term $V$. More explicitly  $j_{\phi,k}$ can be obtained by solving
    \begin{equation}
        \label{eq:prop_invertj}
        j=-P_0\phi-Q_k^{(1)}(\phi)-\frac{1}{Z_k(j)}\omega(S(V)^{-1}\star[S(V+Q_k+J)\cdot_T \mathcal{T}V^{(1)}]),
    \end{equation}
    by induction with respect to the perturbative order in $V$.
\end{proposition}

Since the equation for the classical field is invertible on account of Proposition \ref{prop:invert_j}, for $\phi\in C^\infty(\man)$ a fixed, classical field configuration, we can apply the \emph{Legendre transform} to the functional $W_k$ to obtain the \emph{\textbf{regularized effective action}}
$$\widetilde{\Gamma}_k(\phi)\defeq W_k(j_\phi)-J_\phi(\phi),$$
where 
$$J_\phi(\chi)\defeq \int_\man d\mu_g(x)\,j_\phi(x)\phi(x).$$
This functional may be thought as the quantum counterpart of the classical action, as it satisfies the \emph{quantum equation of motion}
\begin{equation}
    \label{eq:quant_eq_motion}
    \frac{\delta\, \widetilde{\Gamma}_k(\phi)}{\delta \phi(x)}=-j_\phi(x),
\end{equation}
by definition of Legendre transform. This entails that, for $j_\phi=0$, the regularized effective action is minimized by the classical configuration $\phi(x)=\omega(R_V[S(Q_k)\cdot_T \chi(x)])/{Z_k(0)}$.
This leads to the definition of the \emph{\textbf{average (regularized) effective action}}
\begin{equation}
    \label{eq:eff_act}
    \Gamma_k(\phi)\defeq \widetilde{\Gamma}_k(\phi)-Q_k(\phi).
\end{equation}
\begin{remark}[$\Gamma_{k=0}$ is the classical action at leading order]
    \label{rmk:lozioremark}
    Bearing in mind Equation \eqref{eq:interact_action}, it is possible to prove that $\Gamma_0^{(1)}=\nord{I^{(1)}}$ at leading order in the perturbation $V$, \iee that ${\Gamma_0}$ is the normal ordering of the classical, interacting action, modulo constant terms. This can be shown by observing that, at first order in $V$ and in the limit $k\rightarrow 0$, Equation \eqref{eq:prop_invertj} in Proposition \ref{prop:invert_j} reads
    \begin{equation}
    \label{eq:gattobianco}
    \begin{gathered}
         j_\phi=-P_0\phi-\frac{1}{Z(j_\phi)}\omega(S(V)^{-1}\star[S(V+J_\phi)\cdot_T \mathcal{T}V^{(1)}])=\\=-P_0\phi-\frac{\omega((\mathbb{I}-iV)\star[(\mathbb{I}+iV)\cdot_T S(J_\phi)\cdot_T \mathcal{T}V^{(1)}])}{\omega((\mathbb{I}-iV)\star[(\mathbb{I}+iV)\cdot_T S(J_\phi)])}+\mathcal{O}(V^2)=-P_0\phi-\frac{\omega( S(J_\phi)\cdot_T \mathcal{T}V^{(1)})}{\omega(S(J_\phi))}.
    \end{gathered}
    \end{equation}
    Using \cite[Lemma B.1]{mother} Equation \eqref{eq:gattobianco} can be rewritten in the form 
    $$j_\phi=-P_0\phi-\mathcal{T}V^{(1)}(\phi)=-I_0^{(1)}\phi-\mathcal{T}V^{(1)}(\phi),$$
    that is the time-ordering of $I^{(1)}$. This entails that
    $$\Gamma_0^{(1)}=\nord{I^{(1)}},$$
    as the time-ordering prescription on local functionals coincides with normal ordering.
\end{remark}

The conclusion that can be drawn from Remark \ref{rmk:lozioremark} is that the effective action, in the limit $k\rightarrow 0$ and at leading order in $V$, yields the same equations of motion as the classical action, modulo the normal ordering prescription necessary to define the local interaction term $V$. Then the natural interpretation is that all quantum contributions are encoded in $\Gamma_{k=0}-I$, $I$ being as per Equation \eqref{eq:interact_action}.

It is worth investigating the limit $k\rightarrow +\infty$. Following the heuristic properties required at the beginning of this section, we aim for a scenario where, in such a limit, all energy modes are suppressed. This means that all quantum contributions to the action are negligible, and it should be expected that $\Gamma_{k=+\infty}-I=0$. This is proven in \cite[Theorem 3.8]{mother} under the hypothesis that the underlying spacetime is a Lorentzian manifold without boundary, ultra-static and of bounded geometry -- for a definition, refer to \cite{bounded_geo}. 

Observe that Minkowski upper half-space $(\uppmink,\eta)$ as per Example \ref{Exam: backgrounds} abides by these requirements while $(\pads,g)$, the Poincaré patch of Anti-de Sitter spacetime, fails to be ultrastatic as one can infer from Equation \eqref{eq:padsmetric}. Yet this hurdle can be removed by means of a conformal transformation, and hence it is possible to reformulate \cite[Theorem 3.8]{mother} in the case where the underlying manifold possesses a time-like, possibly conformal boundary. Prior to this we outline in the following remark the techniques of perturbative agreement employed in the proof.
\begin{remark}[Principle of perturbative agreement]
    \label{rmk:pert_agree}
    As discussed at the end of Section \ref{sec:inter} it is proven in \cite{pert_agr} that, accounting for a quadratic interaction exactly or by means of perturbation theory yields $*$-isomorphic algebras of observables. This can be applied to the quadratic regulator $Q_k$ in Equation \eqref{eq:reg_classaction} which could be regarded as a mass term, yielding the free action $I_{0,k}\defeq I_0+Q_k$, or as an interaction term in $V_k\defeq V+Q_k$. We denote by $\alg(\man)$ the *-algebra of observables induced by the free action $I_0$, while $\alg_k(\man)$ is the one constructed with respect to the free action $I_{0,k}$. We denote by $\star,\cdot_T$ the star and time-ordered products defined on $\alg(\man)$ and $\star_k,\cdot_{T,k}$ the ones on $\alg_k(\man)$.

    Then the perturbative description of the interaction outlined in Section \ref{sec:inter} yields two interacting *-algebras, $\alg\llbracket V_k \rrbracket$ and $\alg_k\llbracket V \rrbracket$, \iee the formal power series in $V_k,\,V$ with coefficients lying respectively in $\alg(\man),\,\alg_k(\man)$. As shown in \cite{pert_agr} and \cite[Appendix A]{mother}, these algebras are related by the classical M{\o}ller isomorphism $r_{Q_k}:\alg_k(\man)\rightarrow \alg(\man)$ such that, for $F\in\alg_k(\man)$ and $\chi\in \mathcal{E}(\man)$
    $$(r_{Q_k}F)(\chi)\defeq F(\moll_{Q_k}\chi),$$
    where $\moll_{Q_k}$ is the classical M{\o}ller map, see Definition \ref{def:mollermap}. Moreover, given $\Delta_F,\Delta_{F,k}$ the Feynman propagators constructed with respect to $I_0,I_{0,k}$, the map $\gamma_k:\loc\rightarrow \loc$ defined as
    $$\gamma_k\defeq \exp\left(\int_{\man\times\man} d\mu_g(x)d\mu_g(y)\,(\Delta_{F,k}-\Delta_F)(x,y)\frac{\delta^2}{\delta\varphi(x)\varphi(y)}\right),$$
    intertwines between the time-ordered products $\cdot_T,\cdot_{T,k}$, namely for $F,G\in\loc$,
    $$\gamma_k(F\cdot_T G)=\gamma_k \mathcal{T}(\mathcal{T}^{-1}F\cdot \mathcal{T}^{-1}G)=\mathcal{T}_k(\mathcal{T}_k^{-1}\gamma_k F\cdot \mathcal{T}_k^{-1}\gamma_k G)=\gamma_k F\cdot_{T,k}\gamma_k G.$$
\end{remark}
We have the tools to restate \cite[Theorem 3.8]{mother} in the case where the underlying manifold possesses a time-like, possibly conformal, boundary.
\begin{theorem}[Classical limit of the effective action] 
\label{teo:classlim_effact}
Let $(\man,g)$ be either the Minkowski upper-half space or the Poincaré patch of Anti-de Sitter, let $\alg_{\Delta_+}(\man)$ be the *-algebra of observables as per Definition \ref{def:deformreg} and let $\omega:\alg_{\Delta_+}(\man)\rightarrow \co$ be the state constructed out of the Hadamard two-point correlation functions as per Equation \eqref{eq:2pointfunct_uppmink_methodimage} or as per Equation \eqref{eq:twopointfunct_pads_2sol}. In the limit where $\supp Q_k\rightarrow \man$, it holds that, 
$$\Gamma^{(1)}_k\xrightarrow{k\rightarrow \infty} I^{(1)},$$
where the limit is taken in the topology of $\mathcal{E}^\prime(\man)$. 
\end{theorem}
\begin{proof}
    Observing that for all $\phi\in\mathcal{E}(\man)$
    \begin{equation*}
    \Gamma_k^{(1)}(\phi)=P_0\phi+\frac{1}{Z_k(j_\phi)}\omega(R_V[S(J_\phi+Q_k)\cdot_T \mathcal{T}V^{(1)}]),
    \end{equation*}
    it is shown in \cite[Lemma 3.7]{mother} by means of the principle of perturbative agreement, see Remark \ref{rmk:pert_agree}, that 
    \begin{equation}
        \label{eq:nonloso}
        \Gamma_k^{(1)}(\phi)=P_0\phi+\frac{1}{Z_k(j_\phi)}\omega_k(S_k(\gamma_kV-\gamma_k Q_k)^{-1}\star_k[S_k(\gamma_k V+J_\phi)\cdot_{T,k}\gamma_k \mathcal{T}V^{(1)}]).
    \end{equation}
    Here $\gamma_k,\star_k,\cdot_{T,k}$ are as per Remark \ref{rmk:pert_agree}, $S_k$ is the local S-matrix constructed using the time-ordering $\mathcal{T}_k$ while $\omega_k=\omega\circ r_{Q_k}$, see Remark \ref{rmk:pert_agree}.
    Using \cite[Lemma B.2]{mother}, Equation \eqref{eq:nonloso} may be written as
    \begin{equation}
    \label{eq:nonloso3}
    \Gamma_k^{(1)}(\phi)=P_0\phi+\frac{\omega_k(S_k(\gamma_k V_{\phi_0}-\gamma_k Q_{k,\phi_0})^{-1}\star_k(S_k(\gamma_k V_{\phi_0})\cdot_{T,k}\gamma_k \mathcal{T}V^{(1)}_{\phi_0})}{\omega_k(S_k(\gamma_k V_{\phi_0}))},
    \end{equation}
    where $F_{\phi_0}(\chi)\defeq F(\chi+\phi_0)$.
    
    As shown in the proof of \cite[Theorem 3.8]{mother}, on Minkowski spacetime the products $\star_k,\cdot_{T,k}$ reduce to the pointwise counterpart in the case where $k\rightarrow +\infty$. This is shown as a consequence of the fact that, for $F,G\in\loc$
    \begin{equation}
        \label{eq:nonloso2}
        \begin{gathered}
            |\langle F^{(n)},\Delta_{+,k}^{\tensor n} G^{(n)}\rangle|\xrightarrow{k\rightarrow \infty} 0,\\
            |\langle F^{(n)},\Delta_{F,k}^{\tensor n} G^{(n)}\rangle|\xrightarrow{k\rightarrow \infty} 0.
        \end{gathered}
    \end{equation}
    If we consider $(\uppmink,\eta)$ and Dirichlet or Neumann boundary conditions, the two-point correlation function of the underlying ground state is built out of Equation \eqref{eq:2pointfunct_uppmink_methodimage} and therefore Equation \eqref{eq:nonloso2} still holds true by direct inspection. 

    If we consider instead $\pads$ and the Dirichlet boundary condition, the starting point for drawing the same conclusion lies in observing that, being the background static, the integral kernel of the Feynman propagator can be expressed with respect to the 2-point correlation function as
    $$\widetilde{\Delta}_{F,k}(t)=\Theta(t)\widetilde{\Delta}_{+,k}(t)+\Theta(-t)\widetilde{\Delta}_{+,k}(-t),$$
    where $\widetilde{\Delta}_{F,k},\,\widetilde{\Delta}_{+,k}$ are as per \cite[Theorem 3.8]{mother}.
    It is then sufficient to note that, for large $k$, the 2-point correlation function $\Delta_{+,k}\in\mathcal{D}'(\pads\times\pads)$ relative to the Klein-Gordon operator $P_{0,k}$ is
    $$\Delta_{+,k}(x,y)=\mathcal{N}(zz')^{\frac{d}{2}}\int_0^\infty dp\,p\left(\frac{p}{r}\right)^{\frac{d-3}{2}}J_{\frac{d-3}{2}}(pr)\int_0^\infty dq\,q\frac{e^{-i\sqrt{p^2+q^2}(t-t'-i\epsilon)}}{\sqrt{2\pi(p^2+q^2)}}J_\nu(qz)J_\nu (qz'),$$
    \cf \cite{dappia_ads}, where $J_\nu$ is the Bessel function of the first kind and where
    $$\nu=\frac{1}{2}\sqrt{1+4l^2\left(m^2+k^2+\left(\xi-\frac{d-1}{4d}\right)R\right)}.$$
    To draw the sought conclusion it is sufficient to observe that, for fixed $x\in\re$, it holds that $\lim_{\nu\rightarrow \infty} J_\nu(x)=0$.
    Hence the 2-point correlation function and the Feynman propagator on $(\pads,g)$ abide by Equation \eqref{eq:nonloso2}. 
    Eventually this allows us to take the limit $k\rightarrow+\infty$ of Equation \eqref{eq:nonloso3}, where $\star_k,\cdot_{T,k}$ reduces to the pointwise product, hence
    $$\Gamma_\infty^{(1)}(\phi)=\lim_{k\rightarrow +\infty} P_0\phi+\frac{\omega_k(S_k(\gamma_k V_{\phi_0}-\gamma_k Q_{k,\phi_0})^{-1}\star_k[S_k(\gamma_k V_{\phi_0})\cdot_{T,k}\gamma_k \mathcal{T}V^{(1)}_{\phi_0})}{\omega_k(S_k(\gamma_k V_{\phi_0}))}= P_0\phi + V^{(1)}(\phi).$$
\end{proof}

\subsection{Lorentzian Wetterich Equation}
\label{sec:wett}
In this section we derive the Wetterich equation on a Lorentzian manifold following once more \cite{mother}. The rationale consists of computing how the generating functionals in Equation \eqref{eq:reg_genZW} scale under a change of $k$ by consistently suppressing energy modes in a continuous manner. 

As a starting point we evaluate the scaling of the regularized generating functional $W_k(j)$ for the connected, time-ordered $n$-point correlation functions. Considering $W_k(j)$ as per Equation \eqref{eq:reg_genZW}, with $j\in\mathcal{E}(\man)$ fixed, it holds that 
\begin{equation}
    \label{eq:polchinski_prev}
    \partial_k W_k(j)=-i\frac{1}{Z_k(j)}\omega(S(V)^{-1}\star[S(V)\cdot_T \partial_k S(J+Q_k)]),
\end{equation}
where any dependence on $k$ is carried by the regulator term and where $J$ is defined in Equation \eqref{eq:class_current}. The contribution $\partial_k S(J+Q_k)$ takes the form
$$\partial_k S(J+Q_k)=i\mathcal{T}[\partial_k Q_k]\cdot_T S(J+Q_k),\quad \mathcal{T}[\partial_k Q_k]=-k\nord{X_f^2} ,$$
where $X_f\in\reg$ is as per Definition \ref{def:reg_funct} while the normal ordering prescription $\nord{\;}$ is subordinated to the Hadamard parametrix, see Definition \ref{def:normal_ord} and Remark \ref{Rem: Covariant Wick Polynomials}. Consequently Equation \eqref{eq:polchinski_prev} can be rewritten in the form
\begin{equation}
    \label{eq:polchinski}
    \partial_k W_k(j)=-\frac{k}{Z_k(j)}\omega(S(V)^{-1}\star[S(V)\cdot_T S(J+Q_k)\cdot_T \nord{X_f^2}]),
\end{equation}
which, at the level of integral kernels, reads
$$\partial_k W_k(j)=-\int_\man d\mu_g(x) \;kf(x)\frac{1}{Z_k(j)}\omega(S(V)^{-1}\star[S(V)\cdot_T S(J+Q_k)\cdot_T \nord{\chi^2(x)}]).$$
Considering the Legendre transform of Equation \eqref{eq:polchinski} we may obtain the analogous equation for the average effective action $\Gamma_k(\phi)$, namely, from Equation \eqref{eq:eff_act}, it descends that 
$$\Gamma_k(\phi)=W_k(j_\phi)-J_\phi(\phi)-Q_k(\phi),$$
where $\phi\in\mathcal{E}(\man)$ is a fixed, classical field configuration. By deriving with respect to the energy scale $k$, we obtain that
\begin{equation}
    \label{eq:wetterich-prev}
    \partial_k \Gamma_k(\phi)=[\partial_k W_k](j_\phi)+\cancel{\langle W_k^{(1)}(j_\phi),\partial_k j_\phi\rangle} -\cancel{\langle \phi,\partial_k j_\phi\rangle}- [\partial_k Q_k](\phi),
\end{equation}
where we used that $\phi=W_k^{(1)}(j_\phi)$ does not depend on $k$ as it is a fixed configuration, while $j_\phi$ carries a dependency on $k$ from Equation \eqref{eq:prop_invertj}. Eventually this leads to the following equation for the effective action functional
\begin{equation}
    \label{eq:wetterich1}
    \partial_k \Gamma_k(\phi)=-\frac{k}{Z_k(j_\phi)}\omega(S(V)^{-1}\star[S(V)\cdot_T S(J_\phi+Q_k)\cdot_T \nord{X_f^2}])+kX_f^2(\phi).
\end{equation}
Inserting the Bogoliubov map $R_V$, defined in Equation \eqref{eq:bogoliub_map}, we obtain the \emph{\textbf{Lorentzian Wetterich equation}}
\begin{equation}
    \label{eq:wetterich}
    \partial_k \Gamma_k(\phi)=-\frac{k}{Z_k(j_\phi)}\omega(R_V[S(J_\phi+Q_k)\cdot_T \nord{X_f^2}])+kX_f^2(\phi),
\end{equation}
This equation has been derived in \cite{mother} on globally hyperbolic Lorentzian manifolds with empty boundary and it is a Lorentzian counterpart of the celebrated \emph{\textbf{Wetterich equation}}, devised in \cite{wett93} on Euclidean spacetimes. Differently from the original version of the Wetterich equation, Equation \eqref{eq:wetterich} is derived by means of a local regulator $Q_k$, that abides by the requirements of a local and covariant theory. 

\begin{remark}[Physical interpretation of the Wetterich equation]
    \emph{The physical picture emerges by considering the analysis performed in the previous section. In particular we have proven in Theorem \ref{teo:classlim_effact} that the regularized average effective action in the limit $k\rightarrow \infty$ coincides, up to a constant, with the classical action $I$. Subsequently we have shown that at $k=0$ the average effective action $\Gamma$ is the quantum corrected counterpart of the classical action $I$, and that it obeys to the quantum equation of motion, see Equation \eqref{eq:quant_eq_motion}.} 
    
   \emph{Let the theory space represents the abstract space of all possible theories, described at the classical level by the action $I$. A useful choice of coordinates of such space is given by the coupling constants, in our case beings the mass $m$, the interaction coupling constant $\lambda$ and the coupling to the scalar curvature, $\xi$. Given an initial condition at $k=+\infty$, \iee $\Gamma_{k=\infty}=I$, the Wetterich equation can be read as an evolution law for the regularized effective action. In this sense the Wetterich equation \eqref{eq:wetterich} describes a trajectory in the \textit{space of theories} under the change of the infrared cutoff parameter $k$.}
   
   \emph{Note that Equation \eqref{eq:wetterich} depends on the choice of the regulator $Q_k$. As a consequence, it should be expected that the trajectory in the space of theories is not unique, as it depends on the choice of $Q_k$. Eventually for $k\rightarrow 0$ the regulator, as discussed in Section \ref{sec:reggenfun}, must vanish, guaranteeing that for every choice of $Q_k$ the Wetterich equation flows at the same effective action $\Gamma_{k=0}$.}
\end{remark}

As it stands, the Lorentzian Wetterich equation \eqref{eq:wetterich} is not directly solvable. In order to do so and make contact with the physical picture, an approximation scheme is needed and here we follow the approach advocated in \cite{mother}. We adopt the \emph{\textbf{Local Potential Approximation (LPA)}}, see also \cite{disp_erge}. The starting point consists of the formulation of an \emph{ansatz} on the form of the average effective action, namely
\begin{equation}
    \label{eq:ansatz_gamma}
    \Gamma_k(\phi)=-\int_\man d\mu_g\;\left(\frac{1}{2}\nabla_a\phi\nabla^a\phi+U_k(\phi)\right),
\end{equation}
where only the \emph{local potential} $U_k$ carries the $k$-dependency, and in addition $U_k$ does not carry derivative operators. Under this assumption the left-hand side of the Wetterich equation \eqref{eq:wetterich} takes the form
\begin{equation}
    \label{eq:wett_sempl_1}
    \partial_k \Gamma_k(\phi)=-\int_\man d\mu_g\,\partial_kU_k(\phi).
\end{equation}
Subsequently the LPA calls for considering a solution $\phi_{cl}$ of the regularized quantum equation of motion such that $j_{\phi_{cl}}=0$, \ie a classical field configuration that minimizes $\Gamma_k$. By performing a Taylor expansion centered at $\phi_{cl}\in \mathcal{E}(\man)$, it holds that
\begin{equation*}
    \Gamma_k(\phi)=\Gamma_k(\phi_{cl})+\cancel{\langle\Gamma_k^{(1)}(\phi_{cl}),\varphi\rangle}+\frac{1}{2}\langle\Gamma_k^{(2)}(\phi_{cl}),\varphi\tensor \varphi\rangle +\mathcal{O}(\varphi^3),
\end{equation*}
where $\varphi\defeq \phi-\phi_{cl}$ is the \emph{fluctuation field}. We define the \emph{\textbf{truncated effective action}}
\begin{equation*}
    \Gamma_k^t(\phi)\defeq \Gamma_k(\phi_{cl})+\frac{1}{2}\langle\Gamma_k^{(2)}(\phi_{cl}),\varphi\tensor \varphi\rangle,
\end{equation*}
which originates from the \emph{\textbf{truncated action}} $I_0^t\in\loc$, that from \cite[Proposition 5.1]{mother} must take the form
\begin{equation}
\label{eq:truncated_action}
    I_0^t(\phi)=-\int_\man d\mu_g\;\left(\frac{1}{2}\nabla_a\phi\nabla^a\phi+\frac{1}{2}U^{(2)}_k(\phi_{cl})\phi^2\right).
\end{equation}
We can rewrite the truncated action as $I_0^t=I_0+M\in\loc$, where
\begin{equation*}
    M(\phi)\defeq -\frac{1}{2}\int_\man d\mu_g\;\left(U^{(2)}_k(\phi_{cl})-m^2-\xi R\right)\phi^2,
\end{equation*}
acts as a quadratic mass term dependent on the classical field configuration $\phi_{cl}$.

By means of perturbative agreement methods, see Remark \ref{rmk:pert_agree}, it has been shown in \cite[Section 5.2]{mother} that, under this approximation scheme, the Wetterich equation takes the form
\begin{equation}
    \label{eq:wett_lpa}
    \partial_k \Gamma_k(\phi)=-\int_\man d\mu_g\,\partial_kU_k(\phi)=-\int_\man d\mu_g(y) \lim_{x\rightarrow y} k\left(\Delta_{S,M,k}(x,y)-H_{M,k}(x,y)\right),
\end{equation}
where $H_{M,k}$ is the Hadamard parametrix associated to the operator ruling the dynamics of the action $I_0^t+Q_k$, while $\Delta_{S,M,k}$ is the symmetric part of the 2-point correlation function
\begin{equation}
    \label{eq:2point_moller}
    \Delta_{+,M,k}\defeq \moll_{M+Q_k} \Delta_+ \moll^*_{M+Q_k}.
\end{equation}
Here $\moll_{M+Q_k}$ is the classical M{\o}ller map associated to $M+Q_k$, see Definition \ref{def:mollermap} while $\Delta_+\in\mathcal{D}'(\man\times\man)$ is the 2-point correlation function associated to the choice of the state $\omega$ in Equation \eqref{eq:wetterich}. 

\section{Applications}
\label{sec:applications}
\subsection{Flow on half Minkowski with Dirichlet boundary conditions}
\label{sec:flowhalf}
In this section we apply the construction implemented in the past sections to the case where the underlying Lorentzian manifold is the Minkowski upper half-space $(\nuppmink^4,\eta)$. Although we could repeat the analysis for a generic value of $d>1$, we shall restrict the attention to four-dimensional manifolds to make contact with the physically interesting scenarios. Let us consider an action of the form
\begin{equation}\label{eq:kgaction_zdep}
I(\chi)=-\int_{\mathring{\nuppmink^4}}d\mu_\eta(x)\,\left(\frac{1}{2}\nabla_a\phi(x)\nabla^a\phi(x) +\frac{m^2}{2}\phi(x)^2+\lambda\frac{\phi(x)^4}{4!}\right)f(x),
\end{equation}
supplemented with Dirichlet boundary conditions at $\partial \uppmink$.

Since $\nuppmink^4$ is not invariant under translations in the direction perpendicular to $\partial\nuppmink^4$, it descends that, in analyzing the Wetterich equation in this framework, one should consider the coupling constants $m,\lambda$ in Equation \eqref{eq:kgaction_zdep} as functions of the $z$-coordinate, \iee $m=m(z),\,\lambda=\lambda(z)$. 

Under the Local Potential Approximation (LPA), outlined in Section \ref{sec:wett}, we assume the following \textit{ansatz} for the effective action:
\begin{equation}
    \label{lpa_uppmink}
    \Gamma_k(\phi)=-\int_{\mathring{\nuppmink^4}}d\mu_\eta(x)\,\frac{1}{2}\nabla_a\phi(x)\nabla^a\phi(x)+U_k(\phi),
\end{equation}
where $U_k(\phi)\defeq U_{0,k}(z)+m^2_k(z)\frac{\phi^2}{2}+\lambda_k(z) \frac{\phi^4}{4!}$ is the local potential. 
The $k$ dependence is carried by the term $U_k$, and the left-hand side of the Wetterich equation \eqref{eq:wetterich} takes the form
\begin{equation}
    \label{scalingwrtU}
    \partial_k \Gamma_k(\phi)=-\int_{\mathring{\nuppmink^4}} d\mu_\eta\,\partial_k U_k(\phi).
\end{equation}
Subsequently we expand Equation \eqref{lpa_uppmink} near a solution $\phi_{cl}$ of the quantum equation of motion
$$\tilde{\Gamma}_k^{(1)}(\phi_{cl})=0,$$
and we truncate the average effective action to the quadratic order in $\varphi\defeq \phi-\phi_{cl}$, obtaining
$$\Gamma_k^t(\phi)=\Gamma_k(\phi_{cl})+\langle \Gamma_k^{(2)}(\phi_{cl}),\varphi\tensor\varphi\rangle+\mathcal{O}(\varphi^3).$$
Eventually we obtain the following form for the Wetterich equation:
\begin{equation}
    \label{eq:lpawett_halfmink}
    \begin{gathered}
        \partial_k \Gamma_k(\phi)=-\int_{\nuppmink^4} d\mu_g(x)\,\partial_k \left(U_{0,k}(z)+m^2_k(z)\frac{\phi^2(x)}{2}+\lambda_k(z) \frac{\phi^4(x)}{4!}\right)=\\=-\int_{\nuppmink^4} d\mu_g(x) \lim_{y\rightarrow x} k\left(\Delta_{S,M,k}(y,x)-H_{M,k}(y,x)\right),
    \end{gathered}
\end{equation}
where $\Delta_{S,M,k},\,H_{M,k}(y,x)$ are as per Equation \eqref{eq:wett_lpa}.

To evaluate Equation \eqref{eq:lpawett_halfmink} we have to compute the symmetric 2-point correlation function $\Delta_{S,M,k}$, and we have to choose an Hadamard parametrix $H_{M,k}$. Starting from the former, this amounts to computing the 2-point correlation function $\Delta_{+,M,k}\in\mathcal{D}'(\nuppmink^4\times\nuppmink^4)$. From Equation \eqref{eq:wett_lpa} this is
\begin{equation} 
\label{eq:mollon_uppmink}
\Delta_{+,M,k}=\moll_{Q_k+M}\Delta_+ \moll^*_{Q_k+M},
\end{equation}
where $\moll_{Q_k+M}$ is the classical M{\o}ller map, see Definition \ref{def:mollermap}, $\moll^*_{Q_k+M}$ is its formal adjoint while
\begin{equation}
    \begin{gathered}
        Q_k(\phi)\defeq -\frac{1}{2}\int_{\mathring{\nuppmink^4}}d\mu_\eta(x)\,k^2\phi^2(x)f(x),\\
        M(\phi)\defeq  -\frac{1}{2}\int_{\mathring{\nuppmink^4}}d\mu_\eta(x)\,\left(m^2_k(z)+\lambda_k(z)\frac{\phi^2_{cl}(x)}{2} \right)\phi^2(x).
    \end{gathered}
\end{equation}
Having imposed Dirichlet boundary conditions, the r\^{o}le of $\Delta_+$ can be played by $\Delta^{\mathbb{H}}_{+,D}$ as in Equation \eqref{eq:2pointfunct_uppmink_methodimage}, namely the integral kernel of Equation \eqref{eq:mollon_uppmink} takes the form
\begin{equation}
    \label{eq:mollon_uppmink-2}
    \Delta_{+,M,k}(x,y)=\moll_{Q_k+M}[\Delta^{\nmink^4}_+(x,y)-\Delta^{\nmink^4}_+(x,\iota_z y)]\moll^*_{Q_k+M},
\end{equation}
where $\Delta_+^{\nmink^4}$ is the two-point correlation function of the Poincaré invariant ground state on the whole four-dimensional Minkowski spacetime. Using principle of perturbative agreement \cite[Lemma D1]{pert_agr}, in the adiabatic limit it holds that
$$\moll_{Q_k+M}\Delta^{\nmink^4}_+\moll^*_{Q_k+M}\xrightarrow{f\rightarrow 1}\Delta^{\nmink^4}_{+,k},$$
where $\Delta^{\nmink^4}_{+,k}$ is the 2-point correlation function of a theory with mass $\displaystyle m^2_k+\lambda_k\frac{\phi_{cl}^2}{2}+k^2$. In the adiabatic limit Equation \eqref{eq:mollon_uppmink} becomes
\begin{equation}
    \label{eq:mollupp}
    \moll_{Q_k+M}\Delta_+(x,y)\moll^*_{Q_k+M}\xrightarrow{f\rightarrow 1}\Delta^{\nmink^4}_{+,k}(x,y)-\Delta^{\nmink^4}_{+,k}(x,\iota_z y),
\end{equation}
where, as discussed before, we allow for $m^2_k=m^2_k(z),\,\lambda_k=\lambda_k(z)$.

Having chosen a 2-point correlation function we have to select an Hadamard parametrix  $H_{M,k}$, related to the theory whose action is $I_0^t+Q_k$. This amounts to choosing an Hadamard parametrix for a theory whose mass is  $M^2=m^2_k+\lambda_k {\phi^2_{cl}}/{2}+k^2$. Following Equation \eqref{eq:hadparametrix_uppmink_methodimage}, as well as \cite[Appendix A]{bdf09} the Hadamard parametrix $H^{\nmink^4}_M$ for a theory with mass $M^2$ on $\nmink^4$ reads
\begin{equation}
    \label{eq:hadparam_mink}
    H^{\nmink^4}_M(x,y)=\Delta^{\nmink^4}_{S,M}(x,y)+\frac{M}{8\pi^2}\log\left(\frac{\mu^2}{M^2}\right) [\sqrt{-\sigma(x,y)}]^{-1} I_1[\sqrt{-M^2\sigma(x,y)}],
\end{equation}
where $I_1$ is the modified Bessel function of the first kind, see \cite[\S 10]{nist_bessel}, $\sigma$ is the geodesic distance on $\nmink^4$, $\mu^2>0$ is an arbitrary energy parameter and $\Delta^{\nmink^4}_S$ is the symmetric part of the 2-point correlation function $\Delta_+^{\nmink^4}$. 

With reference to the choice of the Hadamard parametrix on $(\nuppmink^4,\eta)$ we have two possibilities:
\begin{enumerate}
    \item consider in Equation \eqref{eq:lpawett_halfmink} the \emph{full Hadamard parametrix} as per Equation \eqref{eq:hadparametrix_uppmink_methodimage}
        \begin{equation}
            \label{hadpar_uppmink}
             H^{\mathbb{H}^4}_D(x,y)= H^{\nmink^4}(x,y)- H^{\nmink^4}(x,\iota_z y).
        \end{equation}
        We call this choice \emph{\textbf{full Hadamard subtraction}}.
    \item Since the subtraction in Equation \eqref{eq:lpawett_halfmink} is performed at the coinciding point limit, to get rid of the singular part of the 2-point function the minimal choice consists of subtracting only the \emph{non-reflecting} part of the Hadamard parametrix, namely
    \begin{equation}
            \label{hadpar_uppmink_nonrelf}
             H^{\mathbb{H}^4}(x,y)= H^{\nmink^4}(x,y)
        \end{equation}
        We call this choice \emph{\textbf{minimal Hadamard subtraction}}.
\end{enumerate}
Note that the second option is the one usually employed in the study of the \emph{Casimir effect}, see \cite{dappia_casimir} and hence it appears to be at first glance preferable from a physical viewpoint. Yet, from a mathematical perspective, the first choice seems quite natural since it is the one necessary to construct a global and covariant algebra of Wick polynomials via a subtraction scheme. In the following we shall see that these two possibilities lead to notable differences also when studying the flow associated to the Wetterich equation.

\subsubsection{Full Hadamard subtraction}
\label{sec:full_sub}
Considering the full Hadamard parametrix as per Equation \eqref{hadpar_uppmink}, to compute the solutions to Equation \eqref{eq:lpawett_halfmink} we have to evaluate the following coincidence limit:
\begin{equation}
\label{scalingU_reg}
\begin{gathered}
    \partial_k U_k(\phi)=k\lim_{x\rightarrow y}(\Delta_{S,M}(x,y)-H^{\mathbb{H}^4}_{D}(x,y))=\\=k\lim_{x\rightarrow y}[(\Delta^{\nmink^4}_{S,M}(x,y)-H^{\nmink^4}_M(x,y))-(\Delta^\nmink_{S,M}(x,\iota_z y)-H^{\nmink^4}_M(x,\iota_z y))]=\\=k\lim_{x\rightarrow y}\frac{M}{8\pi^2}\log\left(\frac{M^2}{\mu^2}\right) \left[(\sqrt{-\sigma(x,y)})^{-1} I_1(\sqrt{-M^2\sigma(x,y)})+\right.\\\left.-(\sqrt{-\sigma(x,\iota_z y)})^{-1} I_1(\sqrt{-M^2\sigma(x,\iota_z y)})\right].
\end{gathered}
\end{equation}
Taking into account that
$$\lim_{x\rightarrow y}\sigma(x,\iota_z y)=2z^2,$$
we can compute the first term between square brackets in Equation \eqref{scalingU_reg}, that is
$$\lim_{x\rightarrow y} \frac{I_1(iM\sqrt{\sigma(x,y)})}{(i\sqrt{\sigma(x,y)})}=\frac{M}{2},$$
while the second one reads
$$\lim_{x\rightarrow y} \frac{I_1(iM\sqrt{\sigma(x,\iota_z y)})}{(i\sqrt{\sigma(x,\iota_z y)})}=\frac{I_{1}(\sqrt{2}iMz)}{\sqrt{2}iz}$$

\noindent From Equation \eqref{scalingU_reg} we obtain 
\begin{equation}
\label{scaling_ind}
\partial_k U_k(\phi)=\frac{k}{16 \pi^2}M^2\log\left(\frac{M^2}{\mu^2}\right)\left(1-\mathfrak{B}(z,M)\right),
\end{equation}
where $\mathfrak{B}(z,M)$ is the \textbf{\emph{boundary contribution}}
\begin{equation}
\label{Eq:factorB}    
\mathfrak{B}(z,M)\defeq \sqrt{2}\frac{I_1(\sqrt{2}iMz)}{iMz}.
\end{equation}

\noindent This result can be compared with the one obtained in \cite[Section 6.1.1]{mother} on $\nmink^4$:
\begin{equation}
\label{scaling_ind_fullmink}
\partial_k U_k(\phi)=\frac{k}{16 \pi^2}M^2\log\left(\frac{M^2}{\mu^2}\right),
\end{equation}
from which we can infer that the scaling in Equation \eqref{scaling_ind}, far from $z=0$, approaches exponentially fast that of Equation \eqref{scaling_ind_fullmink},  see Figure \ref{fig:bterm}, namely
$$\lim_{z\rightarrow \infty} \mathfrak{B}(z,M)=0.$$

\noindent We can compute the \emph{\textbf{$\bm{\beta}$-functions}} for the coupling constants $m^2,\lambda$, that allow to extract physical predictions from the model that we are studying. In this scenario the Wetterich equation takes the form
\begin{equation}
    \nonumber
    \begin{gathered}
    \partial_k \Gamma_k(\phi)=-\partial_k\int_{\nuppmink^4} dx\,\frac{k}{16\pi^2}M^2\log \left(\frac{M^2}{\mu^2}\right)\Big(1-\mathfrak{B}(z,M)\Big)=\\=-\int_{\nuppmink^4} dx\, \partial_k U_{0,k}+\partial_k m^2_k\frac{\phi^2}{2}+\partial_k\lambda_k\frac{\phi^4}{4!}.    
    \end{gathered}
\end{equation}
By considering higher order functional derivatives of $\partial_k U_k(\phi)$ we can isolate the $\beta$-functions:
\begin{equation}
   \begin{cases}
    \beta_{U_0}\defeq \partial_k U_k(\phi=0)=\partial_k U_{0,k},\\
    \beta_{m^2}\defeq\partial_k U_k^{(2)}(\phi=0)=\partial_k m_k^2,\\
    \beta_\lambda\defeq\partial_k U_k^{(4)}(\phi=0)=\partial_k \lambda_k.\\
\end{cases} 
\end{equation}
As a consequence we obtain the following form for the scaling of the coupling constants $m^2_k,\lambda_k$:
\begin{equation}
    \label{eq:scaling_couplconst_dimensionful}
     \left\{
     \begin{gathered}
     \begin{aligned}
       &\partial_k m^2_k=\frac{k\lambda_k}{16\pi^2}\left[(1-\mathfrak{B}(z,\sqrt{k^2+m_k^2}))\left(\log\left(\frac{k^2+m_k^2}{\mu^2}\right)+1\right)\right.+\\&\quad\quad\quad\quad\quad\quad+\left.(k^2+m_k^2)\log\left(\frac{k^2+m_k^2}{\mu^2}\right)\mathfrak{B}_2(z,\sqrt{k^2+m_k^2})\right]\\
       \quad \\
       &\partial_k \lambda_k=\frac{k\lambda_k^2}{16\pi^2}\left[\frac{3}{k^2+m_k^2}(1-\mathfrak{B}(z,\sqrt{k^2+m_k^2}))+\right.\\ &\quad\quad\quad\quad\quad\quad -6\left(\log\left(\frac{k^2+m_k^2}{\mu^2}\right)+1\right)\mathfrak{B}_2(z,\sqrt{k^2+m_k^2})+ \\ &\quad\quad\quad\quad\quad\quad+\left.(k^2+m_k^2)\log\left(\frac{k^2+m_k^2}{\mu^2}\right)\mathfrak{B}_4(z,\sqrt{k^2+m_k^2})\right]
       \end{aligned} 
        \end{gathered} 
        \right.
\end{equation}
where $\mathfrak{B}_2,\mathfrak{B}_4$ are 
\begin{equation}\label{Eq: B2}
	\mathfrak{B}_2(z,\sqrt{k^2+m_k^2})=-\frac{I_2\left(i \sqrt{2} \sqrt{k^2+m^2} z\right)}{k^2+m^2}\;\mathrm{and}\;\mathfrak{B}_4(z,\sqrt{k^2+m_k^2})=-\frac{3izI_3\left(i \sqrt{2} \sqrt{k^2+m^2} z\right)}{\sqrt{2}(k^2+m^2)^{\frac{3}{2}}}
\end{equation}

As customary in the study of the flow of coupling constants, we can write the system in an autonomous form by introducing the following \emph{dimensionless coupling constants} when $d=3$:
\begin{equation}
    \label{eq:rescaling}
   \left\{
   \begin{array}{l}
    m_k\;\longmapsto\;\widetilde{m}^2_{k}=k^{-2}m^2_k,\\
    \lambda_k\;\longmapsto\;\widetilde{\lambda}_k = \lambda_k.
\end{array} 
\right.
\end{equation}
Rescaling also the $z$-direction so to obtain the \emph{dimensionless distance} $z\mapsto \widetilde{z}=kz$ and observing that $\mathfrak{B}(z,M)$ in Equation \eqref{Eq:factorB} is manifestly dimensionless, as it depends on $zM=\widetilde{z}\widetilde{M}$, we obtain
\begin{equation}
    \label{eq:scaling_couplconst_dimensionless}
     \left\{
     \begin{gathered}
     \begin{aligned}
       &k\partial_k \widetilde{m}^2_k=-2\widetilde{m}_k^2+\frac{\widetilde{\lambda}_k}{16\pi^2}\left[\left(1-\mathfrak{B}\left(\widetilde{z},\sqrt{1+\widetilde{m}_k^2}\right)\right)\left(\log\left(\frac{k^2}{\mu^2}\right)+\log(\widetilde{m}_k^2+1)+1\right)\right.+\\&\quad\quad\quad\quad\quad\quad\quad\quad\quad\quad+\left.(1+\widetilde{m}_k^2)\left(\log\left(\frac{k^2}{\mu^2}\right)+\log(1+\widetilde{m}_k^2)\right)\mathfrak{B}_2\left(\widetilde{z},\sqrt{1+\widetilde{m}_k^2} \right)\right]\\
       \quad \\
       &k\partial_k \widetilde{\lambda}_k=\frac{\widetilde{\lambda}_k^2}{16\pi^2}\left[\frac{3}{1+\widetilde{m}_k^2}\left(1-\mathfrak{B}\left(\widetilde{z},\sqrt{1+\widetilde{m}_k^2}\right)\right)+\right. \\&\quad\quad\quad\quad\quad\quad- 6\left(\log\left(\frac{k^2}{\mu^2}\right)+\log(\widetilde{m}_k^2+1)+1\right)\mathfrak{B}_2\left(\widetilde{z},\sqrt{1+\widetilde{m}_k^2}\right)+ \\ &\quad\quad\quad\quad\quad\quad +\left.(1+\widetilde{m}_k^2)\left(\log\left(\frac{k^2}{\mu^2}\right)+\log(\widetilde{m}_k^2+1)\right)\mathfrak{B}_4\left(\widetilde{z},\sqrt{1+\widetilde{m}_k^2}\right)\right]
       \end{aligned} 
        \end{gathered} 
        \right.
\end{equation}

The residual dependence on the energy scale $k$ lies in the logarithmic term. As pointed out in \cite{mother} this dependence can be removed by tuning the arbitrary renormalization scale to $\mu=k$, which yields
\begin{equation}
    \label{eq:scaling_couplconst_dimensionless_choosemu}
     \left\{
     \begin{gathered}
     \begin{aligned}
       &k\partial_k \widetilde{m}^2_k=-2\widetilde{m}_k^2+\frac{\widetilde{\lambda}_k}{16\pi^2}\left[\left(1-\mathfrak{B}\left(\widetilde{z},\sqrt{1+\widetilde{m}_k^2}\right)\right)\left(\log(\widetilde{m}_k^2+1)+1\right)\right.+\\&\quad\quad\quad\quad+\left.(1+\widetilde{m}_k^2)\log(1+\widetilde{m}_k^2)\mathfrak{B}_2\left(\widetilde{z},\sqrt{1+\widetilde{m}_k^2} \right)\right]\\
       \quad \\
       &k\partial_k \widetilde{\lambda}_k=\frac{\widetilde{\lambda}_k^2}{16\pi^2}\left[\frac{3}{1+\widetilde{m}_k^2}\left(1-\mathfrak{B}\left(\widetilde{z},\sqrt{1+\widetilde{m}_k^2}\right)\right)\right.-6\left(\log(\widetilde{m}_k^2+1)+1\right)\mathfrak{B}_2\left(\widetilde{z},\sqrt{1+\widetilde{m}_k^2}\right)+ \\ &\quad \quad\quad \quad +\left.(1+\widetilde{m}_k^2)\log(\widetilde{m}_k^2+1)\mathfrak{B}_4\left(\widetilde{z},\sqrt{1+\widetilde{m}_k^2}\right)\right]
       \end{aligned} .
        \end{gathered} 
        \right.
\end{equation}

For large values of $z$ the function $\mathfrak{B}$ as well as its derivatives $\mathfrak{B}_2,\mathfrak{B}_4$ are exponentially decreasing to $0$, which entails, that for $z\gg 1$, the flow equations reduce to the counterpart on the whole Minkowski spacetime $(\nmink^4,\eta)$, in agreement up to an irrelevant multiplicative factor $2$ with those obtained in \cite{mother}:
\begin{equation}
    \label{eq:scaling_mink}
     \left\{
     \begin{gathered}
     \begin{aligned}
       &k\partial_k \widetilde{m}^2_k=-2\widetilde{m}_k^2+\frac{\widetilde{\lambda}_k}{16\pi^2}\left[\log(\widetilde{m}_k^2+1)+1\right]\\
       &k\partial_k \widetilde{\lambda}_k=\frac{3\widetilde{\lambda}_k^2}{16\pi^2}\frac{1}{1+\widetilde{m}_k^2}
       \end{aligned} 
        \end{gathered} 
        \right. .
\end{equation}

Equation \eqref{eq:scaling_couplconst_dimensionless_choosemu} can be rewritten in an explicitly autonomous form choosing the parametrization with respect to the \emph{renormalization time} $t\defeq \log(\Lambda/k)$, where $\Lambda$ is an energy scale. In this way it is possible to write, for $f=\widetilde{m}^2,\widetilde{\lambda}$,
$$k\partial_k f(k)=k\frac{\partial t}{\partial k}\partial_t f(t)=-k\frac{k}{\Lambda}\frac{1}{k^2}\partial_t f(t)=-\frac{1}{\Lambda}\partial_t f(t).$$

To analyse the behaviour of a system, as the one in Equation \eqref{eq:scaling_mink}, under the action of the renormalization group flow, it is useful to introduce the following nomenclature, here listed for convenience for a general dimension.
\begin{definition}[Relevant, marginal and irrelevant couplings]
    \label{def:relmarg}
    Let us consider the generalized Klein-Gordon action on the $d+1$-dimensional Minkowski spacetime
    \begin{equation}
        \label{eq:general_kgaction}
        I_{gen}(\phi)\defeq \int_{\mink} d\mu_\eta(x)\,\left(\frac{1}{2}\nabla_a\phi(x)\nabla^a\phi(x) +\alpha_2 \phi^2(x)+\alpha_4 \phi^4(x)+\ldots+\alpha_{2N}\phi^{2N}(x)\right)f(x),
    \end{equation}
    where $N\in\mathbb{N}$ and $\alpha_{2n}\in\re$ for every $n=1,\ldots,N$. Considering only even powers of $\phi$ to preserve the $O(1)$ symmetry $\phi(x)\rightarrow -\phi(x)$, we say that
    \begin{itemize}
        \item  $\phi^{2n}$ is a \emph{\textbf{relevant operator}} if $d(1-n)+1+n>0$;
        \item $\phi^{2n}$ is a \emph{\textbf{marginal operator}} if $d(1-n)+1+n=0$;
        \item $\phi^{2n}$ is a \emph{\textbf{irrelevant operator}} if $d(1-n)+1+n<0$.
    \end{itemize}
    The same terminology applies to the associated coupling constants $\alpha_{2n}$.
\end{definition}

\begin{remark}
Under this classification, it turns out that, for the systems in Equations \eqref{eq:scaling_couplconst_dimensionless_choosemu}, \eqref{eq:scaling_mink} where $d+1=4$, the mass is a relevant coupling constant while $\lambda$ is marginal. 

The relevance of this classification become manifest if we consider the resulting $\beta$-function for a generic coupling constant $\alpha_{2n}$. Indeed, when we apply a rescaling procedure as the one in Equation \eqref{eq:rescaling} to obtain the dimensionless coupling constant $\widetilde{\alpha}_{2n}$, we get
$$\alpha_{2n} \longmapsto \widetilde{\alpha}_{2n}\defeq k^{d(1-n)+1+n}\alpha_{2n},$$
where the dimension of $\alpha_{2n}$ is determined again by dimensional analysis. It is clear that, as we lower the energy cutoff $k$, three different behaviors should be expected:
\begin{itemize}
    \item $\widetilde{\alpha}_{2n}$ grows if it is a relevant coupling constant;
    \item $\widetilde{\alpha}_{2n}$ is invariant if it is a marginal one;
    \item $\widetilde{\alpha}_{2n}$ vanishes if it is an irrelevant one.
\end{itemize}
\end{remark}

In order for Equation \eqref{eq:scaling_couplconst_dimensionless_choosemu} to yield a unique solution, one should assign initial conditions at a given scale $\Lambda$, so that $\widetilde{m}^2_{k=\Lambda}=\widetilde{m}^2_\Lambda$ and $\widetilde{\lambda}_{k=\Lambda}=\widetilde{\lambda}_\Lambda$. The physical interpretation is that, at such energy, the values of $m^2$ and of $\lambda$ ought to be measured. Subsequently the flow of the theory, from $k=\Lambda$ up to $k=0$, is implemented by the system in Equation \eqref{eq:scaling_couplconst_dimensionless_choosemu}. As a consequence the energy scale $\Lambda$ acquires the r\^ole of a physical energy scale at which the theory is defined.

\begin{figure}[H]
	\begin{minipage}[b]{0.45\linewidth}
		\centering
		\includegraphics[width=\textwidth]{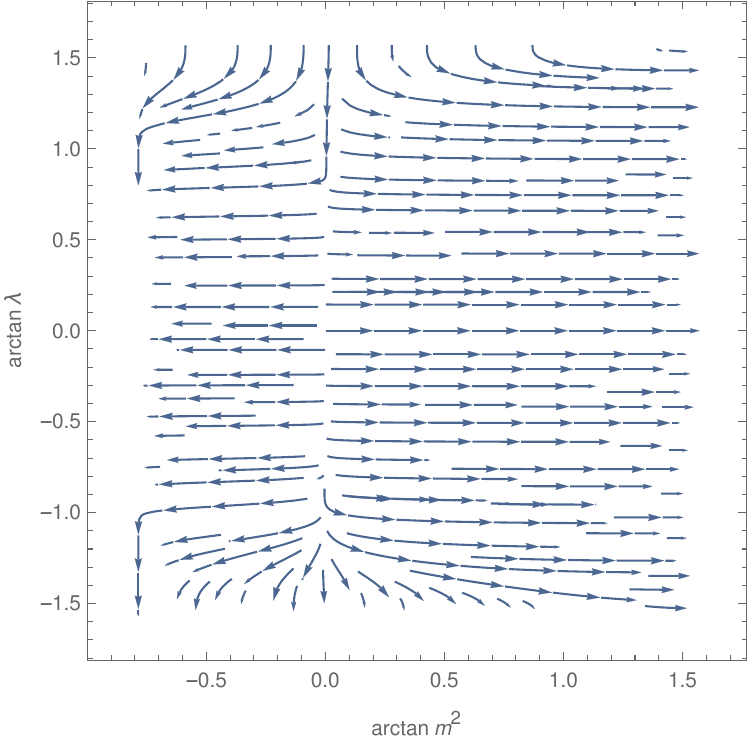}
		\caption{Scaling of the coupling constants on the whole Minkowski spacetime $(\nmink^4,\eta)$. The arrows point in the infrared limit.}
		\label{fig:scaling_mink}
	\end{minipage}
	\hspace{0.5cm}
	\begin{minipage}[b]{0.45\linewidth}
		\centering
		\includegraphics[width=\textwidth]{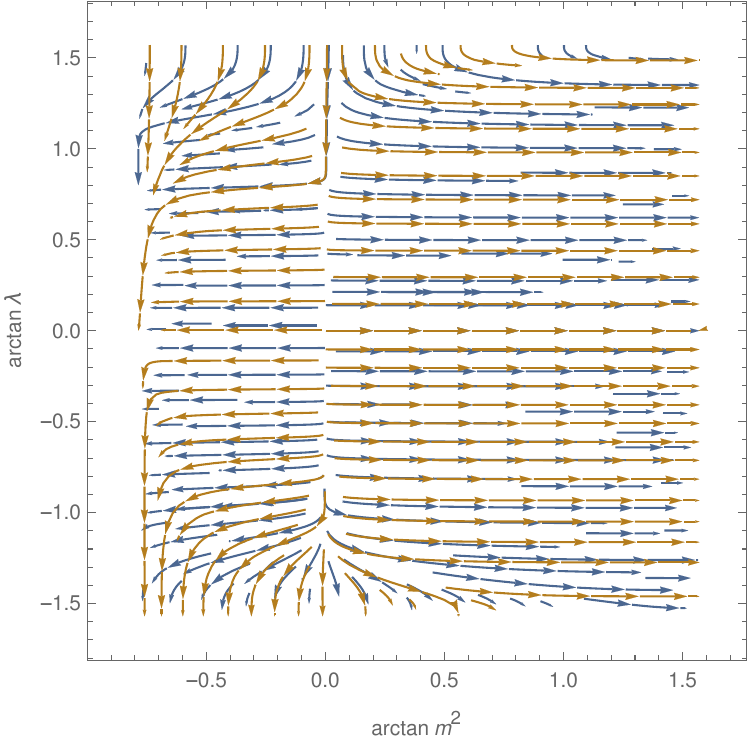}
		\caption{Comparison between the results obtained with a non-local regulator (orange) and the local regulator (blue).}
		\label{fig:scaling_nonlocal}
	\end{minipage}
\end{figure}

Since an analytical solution of Equation \eqref{eq:scaling_couplconst_dimensionless_choosemu} seems unpractical, it is fruitful to visualize the stream lines in the plane $m^2,\lambda$ using the method $\mathsf{StreamPlot}[\;]$ of Wolfram Mathematica. In Figure \ref{fig:scaling_mink}, as a reference, it is represented the scaling of the coupling constants on the whole Minkowski spacetime following Equation \eqref{eq:scaling_mink}, a result that will be presented in \cite{phdthesis_dan}. We have chosen to plot $\tan m, \tan\lambda$ in place of $m^2,\lambda$ to better visualize the asymptotic behaviour of the system. The choice of the direction of the arrows is such that they point towards the infrared region $k=0$. The range includes the non-physical region of negative masses and negative interaction constants. In Figure \ref{fig:scaling_nonlocal} we compare the results obtained with the developed formalism with those coming from using a non-local regulator, known also as optimized cut-off \cite{Litim:2001up}, showing that they are qualitatively the same. Moreover, it is clear from the stream lines that the mass is a relevant coupling constant for the renormalization group flow.

Considering $\mathbb{H}^4$, we have plotted in Figures \ref{fig:scal_dist1}, \ref{fig:scal_dist001} the stream lines in correspondence of different values of $\widetilde{z}$ ($\widetilde{z}=1$ and $\widetilde{z}=0.01$). The behaviour is very similar to the one obtained on Minkowski spacetime, even for small values of $\widetilde{z}$. The reason lies in the form of $\mathfrak{B}(\widetilde{z},\widetilde{M})$, plotted in Figure \ref{fig:bterm}, which for small, positive values of $\widetilde{z}$ does not give large corrections, and it rapidly vanishes for large $\widetilde{z}$.

\begin{figure}[H]
\begin{minipage}[b]{0.45\linewidth}
      \centering
    \includegraphics[width=\textwidth]{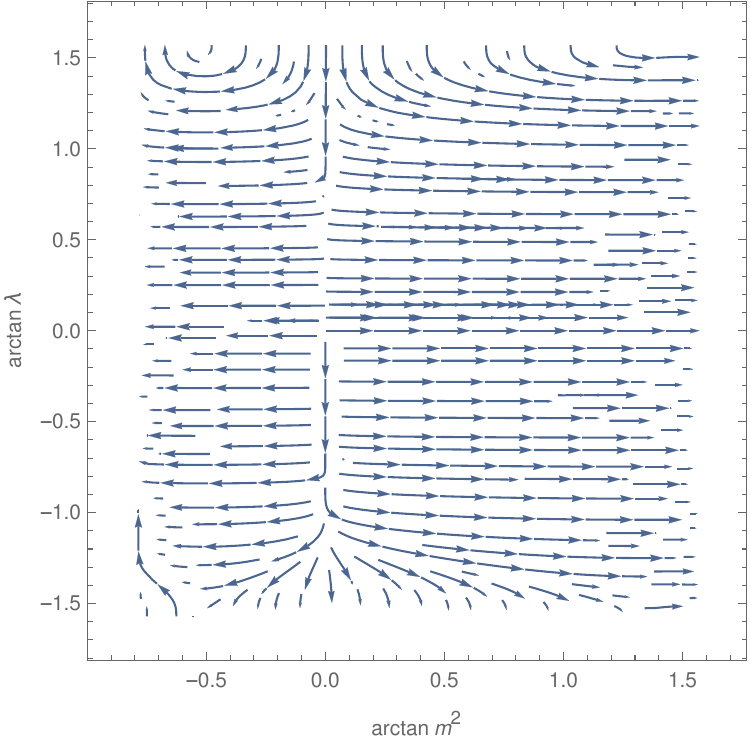}
    \caption{Scaling of the coupling constants on Minkowski upper half-space. Dimensionless distance from the boundary: $\widetilde{z}=1$}
    \label{fig:scal_dist1}
\end{minipage}
\hspace{0.5cm}
\begin{minipage}[b]{0.45\linewidth}
\centering
\includegraphics[width=\textwidth]{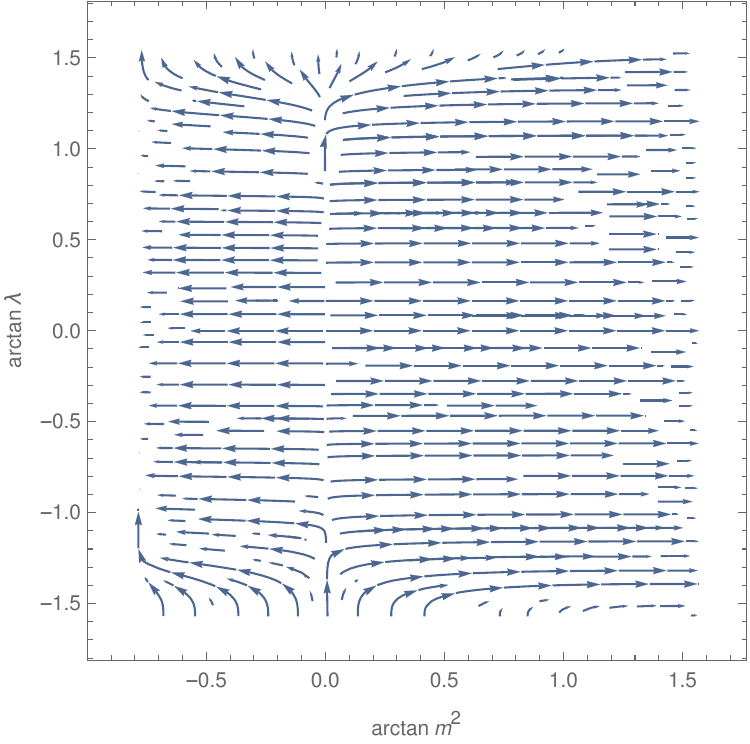}
    \caption{Scaling of the coupling constants on Minkowski upper half-space. Dimensionless distance from the boundary: $\widetilde{z}=0.01$}
    \label{fig:scal_dist001}
\end{minipage}
\end{figure}
In Figure \ref{fig:conf_beta_full} we present the results for the scaling of $\widetilde{\lambda}$ in the massless case. It is useful to study the system on the {\em critical surface}, where the relevant coupling constants are set to $0$. In other words we set $m=0$ to isolate the behaviour of $\lambda$ as the mass, being a relevant coupling, obscures its scaling. Note that for large values of $\widetilde{z}$ the scaling approaches the one on Minkowski spacetime, while if $\widetilde{z}$ is close to $0$, the $\beta$-function continues to remain positive,  approaching $0$.
\begin{figure}[H]
\begin{minipage}[b]{0.45\linewidth}
     \centering
    \includegraphics[width=\textwidth]{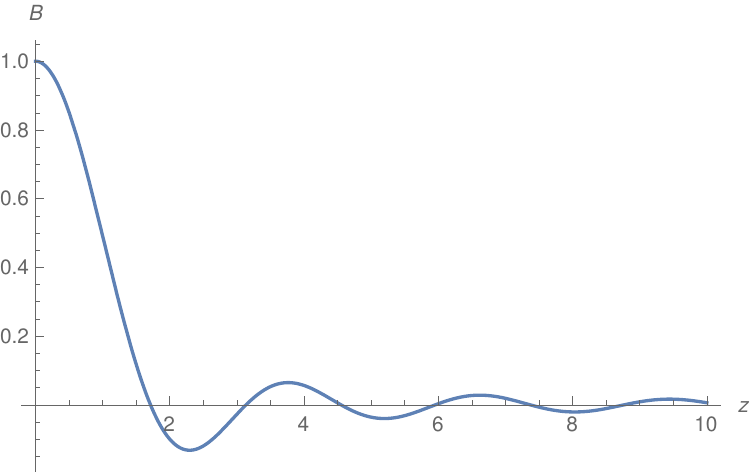}
    \caption{Plot of $\mathfrak{B}(\widetilde{z},\widetilde{M})$ at fixed $M$, with respect to $\widetilde{z}$.\\}
    \label{fig:bterm}
\end{minipage}
\hspace{0.5cm}
\begin{minipage}[b]{0.45\linewidth}
\centering
\includegraphics[width=\textwidth]{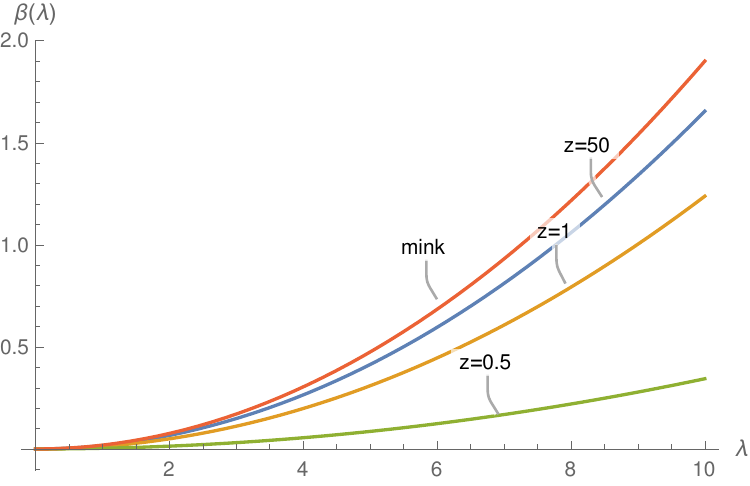}
    \caption{Comparison between the $\beta$-functions of $\lambda$, in the massless case, at different values of $\widetilde{z}$.}
    \label{fig:conf_beta_full}
\end{minipage}
\end{figure}

\subsubsection{Minimal Hadamard subtraction}
In this section we consider the second subtraction prescription outlined at the end of Section \ref{sec:flowhalf}, namely we work with Equation \eqref{eq:lpawett_halfmink}, that is
\begin{equation}
\nonumber
\partial_k U_k(\phi)=\lim_{x\rightarrow y}k(\Delta_{S,M,k}(x,y)-H^{\nmink^4}(x,y)),
\end{equation}
where $\Delta_{S,M,k}$, as discussed in Section \ref{sec:flowhalf}, is the symmetric part of the 2-point function on $(\nuppmink^4,\eta)$, constructed by means of the method of images, see Equation \eqref{eq:mollupp}. In this section we denote by $\Delta_+^{\nmink^4}(x,y)$ the 2-point correlation function on the four dimensional Minkowski spacetime associated to a real scalar field of mass $M$:
\begin{equation}
    \begin{gathered}
        \Delta_+^{\nmink}(x,y)= \frac{1}{(2\pi)^4} \int_{\re^4}  d^4 p \,\delta\left( p_\mu p^\mu + M^2 \right) \Theta( -p_0 ) e^{i p_\mu (x^\mu-y^\mu) } = \\ = \frac{1}{(2\pi)^{3}} \int_{\re^3} \frac{d^3 \mathbf{p}}{2 \omega_{\mathbf{p}}} e^{-i \omega_{\mathbf{p}}(x^0 - y^0) + i \mathbf{p} \cdot (\mathbf{x} - \mathbf{y}) },
    \end{gathered}
\end{equation}
where $\omega_{\mathbf{p}}=\sqrt{|\mathbf{p}|^2+M^2}$. The coincidence point limit for $\Delta^{\nmink^4}_+(x,\iota_z y)$ reads
\begin{equation}
    \label{eq:coincidencelimit_reflect2points}
    \lim_{x\rightarrow y}\Delta^{\nmink^4}_+(x,\iota_z y)=\frac{1}{(2\pi)^{3}} \int_{\re^3} \frac{d^3 \mathbf{p}}{2 \omega_{\mathbf{p}}} e^{2i p_z \, z}=\frac{M}{8\pi^2 z}K_1(2zM),
\end{equation}
where $K_1$ is the modified Bessel function of the second kind, see \cite[\S 10.25]{nist_bessel}. 
As a consequence the flow equation becomes
\begin{equation}
    \label{eq:scal_casim}
    \partial_k U_k(\phi)=\frac{k}{16\pi^2}M^2\log \left(\frac{M^2}{\mu^2}\right)\Big(1-\mathfrak{C}(z,M)\Big),
\end{equation}
where the \emph{\textbf{boundary contribution}} $\mathfrak{C}(z,M)$ is
\begin{equation}\label{Eq:factorC}
\mathfrak{C}(z,M)\defeq \frac{2}{Mz}K_1(2zM),
\end{equation}
which is plotted for convenience in Figure \ref{fig:cterm}.
\begin{remark}
 Comparing $\mathfrak{C}(z,M)$ with $\mathfrak{B}(z,M)$, obtained in the previous section in Equation \eqref{Eq:factorB}, and plotted in Figure \ref{fig:bterm}, we see that also $\mathfrak{C}(z,M)$ rapidly approaches $0$ far from the boundary. Yet, contrarily to the analysis in Section \ref{sec:full_sub}, this decay is more rapid, yielding a contribution comparable to $0$ already at $z\approx 1$. Moreover, differently from $\mathfrak{B}(z,M)$, the new term diverges at $z=0$. This is a sign that near the boundary, the minimal Hadamard subtraction leads to divergences, as one could expect from the microlocal behaviour of $\Delta_{+,D}^{\mathbb{H}^4}$ as per Equation \eqref{Eq: WF-half-Minks}. This codifies that singularities are reflected at $z=0$ and, therefore, $\lim_{x\rightarrow y}H^{\mathbb{M}^4}(x,\iota_z(y))$ becomes singular thereat.
\end{remark}

Following the same line of reasoning used in Section \ref{sec:full_sub}, we can derive the $\beta$-functions for the coupling constants by considering higher order functional derivatives of Equation \eqref{eq:scal_casim} with respect to $\phi$. With the same rescaling of the couplings as in Equation \eqref{eq:rescaling} and with the same choice of the renormalization scale $\mu=k$, we obtain the scaling for the dimensionless coupling constants in the minimal Hadamard subtraction scheme:
\begin{equation}
    \label{eq:betafunc_casim}
     \left\{
     \begin{gathered}
     \begin{aligned}
       &k\partial_k \widetilde{m}^2_k=-2\widetilde{m}_k^2+\frac{\widetilde{\lambda}_k}{16\pi^2}\left[\left(1-\mathfrak{C}\left(\widetilde{z},\sqrt{1+\widetilde{m}_k^2}\right)\right)\left(\log(\widetilde{m}_k^2+1)+1\right)\right.+\\&\quad\quad\quad\quad+\left.(1+\widetilde{m}_k^2)\log(1+\widetilde{m}_k^2)\mathfrak{C}_2\left(\widetilde{z},\sqrt{1+\widetilde{m}_k^2} \right)\right]\\
       \quad \\
       &k\partial_k \widetilde{\lambda}_k=\frac{\lambda_k^2}{16\pi^2}\left[\frac{3}{1+\widetilde{m}_k^2}\left(1-\mathfrak{C}\left(\widetilde{z},\sqrt{1+\widetilde{m}_k^2}\right)\right)\right.-6\left(\log(\widetilde{m}_k^2+1)+1\right)\mathfrak{C}_2\left(\widetilde{z},\sqrt{1+\widetilde{m}_k^2}\right)+ \\ &\quad \quad\quad \quad \left.+(1+\widetilde{m}_k^2)\log(\widetilde{m}_k^2+1)\mathfrak{C}_4\left(\widetilde{z},\sqrt{1+\widetilde{m}_k^2}\right)\right]
       \end{aligned} 
        \end{gathered} 
        \right.
\end{equation}
where
\begin{equation}
	\mathfrak{C}_2(z)=\frac{2K_2\left(2 \sqrt{k^2+m^2} z\right)}{k^2+m^2}\;\mathrm{and}\;\mathfrak{C}_4(z)=\frac{-6zK_3\left(2 \sqrt{k^2+m^2} z\right)}{(k^2+m^2)^{3/2}}
\end{equation}
As in the previous case, from the asymptotic behaviour of $\mathfrak{C}(z,M)$ it is manifest that, far from the boundary, already at $z\approx 1$, the scaling behaviour is the same as the one on the whole Minkowski spacetime. As before, solving Equation \eqref{eq:scal_casim} analytically seems not feasible. We present the stream lines in the plane $m^2,\lambda$ using the method $\mathsf{StreamPlot}[\;]$ of Wolfram Mathematica. In Figures \ref{fig:scal_casim_dist-1}, \ref{fig:scal_casim_dist-5} the scaling is computed at $\widetilde{z}=10^{-1},10^{-5}$.
\begin{figure}[H]
\begin{minipage}[b]{0.45\linewidth}
      \centering
    \includegraphics[width=\textwidth]{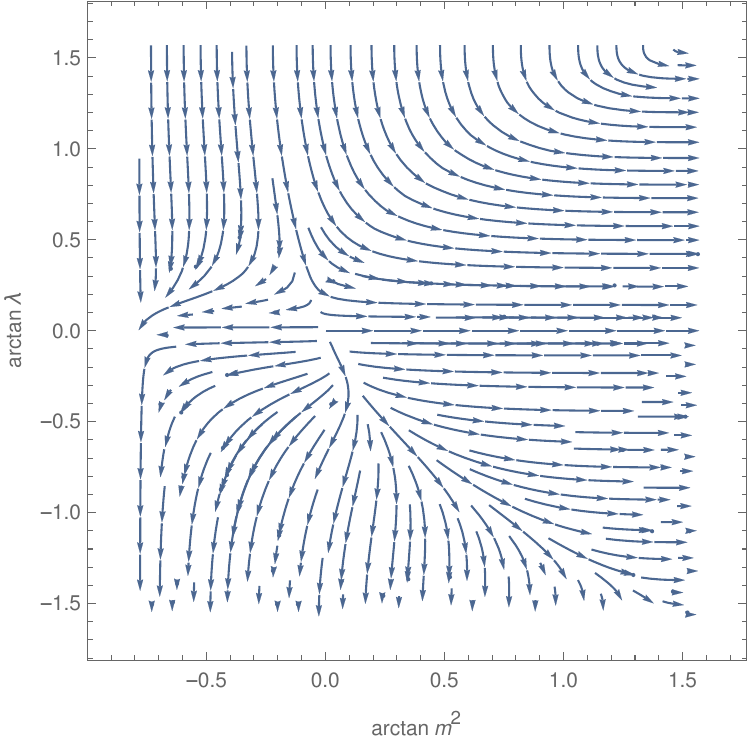}
    \caption{Scaling of the coupling constants on the Minkowski upper half-space. Dimensionless distance from the boundary: $\widetilde{z}=10^{-1}$}
    \label{fig:scal_casim_dist-1}
\end{minipage}
\hspace{0.5cm}
\begin{minipage}[b]{0.45\linewidth}
\centering
\includegraphics[width=\textwidth]{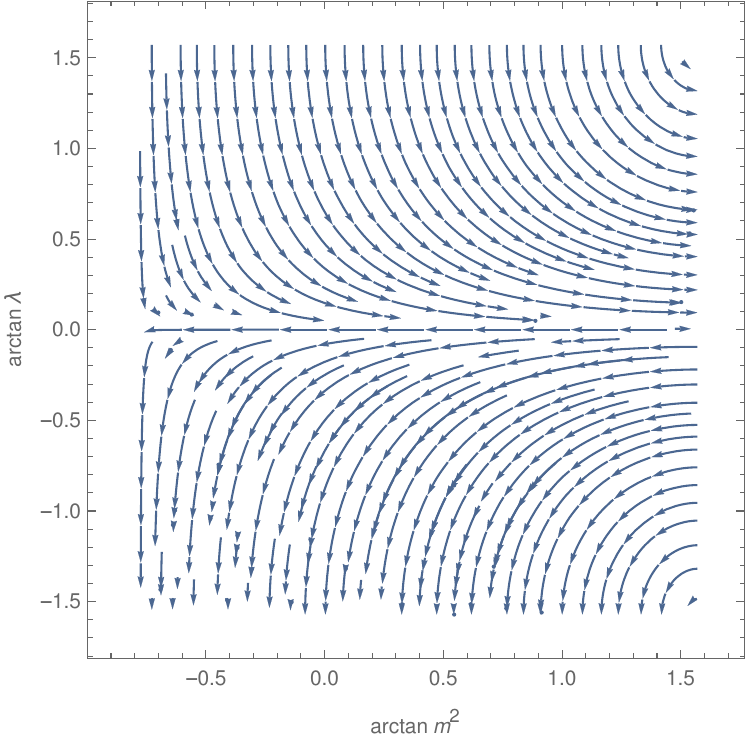}
    \caption{Scaling of the coupling constants on Minkowski upper half-space. Dimensionless distance from the boundary: $\widetilde{z}=10^{-5}$}
 \label{fig:scal_casim_dist-5}
\end{minipage}
\end{figure}
\begin{figure}[H]
\begin{minipage}[b]{0.45\linewidth}
      \centering
    \includegraphics[width=\textwidth]{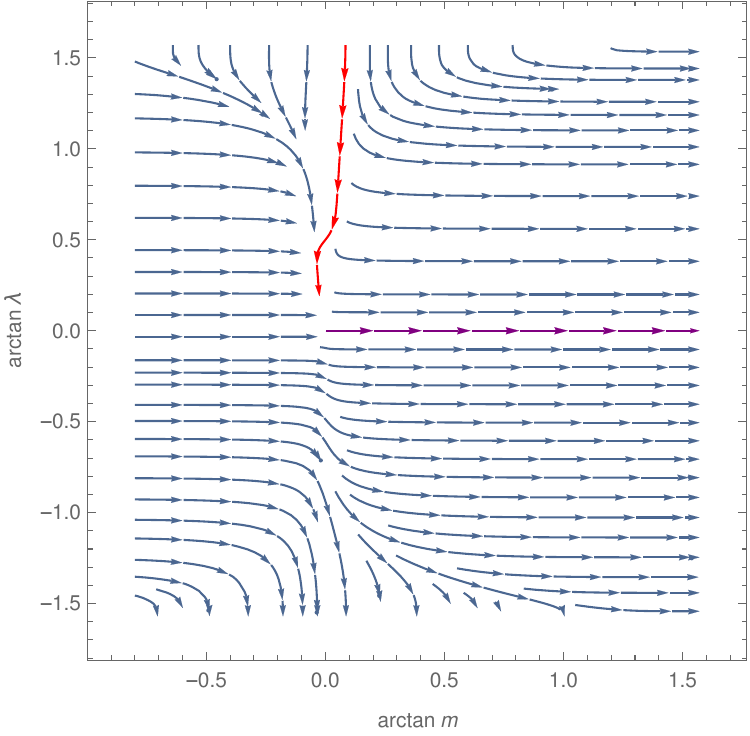}
    \caption{Flow on the whole Minkowski spacetime: (1) for the massless, interacting theory to the Gaussian, IR fixed point (red), (2) for the massive, non-interacting theory to the Gaussian, UV fixed point (violet). Observe that the plot is restricted to $m^2>0$.}
    \label{fig:fixedpoints_mink}
\end{minipage}
\hspace{0.5cm}
\begin{minipage}[b]{0.45\linewidth}
\centering
\includegraphics[width=\textwidth]{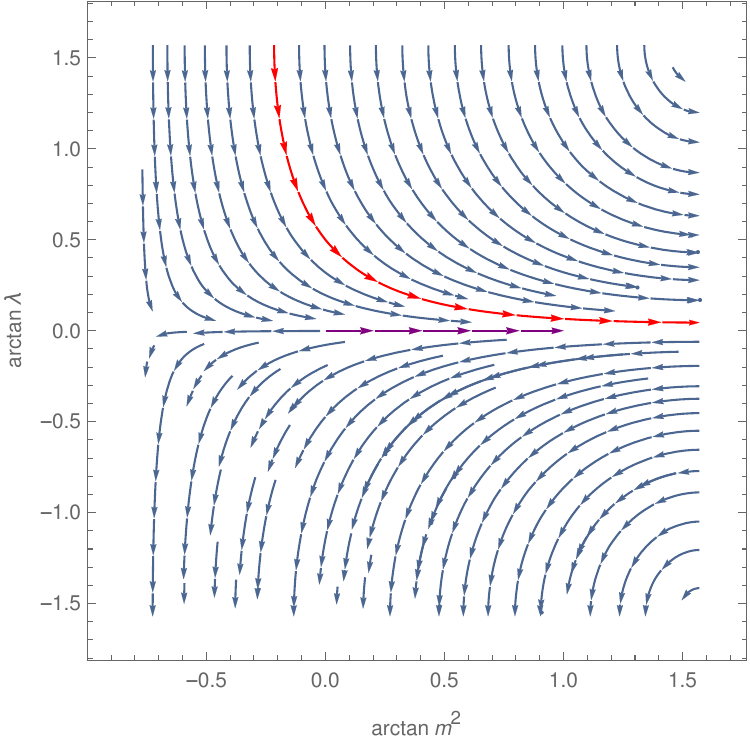}
    \caption{Flow on Minkowski upper half-space: (1) for the non-interacting massive theory to the Gaussian, UV fixed point (violet); (2) of an IR non-complete trajectory (red), that flows to $\lambda=0$ but infinite mass.}
 \label{fig:fixedpoints_casim}
\end{minipage}
\end{figure}

\begin{remark}
The comparison between the flow on $(\nmink^4,\eta)$ and the one on $(\nuppmink^4,\eta)$, at distance $\widetilde{z}=10^{-5}$, see Figures \ref{fig:fixedpoints_mink}, \ref{fig:fixedpoints_casim}, suggests that approaching the boundary does not affect the ultraviolet behaviour of the theory, as it can be still observed the Gaussian UV, fixed point, for the non-interacting theory. Yet, the infrared behaviour is significantly different, since the theory, under the renormalization group flow, seems to spontaneously acquire mass also on the critical surface, that it the locus where the relevant coupling constants are set to $0$. In other words it appears as if there in no trajectory for the massless theory to the Gaussian, infrared fixed point. Yet, this conclusion might be premature since, as shown in Figure \ref{fig:betafunc_casim}, near the boundary the $\beta$-function as a function of $\lambda$ becomes very large, eventually exploding at the boundary. Consequently, the system acquires a large interaction coupling constant near the boundary and the perturbative expansion on which our analysis is based does no longer apply. Once more we stress that this is a byproduct of the fact that, in the minimal Hadamard subtraction, we are not removing all global singularities of the underlying Hadamard two-point correlation function and these become manifest as we approach $z=0$.
\end{remark}

\begin{figure}[H]
\begin{minipage}[b]{0.45\linewidth}
     \centering
    \includegraphics[width=\textwidth]{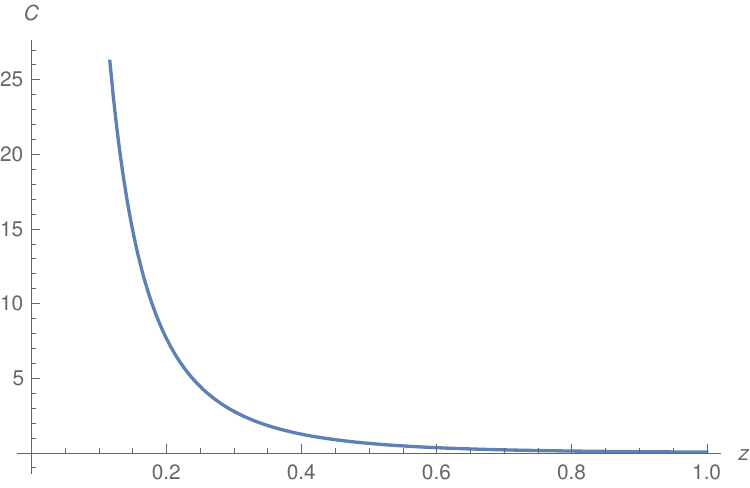}
    \caption{Plot of $\mathfrak{C}(\widetilde{z},\widetilde{M})$ at fixed value of $M$, with respect to $\widetilde{z}$.\\\quad\qquad\\ }
    \label{fig:cterm}
\end{minipage}
\hspace{0.5cm}
\begin{minipage}[b]{0.45\linewidth}
\centering
\includegraphics[width=\textwidth]{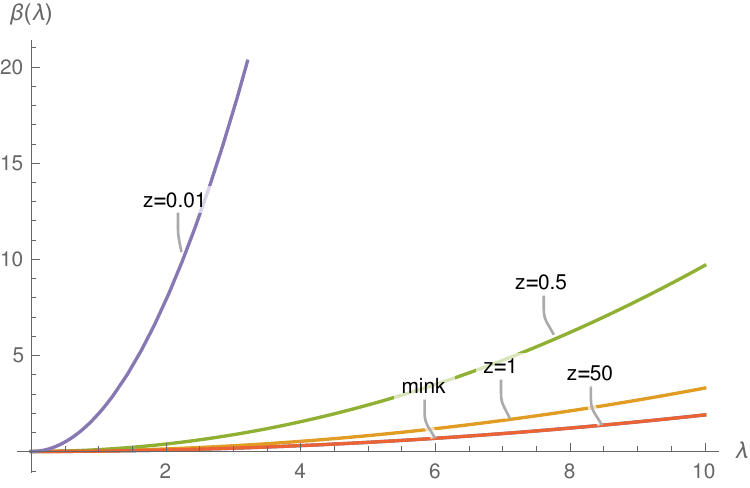}
    \caption{Comparison between the $\beta$-functions of $\lambda$, in the massless case, at different $\widetilde{z}$. Note that for small $\widetilde{z}$ the $\beta$-function explodes.}\label{fig:betafunc_casim}
\end{minipage}
\end{figure}
\subsection{Neumann's boundary conditions}\label{Sec: Neumann flow}
We shall now study the Lorentzian Wetterich equation on $\nuppmink^4$ when the boundary conditions at $z=0$ are of Neumann type. We recall that the relevant 2-point correlation function individuated in Equation \eqref{eq:2pointfunct_uppmink_methodimage} is
$$\Delta^{\uppmink}_{+,N}(x,y)\defeq \Delta^{\mink}_+(x,y)+ \Delta^{\mink}_+(\iota_z(x), y),$$
that amounts to a change of sign with respect to the Dirichlet case. The flow equations become
\begin{subequations}
\begin{equation}
	\label{scaling_in4Neumann}
	\partial_k U_k(\phi)=\frac{k}{16\pi^2}M^2\log \left(\frac{M^2}{\mu^2}\right)\Big(1+\mathfrak{B}(z,M)\Big),
\end{equation}
\begin{equation}
	\label{scaling_in4CasimirNeumann}
	\partial_k U_k(\phi)=\frac{k}{16\pi^2}M^2\log \left(\frac{M^2}{\mu^2}\right)\Big(1+\mathfrak{C}(z,M)\Big),
\end{equation}
\end{subequations}
where Equation \eqref{scaling_in4Neumann} corresponds to the full Hadamard subtraction, while Equation \eqref{scaling_in4CasimirNeumann} to the minimal one. The functions $\mathfrak{B}(z,M)$ and $\mathfrak{C}(z,M)$ are the boundary contributions respectively as per Equation \eqref{Eq:factorB} and \eqref{Eq:factorC}.

From these, we obtain the following scaling for the dimensionless coupling constants with the \emph{full Hadamard subtraction}:
\begin{equation}
	\label{eq:scaling_couplconst_dimensionless_choosemu_neumann}
	\left\{
	\begin{gathered}
		\begin{aligned}
			&k\partial_k \widetilde{m}^2_k=-2\widetilde{m}_k^2+\frac{\widetilde{\lambda}_k}{16\pi^2}\left[\left(1+\mathfrak{B}\left(\widetilde{z},\sqrt{1+\widetilde{m}_k^2}\right)\right)\left(\log(\widetilde{m}_k^2+1)+1\right)\right.+\\&\quad\quad\quad\quad-\left.(1+\widetilde{m}_k^2)\log(1+\widetilde{m}_k^2)\mathfrak{B}_2\left(\widetilde{z},\sqrt{1+\widetilde{m}_k^2} \right)\right]\\
			\quad \\
			&k\partial_k \widetilde{\lambda}_k=\frac{\lambda_k^2}{16\pi^2}\left[\frac{3}{1+\widetilde{m}_k^2}\left(1+\mathfrak{B}\left(\widetilde{z},\sqrt{1+\widetilde{m}_k^2}\right)\right)+\right.6\left(\log(\widetilde{m}_k^2+1)+1\right)\mathfrak{B}_2\left(\widetilde{z},\sqrt{1+\widetilde{m}_k^2}\right)+ \\ &\quad \quad\quad \quad -\left.(1+\widetilde{m}_k^2)\log(\widetilde{m}_k^2+1)\mathfrak{B}_4\left(\widetilde{z},\sqrt{1+\widetilde{m}_k^2}\right)\right],
		\end{aligned} 
	\end{gathered} 
	\right.
\end{equation}
and with respect to the \emph{minimal Hadamard subtraction}:
\begin{equation}
	\label{eq:scaling_couplconst_dimensionless_choosemu_casimir_neumann}
	\left\{
	\begin{gathered}
		\begin{aligned}
			&k\partial_k \widetilde{m}^2_k=-2\widetilde{m}_k^2+\frac{\widetilde{\lambda}_k}{16\pi^2}\left[\left(1+\mathfrak{C}\left(\widetilde{z},\sqrt{1+\widetilde{m}_k^2}\right)\right)\left(\log(\widetilde{m}_k^2+1)+1\right)\right.+\\&\quad\quad\quad\quad-\left.(1+\widetilde{m}_k^2)\log(1+\widetilde{m}_k^2)\mathfrak{C}_2\left(\widetilde{z},\sqrt{1+\widetilde{m}_k^2} \right)\right]\\
			\quad \\
			&k\partial_k \widetilde{\lambda}_k=\frac{\lambda_k^2}{16\pi^2}\left[\frac{3}{1+\widetilde{m}_k^2}\left(1+\mathfrak{C}\left(\widetilde{z},\sqrt{1+\widetilde{m}_k^2}\right)\right)+ \right.6\left(\log(\widetilde{m}_k^2+1)+1\right)\mathfrak{C}_2\left(\widetilde{z},\sqrt{1+\widetilde{m}_k^2}\right)+\\ &\quad \quad\quad \quad \left.-(1+\widetilde{m}_k^2)\log(\widetilde{m}_k^2+1)\mathfrak{C}_4\left(\widetilde{z},\sqrt{1+\widetilde{m}_k^2}\right)\right].
		\end{aligned} 
	\end{gathered} 
	\right.
\end{equation}
As discussed in Section \ref{sec:full_sub} the corrections provided by the term $\mathfrak{B}$ are negligible, hence we shall focus on the second case. We present in Figures \ref{fig:scal_neumann_dist-1}, \ref{fig:scal_neumann_dist-5} the scaling computed at dimensionless distance $\widetilde{z}=10^{-1}$ and $\widetilde{z}=10^{-5}$ employing the minimal Hadamard subtraction.
\begin{figure}[H]
	\begin{minipage}[b]{0.45\linewidth}
		\centering
		\includegraphics[width=\textwidth]{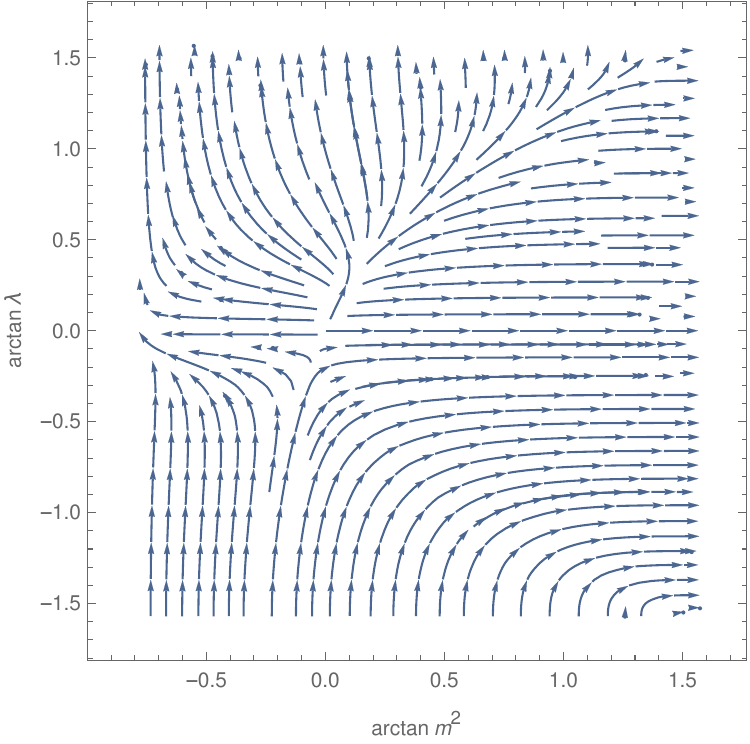}
		\caption{Scaling of the coupling constants on $\nuppmink^4$ with Neumann boundary conditions. Dimensionless distance from the boundary: $\widetilde{z}=10^{-1}$}
		\label{fig:scal_neumann_dist-1}
	\end{minipage}
	\hspace{0.5cm}
	\begin{minipage}[b]{0.45\linewidth}
		\centering
		\includegraphics[width=\textwidth]{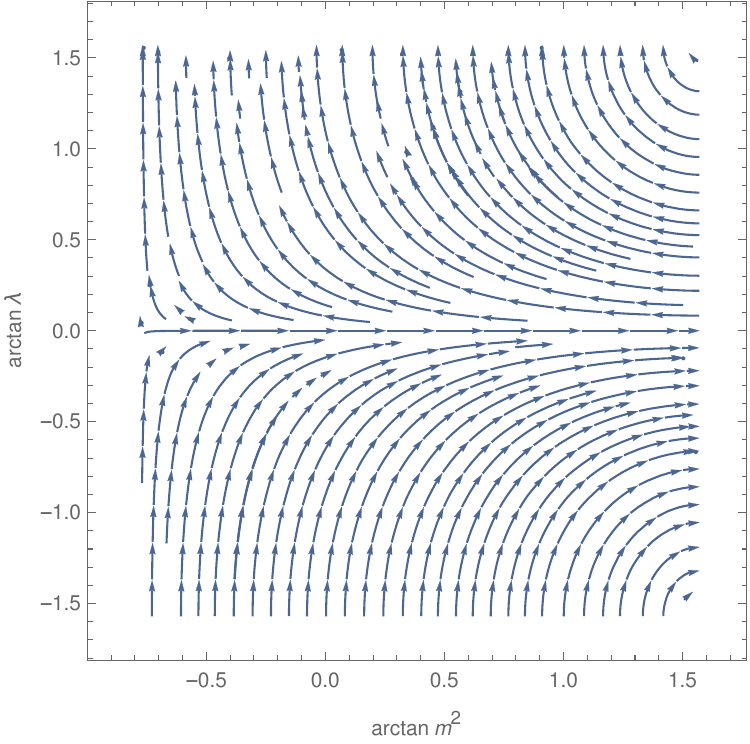}
		\caption{Scaling of the coupling constants on $\nuppmink^4$ with Neumann boundary conditions. Dimensionless distance from the boundary: $\widetilde{z}=10^{-5}$}
		\label{fig:scal_neumann_dist-5}
	\end{minipage}
\end{figure}
The flows for the coupling constants shows that there is a remarkable difference from the Dirichlet counterpart. Indeed, by direct inspection is it possible to see that the UV (IR) behaviour of the flow in Figure \ref{fig:scal_neumann_dist-5} corresponds to the IR (UV) one in Figure \ref{fig:scal_casim_dist-5}, as the direction of the arrows is reversed.
The same applies to the flow in Figure \ref{fig:scal_neumann_dist-1}, that possesses an ultraviolet, Gaussian fixed point at $(\widetilde{m}^{*\,2},\widetilde{\lambda}^*)=(0,0)$. 

As a consequence the system is \emph{asymptotically free}, as the interaction coupling constant vanishes in the high energy limit. This statement is strengthened by the analysis in Figure \ref{fig:comparisonbeta_neumann} providing a comparison of the $\beta$-functions at different values of $\widetilde{z}$. It shows a change of sign in correspondence of $\widetilde{z}\approx 0.897$. At this value the $\beta$-function is identically zero, namely at $\widetilde{z}\approx 0.9$ the theory enjoys the peculiar property of being at a fixed point for every choice of initial conditions. In Figure \ref{fig:scalinglambda_neumann} we present the scaling of $\widetilde{\lambda}$ with respect to the energy scale $k$ at different values of $\widetilde{z}$, obtained with the method $\mathsf{NDSolve}[\;]$ from Wolfram Mathematica, with the initial condition $\widetilde{\lambda}(k=1)=1$. Note that for $\widetilde{z}>0.9$ the interaction coupling constant approaches the Gaussian, infrared fixed point at $k=0$, while for $\widetilde{z}<0.9$ it diverges in the low energy regime. The system is asymptotically free for $\widetilde{z}<0.9$ as $\widetilde{\lambda}$ tends to zero logarithmically as a function of the energy.
\begin{figure}[H]
	\begin{minipage}[b]{0.45\linewidth}
		\centering
		\includegraphics[width=\textwidth]{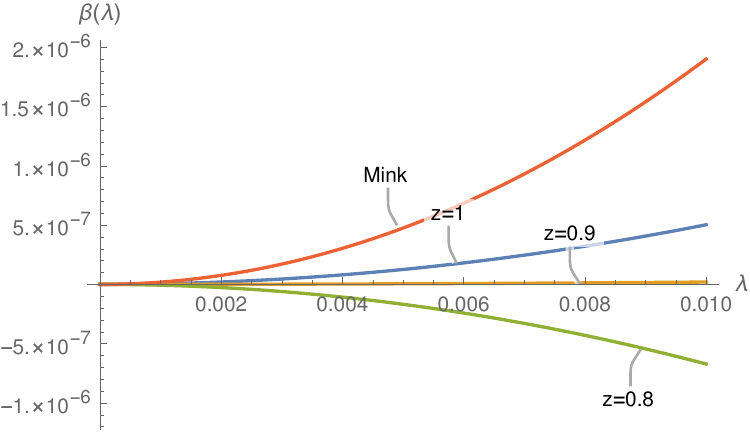}
		\caption{Comparison between the $\beta$-functions of $\lambda$, in the massless case, at different $\widetilde{z}$. Note that at $\widetilde{z}\approx 0.9$ the $\beta$-function changes sign.}
		\label{fig:comparisonbeta_neumann}
	\end{minipage}
	\hspace{0.5cm}
	\begin{minipage}[b]{0.45\linewidth}
		\centering
		\includegraphics[width=\textwidth]{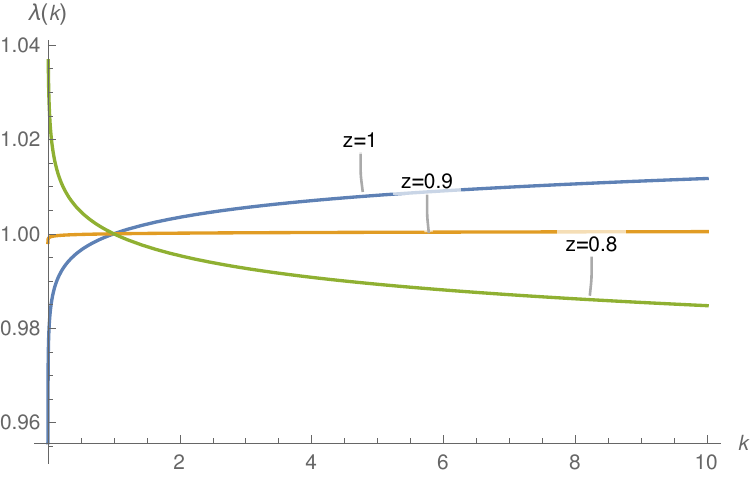}
		\caption{Scaling of the interaction coupling constant $\lambda$ with respect to the energy scale $k$, at different $\widetilde{z}$. Note that for $\widetilde{z}\lessapprox 0.9$ the system becomes asymptotically free.}
		\label{fig:scalinglambda_neumann}
	\end{minipage}
\end{figure}
We stress that for $\widetilde{z}<0.9$ the perturbative approach does no longer apply at low energies, as the system becomes strongly interacting.
\subsection{Flow on Poincaré patch of Anti-de Sitter}
\label{sec:pads_scal}
In this section we study the Lorentzian Wetterich equation in the case where the underlying Lorentzian manifold is the Poincaré patch of Anti-de Sitter $(\pads,g)$, setting also in this case for convenience and for physical relevance $d=3$. Let $I$ be the interacting Klein-Gordon action
$$I(\phi)=-\int_{\npads_4}d\mu_g\,\left(\frac{1}{2}\nabla_a\phi\nabla^a\phi +\frac{m^2}{2}\phi^2+\frac{\xi}{2}R\phi^2+\lambda\frac{\phi^4}{4!}\right)f(x),$$
where $m^2$ is the squared mass of the field, $\nabla$ is the covariant derivative built from the metric, $\xi$ is the coupling to the scalar curvature $R$ while $\lambda$ is the interaction coupling constant. In addition, we assign at $\partial\npads_4$ Dirichlet boundary conditions which are allowed for every values of $m$ and $\xi$ as sketched in Section \ref{Sec: QFT with timelike boundaries}. 

To study the behaviour of the Lorentzian Wetterich equation on the Poincaré patch of Anti-de Sitter, we apply once more the Local Potential Approximation (LPA), discussed in Section \ref{sec:wett}. The first step consists of formulating an \emph{ansatz} for the regularized effective action. Our choice, similar to the one on Minkowski upper-half space, is

\begin{equation}
    \label{lpa_pads}
    \Gamma_k(\phi)=-\int_{{\npads}_4}d\mu_g\,\frac{1}{2}\nabla_a\phi\nabla^a\phi+U_k(\phi),
\end{equation}
where $U_k(\phi)\defeq m^2_k\frac{\phi^2}{2}+\lambda_k \frac{\phi^4}{4!}$ is the local potential. 
Subsequently, following the LPA, for a choice of a Hadamard, on shell 2-point correlation function $\Delta_+\in\mathcal{D}'(\npads_4\times\npads_4)$, the Wetterich equation takes the form
\begin{equation}
    \label{eq:had_subtraction_pads}
    \partial_k \Gamma_k(\phi)=-\int_{\npads_4} d\mu_g\,\partial_kU_k(\phi)=-\int_{{\npads_4}} d\mu_g(y) \lim_{x\rightarrow y} k\left(\Delta^{\npads_4}_{S,D}(x,y)-H^{\npads_4}_{(D)}(x,y)\right),
\end{equation}
where, adapting to the case in hand the argument used in \cite[Section 6.1.3]{mother} $\Delta^{\npads_4}_{S,D}(x,y)$ is the symmetric part of the two-point correlation function for a massive, real scalar field on $\npads_4$ with Dirichlet boundary conditions as per Equation \eqref{eq:twopointfunct_pads_2sol}. Here the parameter $\nu$ is replaced by
\begin{equation*}
    \nu_k \defeq \frac{1}{2}\sqrt{1+4l^2 \overline{M}^2},\quad \overline{M}^2\defeq m_k^2+\lambda_k\frac{\phi^2_{cl}}{2}+k^2+\left(\xi-\frac{1}{6}\right)R.
\end{equation*}

In Equation \eqref{eq:had_subtraction_pads} $H^{\npads_4}_{(D)}(x,y)$ is either the Hadamard parametrix $H^{\npads_4}_{(D)}(x,y)$ associated to $\Delta^{\npads_4}_{+,D}$ as in Equation \eqref{eq:hadpar_nonrefl_pads} with the replacement $\nu\mapsto \nu_k$ or $H^{\npads_4}_{D}(x,y)$ as per Equation \eqref{eq:hadpar_nonrefl_pads_D}. These options correspond in the case in hand to the minimal and full Hadamard subtraction considered in the previous section. 

In the following,  we restrict for simplicity to the conformal coupling, \iee $\xi=1/6$, though we shall not lose in generality of the results as we will comment. The scaling equation on $(\npads_4,g)$ using $H^{\npads_4}_D(x,y)$ reads
\begin{equation}
    \label{eq:scalingpads}
    \begin{gathered}
         \partial_k U_k(\phi)=\frac{k}{4l^2 \Gamma(\nu_k-1/2)}\Big[(1+2\nu_k)\Gamma(\nu_k-1/2)+2\log(2l\mu)+\\-2\Gamma(3/2+\nu_k)[-1+2\gamma+\psi(\nu_k+1/2)+\psi(\nu_k+3/2)]\Big],
    \end{gathered}  
\end{equation}
where $\psi$ is the digamma function, $\gamma$ is the Euler-Mascheroni constant and where we used the fact that, in the coinciding point limit, $\Delta^{\npads_4}_{+,D}(x,y)\equiv \Delta^{\npads_4}_{S,D}(x,y)$.
We stress the presence, in Equation \eqref{eq:scalingpads}, of an arbitrary energy scale $\mu$ in the logarithm, that corresponds to a renormalization energy scale. Differently from the flat spacetime scenario, this ambiguity can be fixed by the geometry. As a matter of fact $(\npads_4,g)$ possesses a characteristic length scale $l$, and a natural choice is $\mu\equiv (2l)^{-1}$, which allows to discard the logarithmic term. 
\begin{remark}
    \label{rmk:fullsub_pads}
    The scaling in Equation \eqref{eq:scalingpads} has been obtained employing the minimal Hadamard subtraction. It was shown in Section \ref{sec:full_sub} that the choice of a different prescription, particularly the full Hadamard subtraction, might yield different physical results. We highlight that, differently from the case of half Minkowski spacetime, this difference seems not to occur on $(\npads_4,g)$. Indeed, if one considers the \emph{reflected component} of the Hadamard parametrix of Equation \eqref{eq:hadpar_nonrefl_pads_D}, namely $H^{\npads_4}(\iota_z(x),y)$, and if one computes its coinciding point limit, one finds that
    \begin{equation*}
         \lim_{x\rightarrow y} H^{\npads_4}(\iota_z(x),y)=\lim_{\epsilon\rightarrow 0^+}\frac{1}{\sigma_\epsilon(\iota_z(x),y)}+\left(\nu^2-\frac{1}{4}\right)\log \left(\frac{\sigma_\epsilon(\iota_z(x),y)}{2l^2}\right)=\frac{1}{2l^2},
    \end{equation*}
    where we used the expression for the geodesic distance in Equation \eqref{eq:cordaldist_pads}. As a consequence Equation \eqref{eq:scalingpads} is shifted by a constant, yielding the same $\beta$-functions as in the minimal Hadamard subtraction scenario. 

\end{remark}
From Equation \eqref{eq:scalingpads}, by considering higher order derivatives with respect to $\phi$, it is possible to extract the $\beta$-functions, as discussed in Section \ref{sec:flowhalf}. These read
\begin{equation*}
   \begin{cases}
    \beta_{m^2}\defeq\partial_k U_k^{(2)}(\phi=0)=\partial_k m_k^2,\\
    \beta_\lambda\defeq\partial_k U_k^{(4)}(\phi=0)=\partial_k \lambda_k.
\end{cases} 
\end{equation*}
As before, the system obtained is not autonomous. We proceed by rewriting it in terms of the \emph{dimensionless coupling constants}. These can be defined with the aid of the characteristic length scale $l^2$, as
\begin{equation*}
   \begin{cases}
    m_k^2\,\longmapsto\,\widetilde{m}_k^2\defeq l^2 m_k^2,\\
    \lambda_k\,\longmapsto\,\widetilde{\lambda}_k\defeq  k^2l^2\lambda_k.
\end{cases} 
\end{equation*}
Under this choice of rescaling, the $\beta$-functions become
\begin{equation}
   \begin{cases}
    \partial_k \widetilde{m}_k^2=l^2 \partial_k m_k^2,\\
    \partial_k \widetilde{\lambda}_k=l^2(2k\lambda_k+k^2\partial_k \lambda_k).
\end{cases} 
\end{equation}
Written explicitly, we obtain the scaling equations for the dimensionless coupling constants:
\begin{equation}
    \label{eq:betafunct_pads}
    \left\{
    \begin{gathered}
    \begin{aligned}
    &k\partial_k \widetilde{m}_k^2=\frac{\widetilde{\lambda}_k}{4\tilde{\nu}} \left[1-\left(\tilde{\nu}^2-\frac{1}{4}\right) \left([-1+2\gamma+\psi(1/2+\tilde{\nu})]\psi(3/2+\tilde{\nu})+\psi^2(3/2+\tilde{\nu})+\right.\right.\\ &\quad\quad\quad\quad\quad\left.-\psi(-1/2+\tilde{\nu})[-1+2\gamma+\psi(1/2+\tilde{\nu})+\psi(3/2+\tilde{\nu})]+\psi'(1/2+\tilde{\nu})+\psi'(3/2+\tilde{\nu})\right)\Big]\\
    \quad\\
    &k\partial_k \widetilde{\lambda}_k=2\widetilde{\lambda}_k+\frac{3 \widetilde{\lambda}_k^2}{16\tilde{\nu}^3}  \Bigg[-2-\left(\tilde{\nu}^2-\frac{1}{4}\right) \Big[\tilde{\nu}\; \psi^3(3/2+\tilde{\nu})+\big[2 \gamma  \tilde{\nu}-\tilde{\nu}+\tilde{\nu} \psi(1/2+\tilde{\nu})-2\big] \;\psi^2(3/2+\tilde{\nu})+\\
    &+\big[-2 \psi(1/2+\tilde{\nu})+2 \tilde{\nu} \psi'(1/2+\tilde{\nu})+3 \tilde{\nu} \psi'(3/2+\tilde{\nu})-\tilde{\nu} \psi'(-1/2+\tilde{\nu})-4 \gamma +2\big] \;\psi(3/2+\tilde{\nu})+\\
    &+\tilde{\nu} \left(\psi(1/2+\tilde{\nu})+\psi(3/2+\tilde{\nu})+2 \gamma -1\right) \psi^2(-1/2+\tilde{\nu})-2 \psi'(1/2+\tilde{\nu})+2 \gamma  \tilde{\nu} \psi'(3/2+\tilde{\nu})\\
    &-\tilde{\nu} \psi'(3/2+\tilde{\nu})+\tilde{\nu} \psi(1/2+\tilde{\nu}) \psi'(3/2+\tilde{\nu})-2 \psi'(3/2+\tilde{\nu})-2 \psi(-1/2+\tilde{\nu}) \Big[\tilde{\nu} \psi^2(3/2+\tilde{\nu})+\\
    &+(2 \gamma  \tilde{\nu}-\tilde{\nu}-1) \psi(3/2+\tilde{\nu})+\psi(1/2+\tilde{\nu}) \left(\tilde{\nu} \psi(3/2+\tilde{\nu})-1\right)+\tilde{\nu} \psi'(1/2+\tilde{\nu})+\tilde{\nu} \psi'(3/2+\tilde{\nu})+\\
    &-2 \gamma +1\Big]-2 \gamma  \tilde{\nu} \psi'(-1/2+\tilde{\nu})+\tilde{\nu} \psi'(-1/2+\tilde{\nu})-\tilde{\nu} \psi(1/2+\tilde{\nu}) \psi'(-1/2+\tilde{\nu})+\tilde{\nu} \psi''(1/2+\tilde{\nu})\\
    &+\tilde{\nu} \psi''(3/2+\tilde{\nu})\Big]\Bigg]
    \end{aligned}
    \end{gathered} \right.
\end{equation}
where $\tilde{\nu}\defeq \frac{1}{2}\sqrt{1+4k^2l^2+4\widetilde{m}_k^2}$, while $\psi,\,\psi',\,\psi''$ are the digamma function and its derivatives.

The system is not yet autonomous, as the independent variable appears explicitly in the term $\tilde{\nu}$. In the following remark we present physical considerations that allow to discard it and which are in spirit inspired by the analysis in \cite{mother} on the Wetterich equation on de Sitter spacetime.
\begin{remark}[Small $k^2l^2$ limit and the large deflation approximation]\label{Rem: Large Deflation}
Let us consider the heuristic picture where, performing an experiment at an energy scale $k$, we are probing lengths of order $\rho\defeq k^{-1}$. Consequently, we are interested in estimate the magnitude of $l^2/\rho^2$. 

If we wish to compare the scaling of the coupling constants on $(\npads_4,g)$ with the one on flat Minkowski spacetime, we expect that large differences lie in the infrared spectrum, since large values of $\rho$ allow to explore a broader volume of the underlying manifold, hence discriminating between flat and curved spacetime. Accordingly, it is desirable to work in the approximation $l^2/\rho^2\approx 0$, to which we refer as \emph{large deflation approximation}. Indeed, as the cosmological constant of Anti-de Sitter spacetime in 4-dimension is $\Lambda=-6l^{-2}$, for small values of $l$, $\Lambda$ is (negatively) large. Under the large deflation approximation, we can discard the contribution from $k^2l^2$ in $\tilde{\nu}$ and, therefore, the system of ODEs in Equation \eqref{eq:betafunct_pads} becomes autonomous.
\end{remark}

 We can start the resolution of the autonomous system, obtained from Equation \eqref{eq:betafunct_pads} recalling the approximation $\tilde{\nu}\defeq \frac{1}{2}\sqrt{1+4\widetilde{m}_k^2}$. Since looking for an analytic form of the solution seems a daunting task, we proceed with a numerical analysis. Yet, to start with, we observe that, by direct inspection, we can individuate a \emph{line of fixed points} in correspondence to $\widetilde{\lambda}=0$, since thereon the $\beta$-functions vanish. 

In Figure \ref{fig:zerobeta_pads} we show the zero-level set of the $\beta$-functions of $\widetilde{m}^2$ and $\widetilde{\lambda}$. Note that, apart from the line of fixed points that we have observed, there is another fixed point for the renormalization group flow, where both $\beta$-functions vanish. Notice that $\widetilde{m}^2>-1/4$ in agreement with the Breitenlohner-Freedman bound \cite{Breitenlohner:1982jf}.

By using the method $\mathsf{StreamPlot}[\;]$ from Wolfram Mathematica, we can study the dimensionless version of Equation \eqref{eq:betafunct_pads} by plotting the stream lines for the system of ODEs, depicted in Figure \ref{fig:fullscal_ads}. As on Minkowski upper half-space, we choose to plot with respect to the tangents of $\widetilde{m}^2,\widetilde{\lambda}$. Moreover, the arrows point in the infrared direction. In Figure \ref{fig:scal_pads} we restrict the plot to the physical region $\widetilde{\lambda}>0$. Note that the fixed point depicted in purple, previously observed in Figure \ref{fig:zerobeta_pads}, is an ultraviolet, stable fixed point for the renormalization group flow. A numerical estimate for it yields the value $(\widetilde{m}^{*},\widetilde{\lambda}^*)\approx (0.486,0.469)$. 

We emphasise that the interaction coupling constant near the UV fixed point is small, confirming \emph{a posteriori} that a perturbative analysis of the Wetterich equation near the UV fixed point is legit. The result obtained in Figure \ref{fig:scal_pads} is noteworthy, as it shows that, contrarily to the flat spacetime scenario, a class of complete renormalization group trajectories does exist.

\begin{figure}[H]
\begin{minipage}[b]{0.45\linewidth}
\centering
\includegraphics[width=\textwidth]{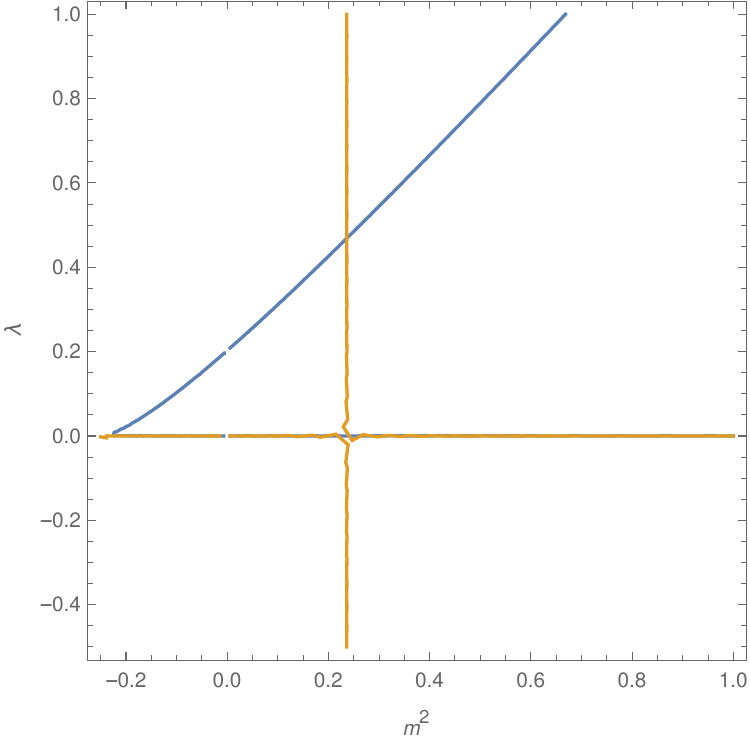}
    \caption{Zeros of the $\beta$-functions: in blue $\beta_{\tilde{m}^2}=0$, in orange $\beta_{\tilde{\lambda}}=0$. The line $\widetilde{\lambda}=0$ is a zero for both $\beta$-functions. We observe a non-trivial, fixed point at their intersection.}
 \label{fig:zerobeta_pads}
\end{minipage}
\hspace{0.5cm}
\begin{minipage}[b]{0.45\linewidth}
      \centering
    \includegraphics[width=\textwidth]{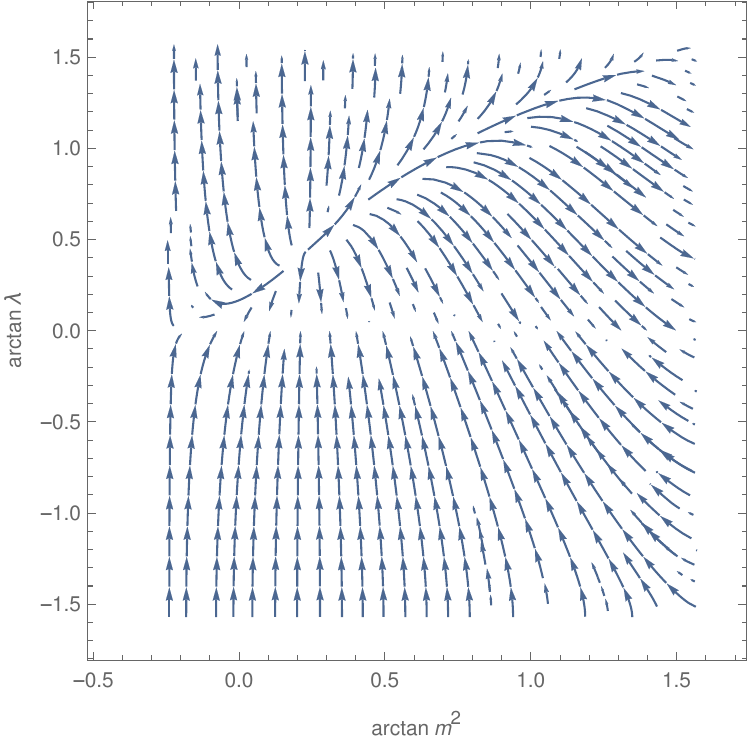}
    \caption{Scaling of the coupling constants on $(\npads_4,g)$, conformally coupled to the scalar curvature and with Dirichlet boundary conditions at the boundary.}
    \label{fig:fullscal_ads}
\end{minipage}
\end{figure}

\begin{figure}[H]
	\centering
	\includegraphics[width=0.7\textwidth]{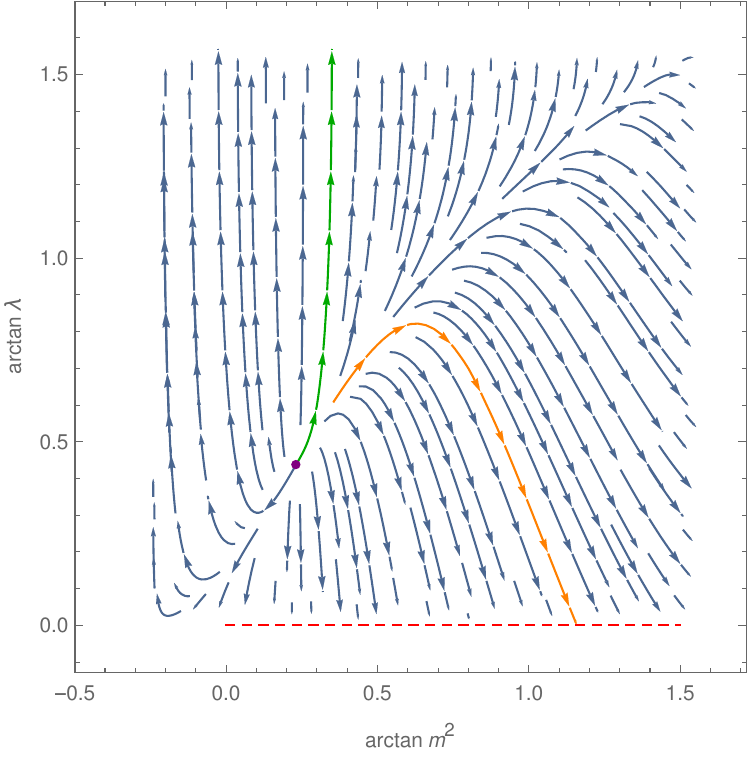}
	\caption{Scaling on $(\npads_4,g)$ with Dirichlet boundary conditions. Note the dashed line of infrared fixed points, depicted in red. The purple point at $(\widetilde{m}^*,\widetilde{\lambda}^*)\approx (0.486,0.469)$ is an ultraviolet, stable fixed point. A bifurcation-like behaviour separates the infrared complete trajectories (orange) from the infrared incomplete ones (green). Note that the orange line is complete both in the ultraviolet and in the infrared regime.}
	\label{fig:scal_pads}
\end{figure}

\begin{remark}
It is possible to observe a bifurcation which separates the infrared complete trajectories from those which are incomplete in the infrared regime. The former ones are depicted in Figure \ref{fig:ircomplete}. Note that since, at low energies, $m^2$ completely characterizes the behaviour, the classification presented in the previous section of $m^2$ and $\lambda$ as relevant and marginal couplings continues to hold.
The infrared behaviour is very different from the one obtained on Minkowski spacetime and on $\mathbb{H}^4$. It can be explained as a geometric effect of the negative curvature of $(\npads_4,g)$. As a matter of facts, this can be seen as a signature that, being negatively curved, Anti-de Sitter acts as a box, \emph{confining} the massless modes: as a consequence, the infrared behaviour is less singular. This heuristic picture explains also why small, negative values of $m^2$ are allowed, as the Breitenlohner-Freedman bound corresponds to $m^2>-1/4$. These modes, that on Minkowski spacetime are tachyonic hence unphysical, on Anti-de Sitter are regulated by the underlying geometry.  

Yet, one must always bear in mind that, in the ultraviolet regime, we are working under the approximation introduced in Remark \ref{Rem: Large Deflation} and therefore, before drawing a definitive conclusion on the behaviour of the flow equations in this instance, it is necessary to verify that the behaviour of the solutions is not qualitatively altered by the term discarded.

We observe that our results are in agreement with the literature and in particular, in \cite{confining_geo} it has been shown that, in the Euclidean signature, with Dirichlet boundary conditions and with a self-interaction $\lambda \varphi^3$, the low energy modes behave in a similar way as in the scenario considered in this paper.
\end{remark}

These results have been obtained under the assumption that the scalar theory is conformally coupled to the underlying geometry, namely $\xi=1/6$. This might raise the doubt that the UV and IR points observed are a consequence of this particular choice. As a matter of facts, the non-interacting Klein-Gordon action, conformally coupled, is by definition a fixed point for the renormalization group flow, and that could be the source of the line of infrared, fixed points at $\widetilde{\lambda}=0$. Nevertheless, we can show that this is not the case, as the choice $\xi=1/6$ merely simplifies the scaling without changing the overall behaviour. Indeed, let us include a non-conformal coupling $\xi$ in the parameter $\nu_k$, obtaining
$$\nu_k=\frac{1}{2}\sqrt{1+4l^2\left( m_k^2+\lambda_k\frac{\phi^2}{2}+k^2-\frac{12}{l^2}\left(\xi-\frac{1}{6}\right)\right)},$$
where we used that $R=-12/l^2$. Eventually this leads to the following change for $\tilde{\nu}$:
\begin{equation}
\label{eq:tildenu_nonconf_pads}
    \tilde{\nu}=\frac{1}{2}\sqrt{1+4\widetilde{m}_k^2-48\left(\xi-\frac{1}{6}\right)},
\end{equation}
where $\xi$ does not require any rescaling and where, as before, we have applied the large deflation approximation. Inserting $\tilde{\nu}$ in the system in Equation \eqref{eq:betafunct_pads}, we observe that $\lambda=0$ is still a line of fixed points. In Figure \ref{fig:scal_1-7} we show the scaling in the case where $\xi=2/15$, a value at which it is possible to see that the UV fixed point, marked in purple, lies at negative values of $m^2$. However, as it fulfills the Breitenlohner-Freedman bound, it is a physical admissible case.

In the case where the coupling $\xi$ is not the conformal one, we identify 2 special situations. The first one is the minimally coupled scenario, where $\xi=0$, that is represented in Figure \ref{fig:scaling_xi=0}. As it can be seen the UV fixed point, drawn in purple, lies in the unphysical region ruled out by the Breitenlohner-Freedman bound $m^2>-1/4$. Hence, when discussing the scalar quantum field on $(\npads_4,g)$ in the minimally coupled case, it must be considered an ultraviolet incomplete theory, while the infrared behaviour remains the same. We were unable to determine analytically the value $\overline{\xi}$ corresponding to an UV fixed point with $\widetilde{m}^{*\,2}=-1/4$: a numerical estimate for it is $\overline{\xi}\approx 0.127$. In this sense, $\overline{\xi}$ can be thought as the analogue of the Breitenlohner-Freedman bound at the level of the coupling to the scalar curvature, as the UV complete theories are those with $\xi>\overline{\xi}$. Yet, also in this instance, we have to keep in mind that we are working under the approximation introduced in Remark \ref{Rem: Large Deflation}.

The second particular case is the choice $\xi=3/16$, which corresponds to $\tilde{\nu}=\widetilde{m}_k$ in view of Equation \eqref{eq:tildenu_nonconf_pads}. As a consequence, as $\tilde{\nu}$ appears in the denominator of Equation \eqref{eq:betafunct_pads}, we must impose $\widetilde{m}\neq 0$. The scaling of the coupling constants in this scenario is depicted in Figure \ref{fig:scaling9-48}.

\begin{figure}[H]
	\begin{minipage}[b]{0.45\linewidth}
		\centering
		\includegraphics[width=\textwidth]{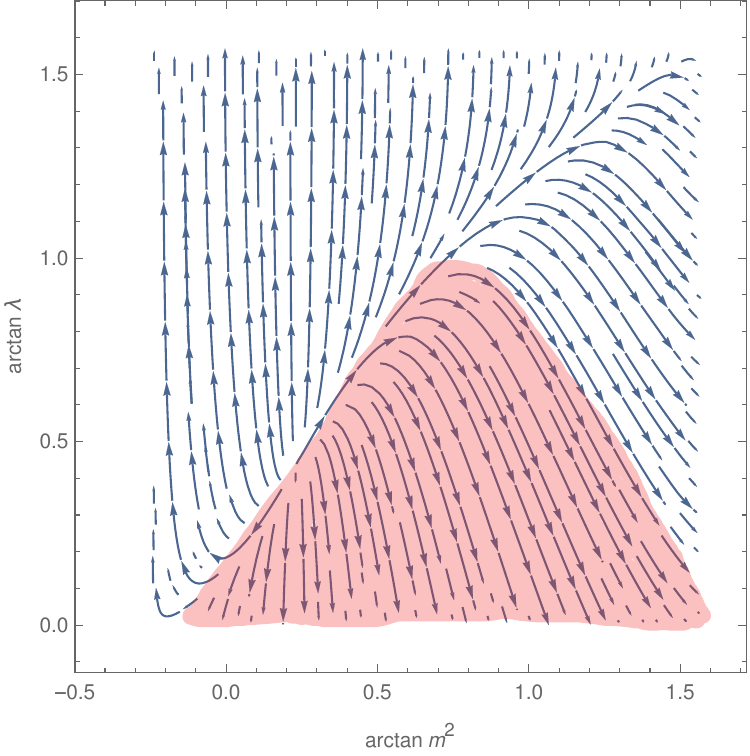}
		\caption{The class of infrared complete theories is highlighted in red. Note that they all share the same non-trivial UV fixed point.}
		\label{fig:ircomplete}
	\end{minipage}
	\hspace{0.5cm}
	\begin{minipage}[b]{0.45\linewidth}
		\centering
		\includegraphics[width=\textwidth]{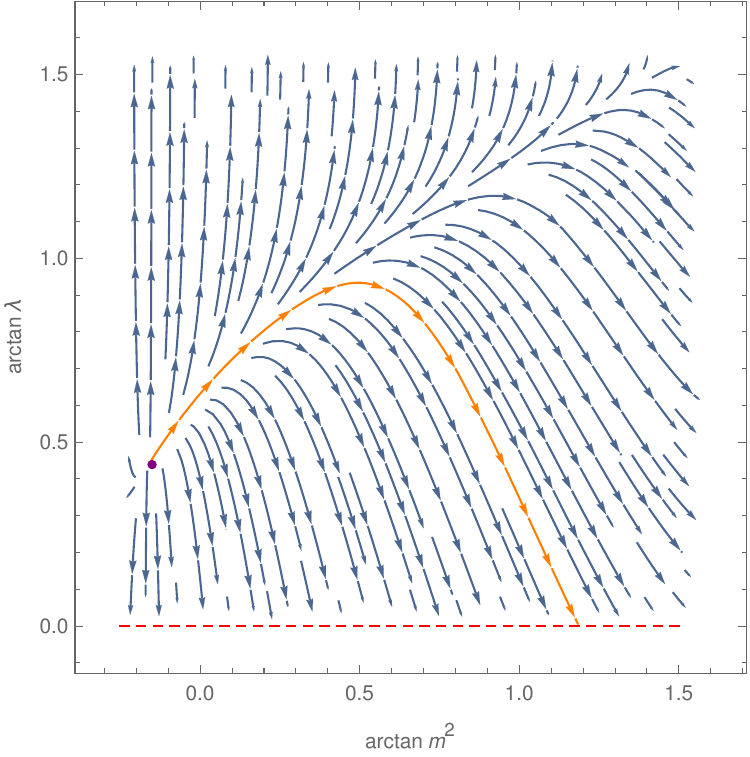}
		\caption{Scaling of the coupling constants on $\npads_4$ with $\xi=2/15$. \\}
		\label{fig:scal_1-7}
	\end{minipage}
\end{figure}

\begin{figure}[H]
\begin{minipage}[b]{0.45\linewidth}
\centering
\includegraphics[width=\textwidth]{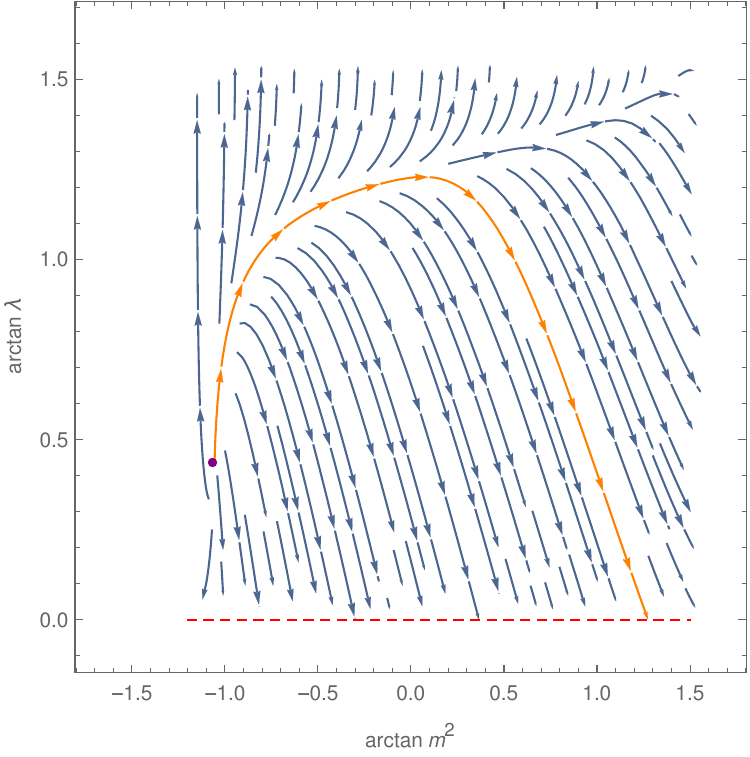}
    \caption{Scaling of the coupling constants on $\npads_4$ in the case where $\xi=0$. Note that the UV fixed point, depicted in purple, has a mass that is not allowed by the Breitenlohner-Freedman bound \cite{Breitenlohner:1982jf}.}
 \label{fig:scaling_xi=0}
\end{minipage}
\hspace{0.5cm}
\begin{minipage}[b]{0.45\linewidth}
      \centering
    \includegraphics[width=\textwidth]{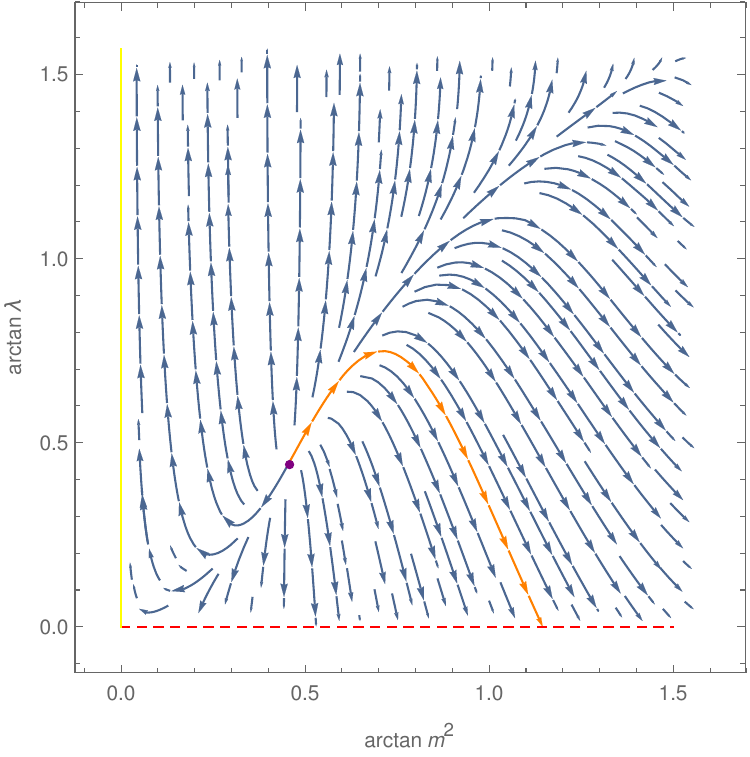}
    \caption{ Scaling of the coupling constants on $\npads_4$ in the case where $\xi=3/16$. Note that the negative values for the mass squared, to the left of the vertical yellow line, are not allowed.}
    \label{fig:scaling9-48}
\end{minipage}
\end{figure}

\begin{remark}
We briefly comment about the possibility of studying the scaling of the coupling constants with the prescriptions of boundary conditions different from those of Dirichlet type. Regarding Robin boundary conditions, a closed form for the 2-point correlation function, and consequently of the Hadamard parametrix, is not known, preventing the possibility of employing the methods that we have used in this Section. 

If Neumann boundary conditions are imposed at $\partial\npads_4$, even if there is a closed form expression for the 2-point correlation and for the Hadamard parametrix, another obstruction arises. Indeed the Neumann boundary conditions can be imposed for $\nu\in [0,1)$, see \cite{dappia_ads}. Having fixed the values of $\widetilde{m}^2,\widetilde{\lambda}$ at an energy scale $\Lambda$, we may apply the same analysis outlined in this Section, writing a system analogous to the one in Equation \eqref{eq:betafunct_pads} for the scaling of the coupling constants. Yet, we cannot expect that $\nu_k$ will continue to lie in $[0,1)$ for every $k$, causing an ambiguity on which class of boundary conditions should be applied at $\partial\npads_4$.
\end{remark}

\section*{Acknowledgments}
We are grateful to Edoardo d'Angelo and to Nicolò Drago for useful discussions especially on the content of \cite{mother}. The work of L.S. is supported by a PhD fellowship of the University of Pavia and part of this work is based on the MSc thesis of F.N. In addition C.D and L.S. acknowledge the support of the GNFM of INDAM.


\end{document}